\documentclass[12pt]{article}

\usepackage{amsmath,amsthm, amssymb, amsfonts}
\usepackage[ruled,vlined]{algorithm2e}
\usepackage{afterpage}
\usepackage{babel}
\usepackage{booktabs}
\usepackage{color,soul}
\usepackage{caption}
\usepackage{enumitem}
\usepackage[T1]{fontenc}
\usepackage{graphicx}
\usepackage{geometry}

\usepackage[hidelinks]{hyperref}
\usepackage{cleveref} 
\usepackage{mathrsfs}
\usepackage{MnSymbol, multirow}
\usepackage[authoryear]{natbib}

\usepackage{setspace}
 \usepackage{subcaption}
\usepackage{times}
\usepackage{todonotes}

\usepackage[modulo]{lineno}
\usepackage{tikz}
\usetikzlibrary{arrows,shapes.arrows,shapes.geometric,shapes.multipart,backgrounds,decorations.pathmorphing,positioning,fit,automata}
\tikzset{
	>=stealth',
	true/.style={
		rectangle,
		draw=black, very thick,
		text width=6.5em,
		minimum height=2em,
		text centered,
		fill=gray, opacity = 0.5},
	punkt/.style={
		rectangle,
		rounded corners,
		draw=black, very thick,
		text width=6.5em,
		minimum height=2em,
		text centered},
	est/.style={
		circle,
		draw=black, very thick,
		text centered},
	shade/.style={
		circle,
		draw=black, very thick, fill=gray!50,
		text centered},
	weight/.style={
		circle,
		draw=black, very thick,
		text width=6.5em,
		minimum height=2em,
		text centered},
	pil/.style={
		->,
		thick,
		shorten <=2pt,
		shorten >=2pt,},
	double/.style={
		<->,
		thick,
		shorten <=2pt,
		shorten >=2pt,},
	dash/.style={
		dashed,
		thick,
		shorten <=2pt,
		shorten >=2pt,},
	dashdouble/.style={
		<->,Co
		dashed,
		thick,
		shorten <=2pt,
		shorten >=2pt,}
}

\allowdisplaybreaks

\makeatletter
\@ifundefined{date}{}{\date{}}

\newtheorem{assumption}{Assumption}
\newtheorem{theorem}{Theorem}

\newtheorem{proposition}{Proposition}
\newtheorem{definition}{Definition}
\newtheorem{lemma}{Lemma}

\newtheorem{remark}{Remark}
\newtheorem{example}{Example}
\newtheorem{interpretation}{Interpretation}
\newtheorem{claim}{Claim}

\newtheorem{innercustomgeneric}{\customgenericname}
\providecommand{\customgenericname}{}
\newcommand{\newcustomtheorem}[2]{%
	\newenvironment{#1}[1]
	{%
		\renewcommand\customgenericname{#2}%
		\renewcommand\theinnercustomgeneric{##1}%
		\innercustomgeneric
	}
	{\endinnercustomgeneric}
}

\newcustomtheorem{myassumption}{Assumption}

\newcommand*{\nindep}{%
	\mathbin{
		\mathpalette{\@indep}{\not}
	}%
}
\newcommand*{\@indep}[2]{%
	\sbox0{$#1\perp\m@th$}
	\sbox2{$#1=$}
	\sbox4{$#1\vcenter{}$}
	\rlap{\copy0}
	\dimen@=\dimexpr\ht2-\ht4-.2pt\relax
	\kern\dimen@
	{#2}%
	\kern\dimen@
	\copy0 
} 

\definecolor{mygreen}{RGB}{144,241,47}

\newcommand{\pr}{\mathrm{pr}} 
\newcommand{\mPn}{\mathbb{P}_n}
\newcommand{\expect}{\mathbb{E}}
\newcommand{\splavg}{\mathbb{P}_n}
\newcommand{\var}{\mathrm{var}}

\newcommand{\mI}{\mathcal{I}}

\newcommand{\mR}{\mathbb{R}}

\newcommand{\Y}{Y}

\newcommand{\y}{\upsilon}

\newcommand*{\FM}{%
	{\mathrm{E}\mkern-9mu\raisebox{0.26ex}{$\circ$}}
}
\renewcommand{\mid}{\,|\,}
\newcommand{\ws}{{\mathcal{W}}_2(\mathcal{I})}

\def\rlog{\mathcal{L}}
\def\mp{\lambda}
\def\wwlog{\mathcal{L}}

\def\manifold{\mathcal{M}}
\def\hilbert{\mathcal{H}}
\def\real{\mathbb{R}}

\def\define{:=}
\def\pt{\tau}
\def\eY{\widehat{\Y}}

\def\op{o_P}
\def\Op{O_P}
\def\tdomain{\mathcal{I}}
\def\ldomain{\mathcal{J}}
\def\xdomain{\mathcal{X}}
\def\diffop{\mathrm{d}}

\def\idf{\mathrm{id}}
\def\ACE{\mathrm{ACE}}

\def\QTE{\mathrm{QTE}}

\newcommand{\ind}{\perp \!\!\! \perp }

\newcommand{\vertiii}[1]{{\left\vert\kern-0.25ex\left\vert\kern-0.25ex\left\vert #1 
		\right\vert\kern-0.25ex\right\vert\kern-0.25ex\right\vert}}

\makeatother

\makeatletter
\newcommand*{\addFileDependency}[1]{
	\typeout{(#1)}
	\@addtofilelist{#1}
	\IfFileExists{#1}{}{\typeout{No file #1.}}
}
\makeatother

\newcommand*{\myexternaldocument}[1]{%
	\externaldocument{#1}%
	\addFileDependency{#1.tex}%
	\addFileDependency{#1.aux}%
}

\usepackage{xr}
\myexternaldocument{supp}

\usepackage{authblk}
\author[]{Zhenhua Lin$^{1}$}
\author[]{Dehan Kong$^{2}$}
\author[]{Linbo Wang$^{2}$} 
\affil[]{$^{1}$ National University of Singapore, Singapore\\ $^{2}$ University of Toronto, Toronto, Ontario, Canada}

\begin{document}
	
	\title{\textbf{\Large{Causal Inference on Distribution Functions}}}
	

	\maketitle
	
	\begin{abstract}
		
		Understanding causal relationships is one of the most important goals of modern science. So far, the causal inference literature has focused almost exclusively on outcomes coming from the Euclidean space $\mathbb{R}^p$. However, it is increasingly common that complex datasets are best summarized  as data points in non-linear spaces. In this paper, we present a novel framework of causal effects for outcomes from the Wasserstein space of cumulative distribution functions, which in contrast to the Euclidean space, is non-linear. We develop doubly robust estimators and associated asymptotic theory for these causal effects. As an illustration, we use our framework to quantify the causal effect of marriage on physical activity patterns using wearable device data collected through the National Health and Nutrition Examination Survey. 
		
	\end{abstract}
	
	\noindent%
	\vfill
	\newpage
	
	\section{Introduction}
	\label{sec:introduction}
	
	Causal inference has received increasing attention in contemporary  data analysis. So far, researchers in causal inference have engaged almost exclusively in studying causal effects on objects from a linear space, most commonly the Euclidean space $\mathbb{R}^p.$  On the other hand, in many modern applications,
the observed data either naturally emerge or may be summarized as distribution functions. 
Often in these applications, the interest lies in the causal effect on the distributions themselves, rather than a summary measure such as the mean or the quantiles.  Here we detail an example from the study of physical activities; see the Supplementary Material for  additional examples on cellular differentiation and metagenomics.

	\begin{example}[Physical Activities]\label{ex:pa} 
Behavioral scientists are often interested in evaluating the effects of potential risk factors, such as marriage, on physical activity patterns \citep[e.g.][]{king1998effects}. 
The physical activity patterns are often recorded over a certain monitoring period. For example, in the National Health and Nutrition Examination Survey 2005--2006,
physical activity intensity, ranging from 0 to 32767 counts per minute, was recorded consecutively for  7 days by a wearable device for subjects at least six years old. 
The trajectory of activity intensity
is not directly comparable across different subjects as different individuals might have different circadian rhythms. Instead,  the distribution of activity intensity is invariant to circadian rhythms and hence can be compared between groups of individuals \citep{Chang2020}. 
	\end{example}


On the surface, since the distribution functions belong to $L^2$, a linear space endowed with a Euclidean distance, one may directly extend the classical potential outcome framework \citep{neyman1923application,rubin1974estimating} in causal inference using Euclidean averages. For example, the mean potential distribution functions may be defined as the expectation of the random potential distribution function, and the causal contrast among different interventions may be defined as the Euclidean distance among the mean potential distribution functions. However, there has been a growing recognition among the statistics and data science community that in many applications, the structure of data summarized by distribution functions may be best captured by the use of non-Euclidean distances \citep[e.g.][]{del1999tests,courty2016optimal,arjovsky2017wasserstein,ho2017multilevel,bernton2019approximate,verdinelli2019hybrid,panaretos2019statistical}. One of the most common choices of distance is the so-called Wasserstein distance based on the geometry of optimal transport. When equipped with such a distance, the space of probabilities measures on a real interval $\mathcal{I}$ is referred to as the Wasserstein space, and the average of distribution functions under the Wasserstein distance is known as the Wasserstein barycentre of these distributions.

Significant advances have been made by the emerging field of statistical optimal transport studying the Wasserstein space. For instance, \cite{Agueh2011,Bigot2012,Kim2017} introduced the notion of Wasserstein barycentre of a random distribution and studied its existence, uniqueness and characteristics. The concept of Wasserstein barycentre, defined in \eqref{eq:FM}, is a generalization of the mean of random variables/vectors to random distributions. Moreover, \cite{Bigot2017} generalized the technique of principal component analysis to data sampled from a Wasserstein space, while \cite{Petersen2016, chen2021wasserstein} focused on regression models for such data. Recently, \cite{zhang2020wasserstein,Zhu2021} investigated distribution-valued time series and generalized the concept of autoregressive models to Wasserstein spaces, while \cite{Zhou2021} developed a framework for canonical correlation analysis of random distributions. For more related works on statistical data analysis in a Wasserstein space, we refer readers to the survey paper by \cite{bigot2020statistical} and references therein.


The Wasserstein barycentre and distance have several features that make them particularly appealing for defining causal effects on distribution functions. 
First, the Wasserstein barycentre reduces to the usual Euclidean mean in the degenerate case where the distribution functions take point mass at real values. Specifically, if $\delta_t$ denotes the Dirac delta function with point mass at $t$, then the Wasserstein barycentre of $\delta_{t_i}, i=1,\ldots, n$ is $\delta_{\bar{t}},$ with $\bar{t} = \sum\limits_{i=1}^n t_i/n.$ In contrast, the Euclidean average, defined as the average of the cumulative distribution functions of $\delta_{t_1},\ldots,\delta_{t_n}$, corresponds to a uniform distribution on $t_i, \ldots, t_n$. Consequently, in the degenerate case where the random distribution function takes point mass at a random real value, their expectation does not correspond to the point mass function at the expectation of the random real value.

Second, compared to the Euclidean distance and other distances such as the Hellinger distance, the Wasserstein barycentre performs exceptionally well in capturing the structure of random distributions \citep[e.g.][]{cuturi2014fast}.
In Figure \ref{fig:wasserstein}, we provide a graphical illustration. The  distributions in Figure \ref{fig:wasserstein}(a) are unimodal distributions;  typical distributions of this kind are adult age-at-death distributions \citep{chen2021wasserstein}. One can see from Figure \ref{fig:wasserstein}(b) and (c) that the Wasserstein barycentre preserves unimodality while the Euclidean average does not.



Third,
the Wasserstein distance has an intuitive interpretation of the amount of ``work'' required to transform one distribution to another \citep{sommerfeld2018inference}. Moreover, it comes with a map that shows how to move from one random object to another, and allows one to  create a path of distributions that interpolates between Wasserstein barycentres, while  preserving the structural information in the random objects. This is particularly useful when comparing two Wasserstein barycentres each representing potential distribution functions under a particular intervention, as it not only contains information on how much they differ, but also on how one of these barycentres can be moved to the other; see Interpretation \ref{interpretation:transport} in Section \ref{sec:W} for more details.  In fact, as the path implied by the Wasserstein distance to move distributions involves the ``least effort,'' it is the natural path taken by biological systems to move from one state to another \citep[e.g.][]{Schiebinger2019}. 

Lastly, in the Wasserstein space of distributions on an interval of the real line, which is the focus of this paper, the optimal way to transfer one distribution to another is to move between their corresponding quantiles. In this case, 
the causal effect map defined using optimal transport (see Definition \ref{definition:causal}) can be directly interpreted as the difference in quantiles of the Wasserstein barycentres of different potential outcome distributions; see Interpretation  \ref{interpretaion:quantiles} in Section \ref{sec:W} for more details.

	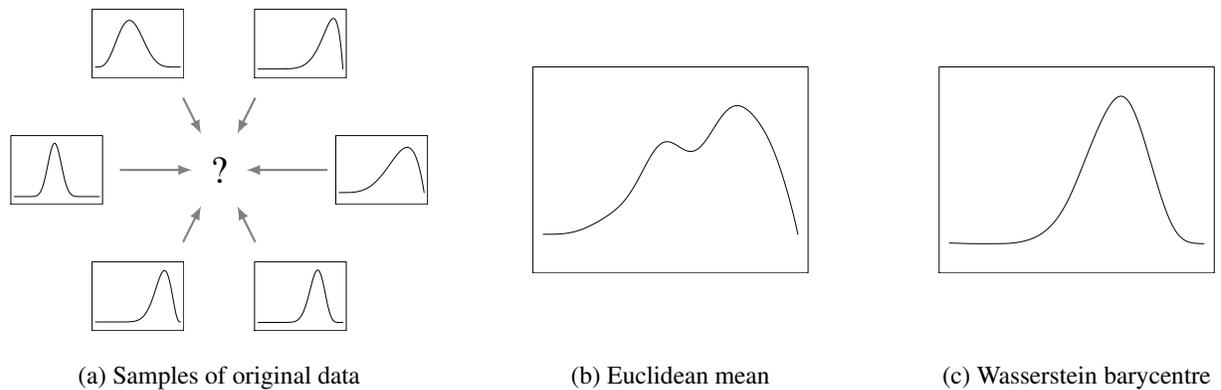
\begin{figure*}[t!]
		\begin{center}
			\begin{tikzpicture}[scale=1.2, every node/.style={scale=1.2}]
				\newcommand\x{0.8}
				\newcommand\dx{5}
				
				\node at (-0.3,0) {\includegraphics[scale=0.1]{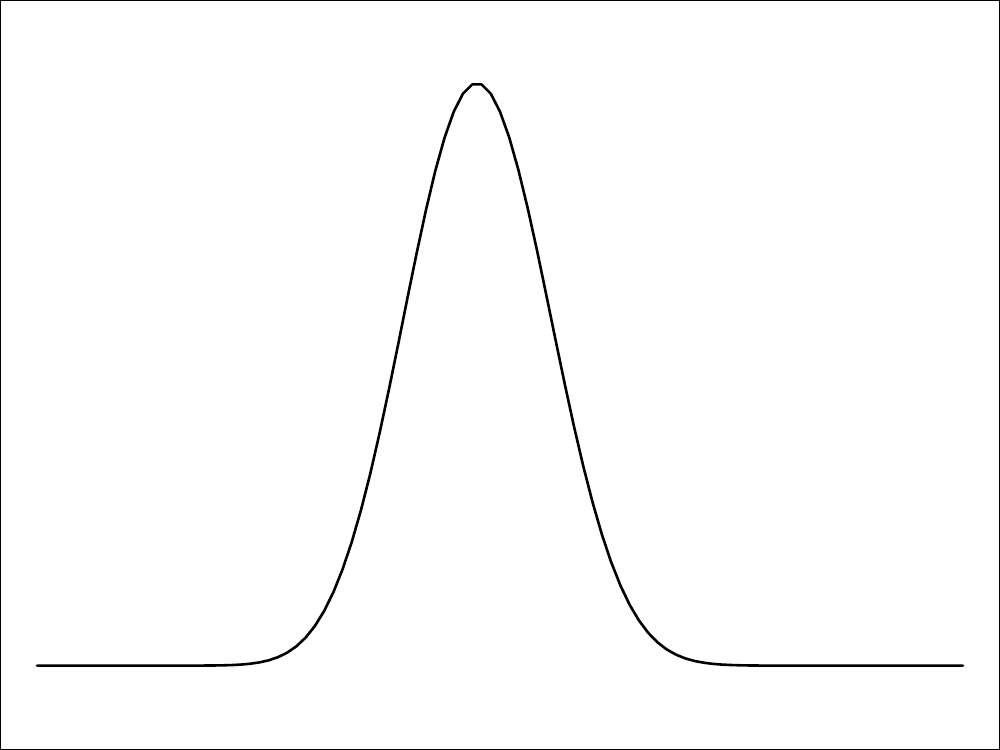}};
				\node at (3.3,0) {\includegraphics[scale=0.1]{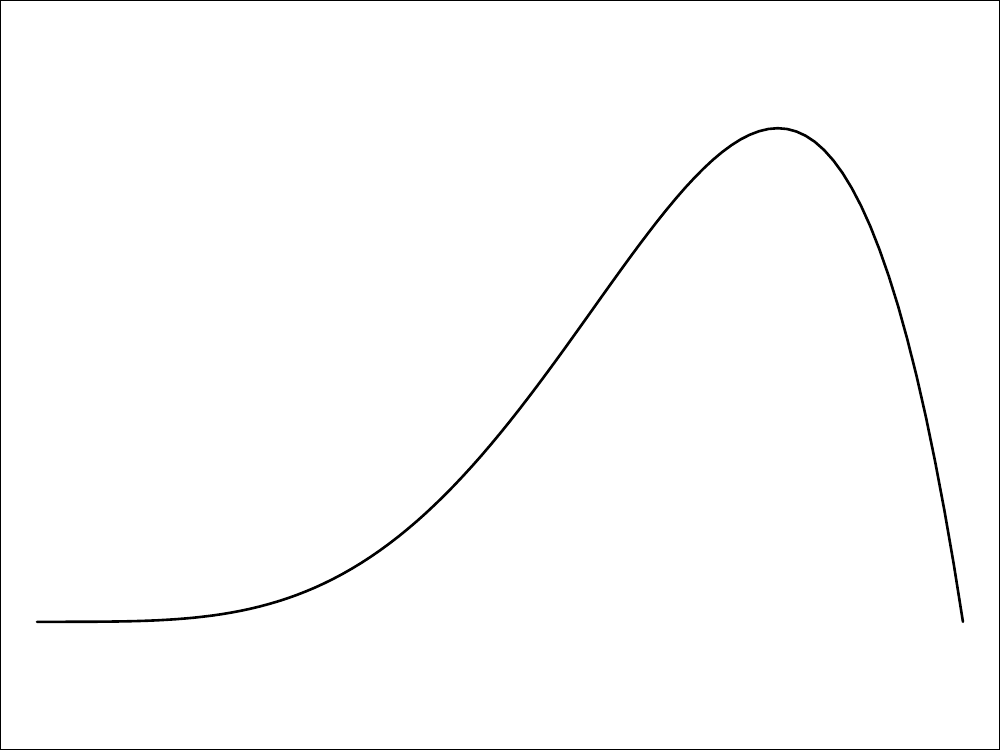}};
				\node at (0.6,1.4) {\includegraphics[scale=0.1]{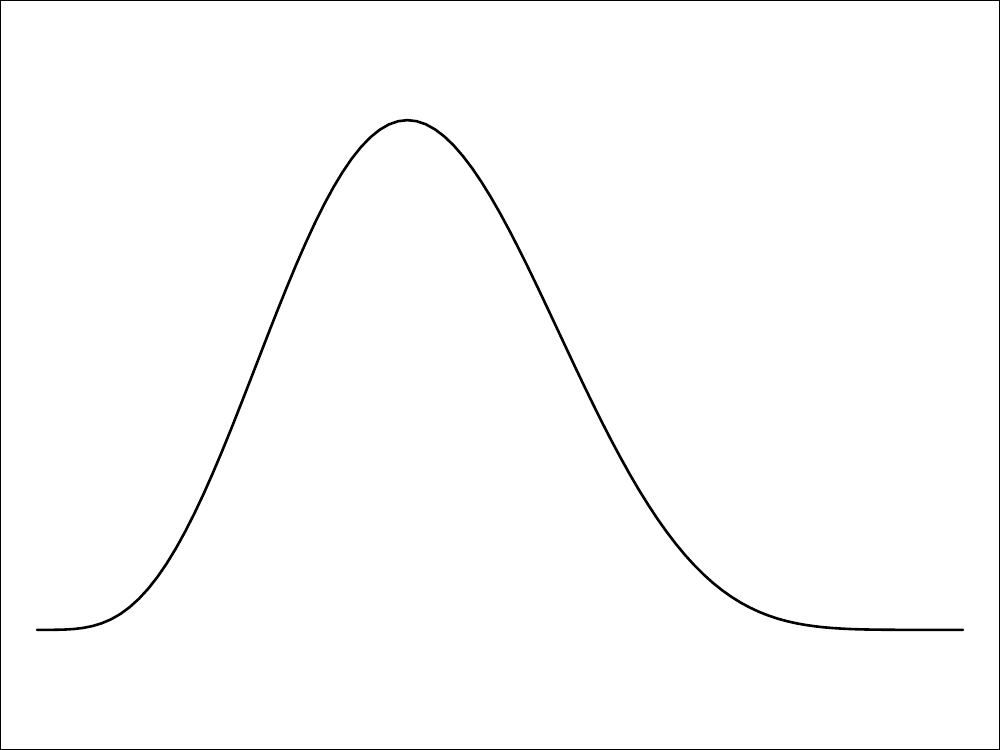}};
				\node at (2.4,1.4) {\includegraphics[scale=0.1]{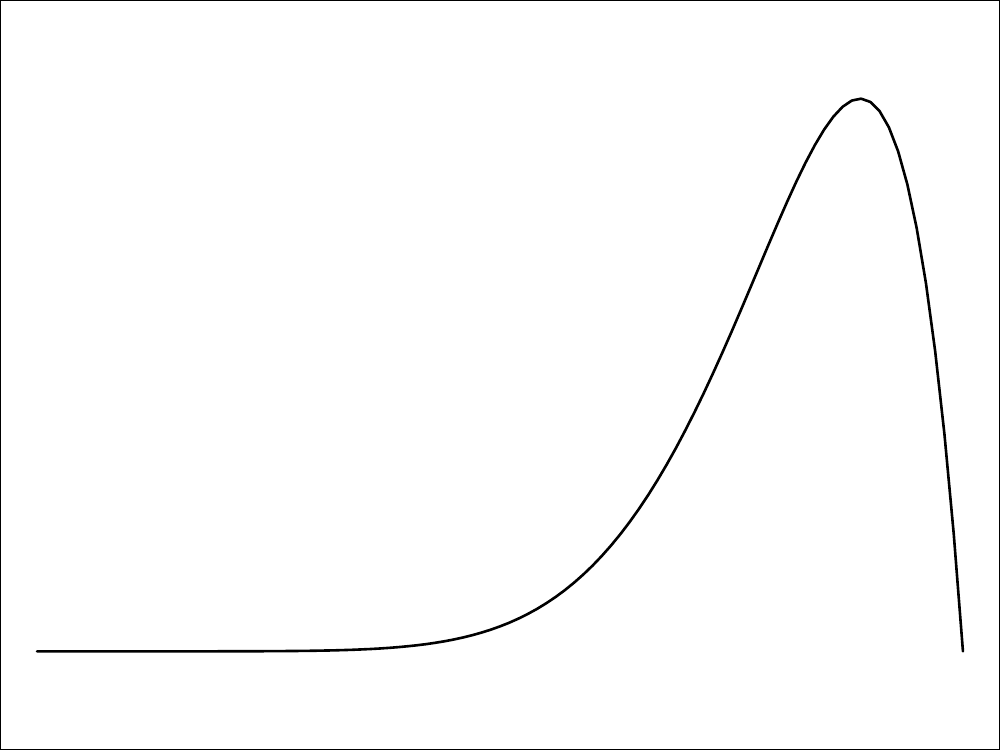}};
				\node at (0.6,-1.4) {\includegraphics[scale=0.1]{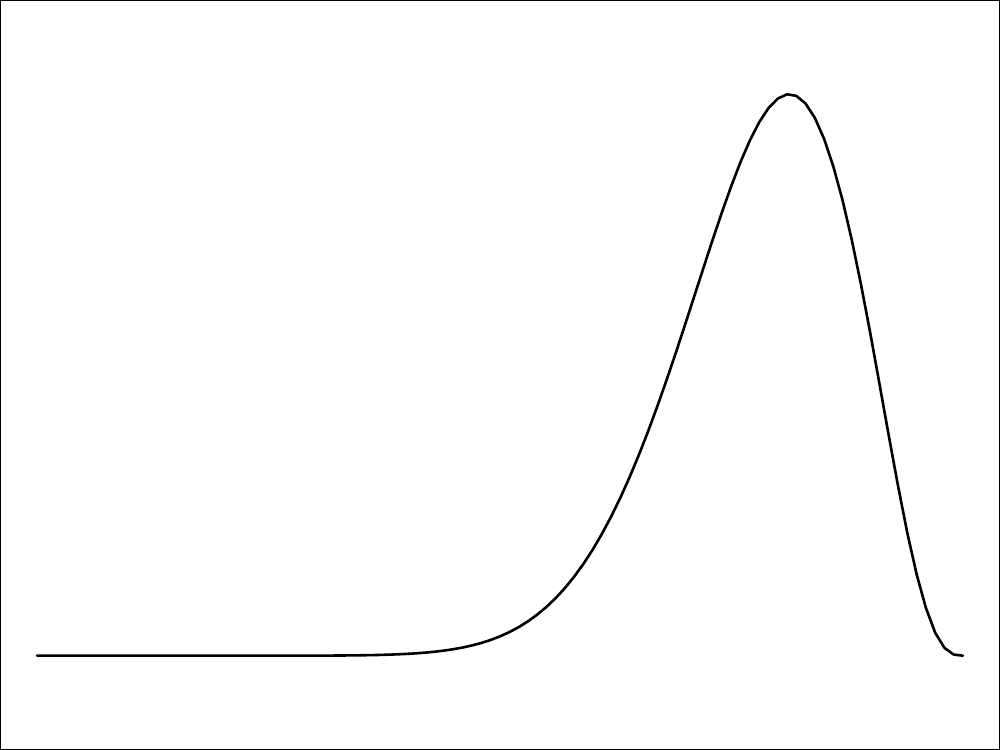}};
				\node at (2.4,-1.4) {\includegraphics[scale=0.1]{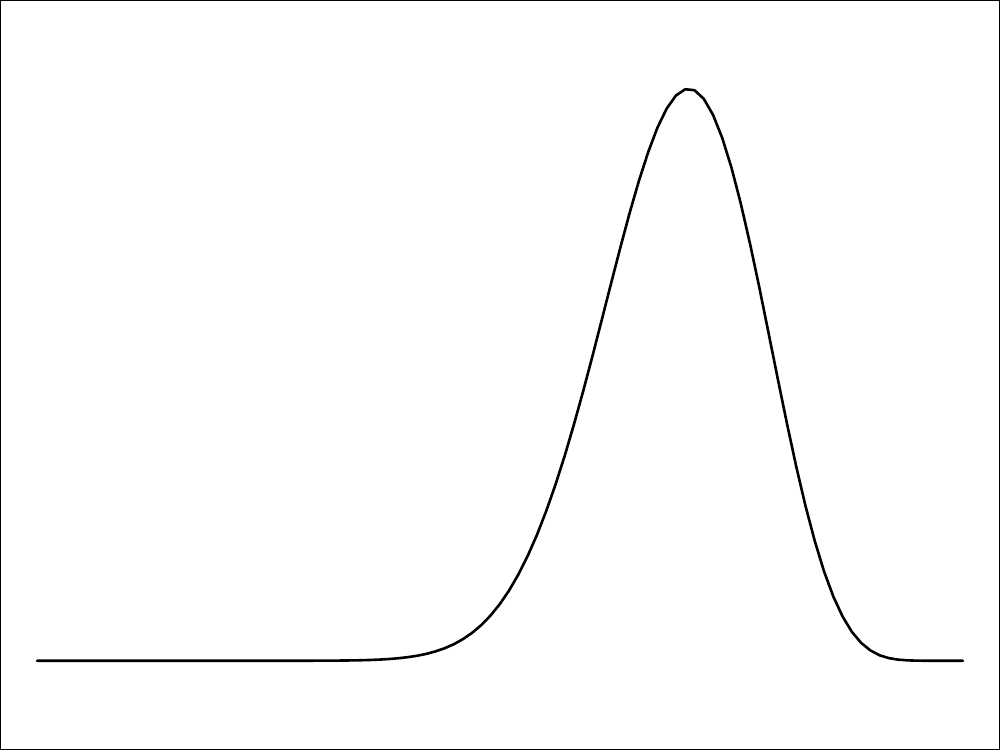}};
				
				\node at (1.5,0) {?};
				
				\draw[-latex,gray, thick] (0.4,0) -- (1.2,0); 
				\draw[-latex,gray, thick] (2.7,0) -- (1.8,0); 
				\draw[-latex,gray, thick] (1.1,0.8) -- (1.3,0.4); 
				\draw[-latex,gray, thick] (1.9,0.8) -- (1.7,0.4); 
				\draw[-latex,gray, thick] (1.9, -0.8) -- (1.7,-0.4);
				\draw[-latex,gray, thick] (1.1, -0.8) -- (1.3, -0.4);

				\node at (6.5,0) {\includegraphics[scale=0.3]{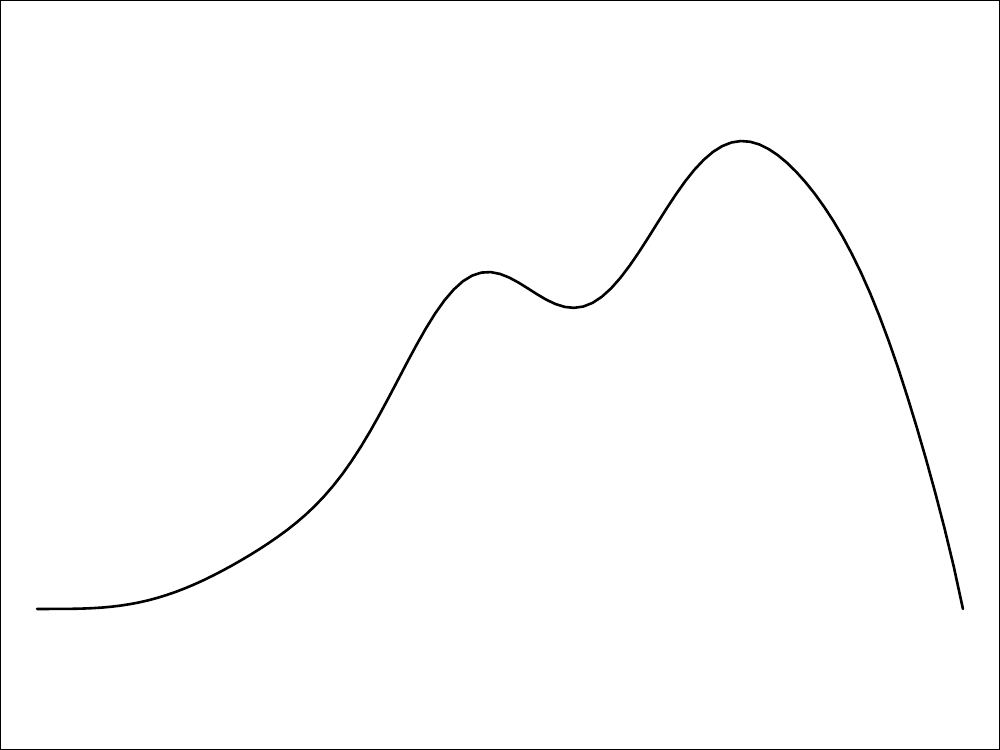}};
				\node at (11,0) {\includegraphics[scale=0.3]{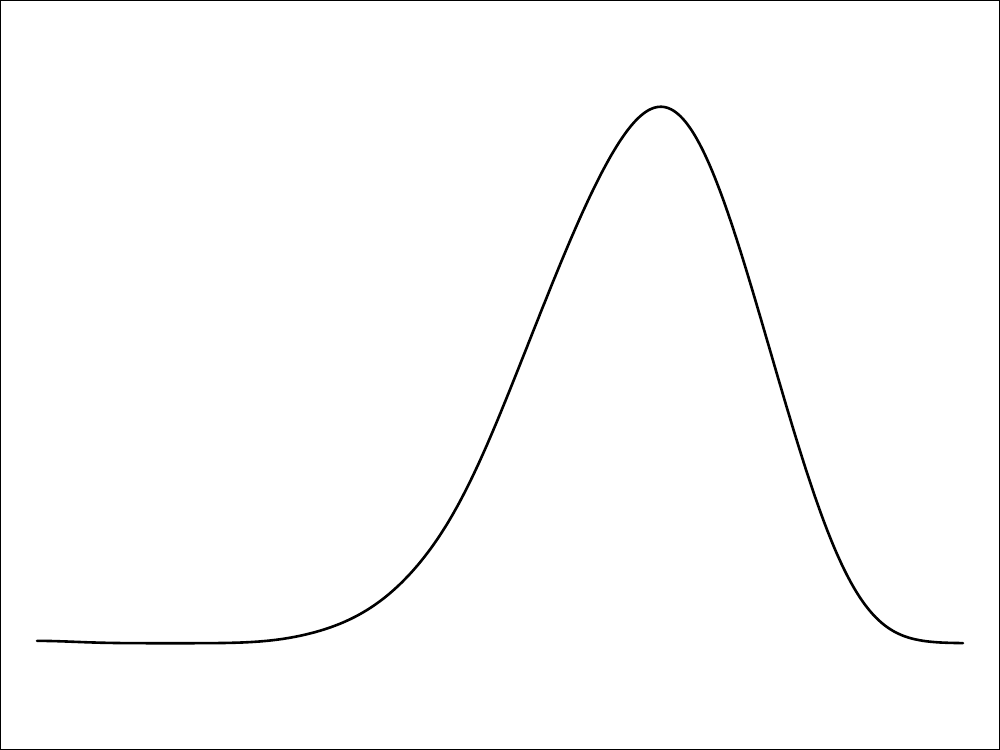}};

				\node at (1.5,-2.3) {\scriptsize (a) Samples of original data};
				\node at (6.5,-2.3) {\scriptsize (b) Euclidean mean};
				\node at (11,-2.3) {\scriptsize (c) Wasserstein barycentre};
			\end{tikzpicture} 
		\end{center}
		\caption{Comparison of Wasserstein barycentre versus Euclidean average for capturing structural information in random distributions. The barycentre and Euclidean average are based on six unimodal distributions. The Euclidean average has two modes while the Wasserstein barycentre preserves unimodality.}
		\label{fig:wasserstein}
	\end{figure*}

For these advantages and motivated by the applications in Example \ref{ex:pa} and Examples \ref{ex:mp}, \ref{ex:mg} in the Supplementary Material, we propose to define causal effects for distribution functions using the optimal transport between the Wasserstein barycentres of different potential outcome distributions. Our definitions of causal effect, called the average causal effect map, can be identified under straightforward generalizations of standard assumptions in the causal inference literature. We develop doubly robust and cross-fitting procedures for estimating the average causal effect map, and 
establish asymptotic properties for these estimators.
  In contrast to the setting for the classical doubly robust estimators \citep{robins1994estimation,chernozhukov2018double}, 
	typically even for a fixed unit, the outcome may not be fully observed and needs to be estimated from data. For instance, in our real data application, the data consists of empirical distribution functions of physical activity patterns for individuals rather than their underlying distribution functions.  
	To establish the asymptotic properties with distribution-valued outcomes, our  analyses rely on several  geometric properties of the Wasserstein space, most notably the isometry between the Wasserstein space and the space of quantile functions.  To the best of our knowledge, this is the first systematic study on causal inference for distribution-valued outcomes.
	
	The rest of the paper is structured as follows. In Section \ref{sec:background}, we present background on causal inference and Wasserstein space. In Section \ref{sec:W} we introduce the notion of average causal effect map for outcomes from a Wasserstein space, and develop identifiability results and doubly robust estimators for the average causal effect map. We  then study their asymptotic properties in Section \ref{sec:theory}.
We provide numerical studies in Section \ref{sec:simulation} and a real data illustration in Section \ref{sec:app}. We end with a brief discussion in Section \ref{sec:discussion}.

	\section{Background}
	\label{sec:background}
	
	\subsection{The potential outcomes framework}
	\label{sec:background-causal}
	
	We shall define causal effects using the  potential outcomes framework. 
	Suppose that the treatment is $A\in\{0,1\}$, with $0$ and $1$ being  the labels for control and active treatments, respectively. We use $X$ to denote baseline covariates taking values in $\mathbb{R}^d$.
	For each
	level of treatment $a$, we assume there exists a potential outcome $\Y(a)$,
	representing the outcome had the subject, possibly contrary to the
	fact, been given treatment $a$.  Here  $ \Y(a) $ is a random object that resides in a possibly non-linear  space.  We make the stable unit treatment value assumption \citep[SUTVA,][]{rubin1980comment} so that the potential outcomes for any unit do not vary with the treatments assigned to other units, and, for each unit, there are no different versions of treatments that lead to different potential outcomes. Under this assumption, the observed outcome $\Y = 
	\Y(A),$ where $\Y(A) = \Y(1)$ if $A=1$ and $\Y(A)=\Y(0)$ if $A=0$.   We assume we observe $n$ independent samples from an infinite  super-population of $(A,X,\Y),$  denoted by $(A_i, X_i, \Y_i), i=1,\ldots, n.$

	When $\Y(a)$ resides in    $\mathbb{R}$, the  causal effect is commonly defined as the contrast between a summary measure of the potential outcome distributions. For example, the average causal effect is defined as the difference between the means of the potential outcome distributions:
	\begin{equation}
		\label{eqn:ace}
		\ACE = \expect \{\Y(1)\} - \expect\{\Y(0)\};
	\end{equation} 
	the quantile treatment effect is  defined as the difference between the quantiles of the potential outcome distributions
	\begin{equation}
		\label{eqn:qte}
		\QTE(\alpha) =  F^{-1}_{Y(1)} (\alpha) - F^{-1}_{Y(0)} (\alpha), \alpha \in [0,1],
	\end{equation}
	where $F_Z(z)=P(Z\leq z)$ is the cumulative distribution function (CDF) of the random variable $Z$ and 
	\begin{equation}\label{eq:quantile}
	    F^{-1}_Z(\alpha) = \inf\{z: F_Z(z) \geq \alpha\}
	\end{equation}
	 is the corresponding quantile function.
	Causal effects defined in this manner can be interpreted on the population level, as they concern contrasts between potential outcomes in two hypothetical populations.  These population-level interpretations concern the effect of introducing a particular treatment to a population and are most relevant to policy makers.

	There is, however, a subtle but important distinction between the individual-level interpretations of $\ACE$ and $\QTE(\alpha)$. Let $CE_i = Y_i(1) - Y_i(0)$ be the individual causal effect for unit $i$. Individual causal effects provide useful information for  individualized treatment decision-making and are most relevant to individual subjects.  Since $\ACE = \expect (CE_i)$, it can be interpreted as averages of individual causal effects; here the expectation is taken over units in the super-population. The individual-level interpretation of $\ACE$ extends to the conditional average treatment effect, $\mathrm{CACE}(L) = \expect\{Y(1)\mid L\} - \expect\{Y(0)\mid L\},$ where $L$ is a subset of observed baseline covariates.
	In contrast, generally, $\QTE(\alpha)$ cannot be interpreted as the $\alpha-$quantile of individual causal effects. {This distinction connects to desideratum (d)   in Section \ref{subsec:def-ce}.}

	\subsection{Causal effect identification and estimation}
	\label{background-identification}
	
	The following assumptions are standard in the causal inference literature \citep[e.g.][]{rosenbaum1983central,hernan2020}.
	\begin{assumption}[Ignorability]\label{assu:ignorability}
		$A\ind \Y(a)\mid X, a=0,1.$ 
	\end{assumption}
	\begin{assumption}[Positivity]\label{assu:positivity}
		The propensity score $\pi(X):=P(A=1\mid X)$ is bounded away from 0: There exists $\epsilon>0$, such that $\epsilon<\pi(X) <1-\epsilon, a.e.$
	\end{assumption}
	Under Assumptions \ref{assu:ignorability} and \ref{assu:positivity}, when $Y(a)$ resides in $\mathbb{R}$, the mean potential outcome  is given by 
	\begin{equation}
		\label{eqn:identify}
		\mu_a:=\expect[Y(a)] = \expect_X\{\expect[\Y\mid A=a, X]\} = \expect\left\{ \textstyle\dfrac{I(A=a)Y}{P(A=a\mid X)}  \right\}.
	\end{equation}
	Similarly, the potential outcome distributions $F_{Y(a)}(y)$  can be identified by replacing $\Y$ in the last term of eqn. \eqref{eqn:identify} with $I(Y \leq y).$ Based on \eqref{eqn:identify}, the average causal effect can be identified as $\ACE = \expect\{Y(1)\} -\expect\{Y(0)\}$, and the quantile treatment effect can be identified as $\QTE(\alpha) = F_{Y(1)}^{-1}(\alpha) - F_{Y(0)}^{-1}(\alpha).$  
	
	Given the identification formula \eqref{eqn:identify}, one may use plug-in estimators to estimate the mean potential outcomes. Let ${m}_a(X) = \expect(\Y\mid A=a, X)$ and $f(A \mid X) = A\pi(X) + (1-A) (1-\pi(X)).$  Also denote $\widehat{m}_a(X), \widehat{\pi}(X), \widehat{f}(A \mid X)$ as  estimates of their corresponding population quantities obtained using standard parametric or nonparametric/machine-learning techniques. 
	Some leading estimators of $\mu_a$ include the outcome regression estimator $\widehat{\mu}_a^{OR} = \mathbb{P}_n \widehat{m}_a(X)$,  the inverse probability weighting estimator $\widehat{\mu}_a^{IPW} = \mathbb{P}_n \textstyle\frac{I(A=a)Y}{\widehat{f}(A \mid X)},$ and the so-called doubly robust estimator $\widehat{\mu}_a^{DR} = \widehat{\mu}_a^{OR} + \mathbb{P}_n  \left[\textstyle\frac{I(A=a)}{\widehat{f}(A \mid X)} \left\{Y-\widehat{m}_a(X)\right\}\right];$ here $\mathbb{P}_n$ refers to the empirical average operator: $\mathbb{P}_n(O) = \dfrac{1}{n}\sum\limits_{i=1}^n O_i.$

	\subsection{Wasserstein space}
	
	Let $\tdomain$ be an interval of $\real$,  $V_1$ and $V_2$ be random variables taking values in $\mI$ with finite second moments, and $\lambda_1, \lambda_2$ be their (cumulative) distribution functions, respectively.  To define the Wasserstein distance between $\lambda_1$ and $\lambda_2$, we let $\Lambda(\lambda_1,\lambda_2)$ denote all joint distributions  $\lambda_{12}$ of $(V_1, V_2)$ that have marginal distributions $\lambda_1$ and $\lambda_2$. The (2-)Wasserstein distance between $\lambda_1$ and $\lambda_2$ is defined as
	\begin{equation}
		\label{eqn:w2}
		W_2(\lambda_1,\lambda_2) = \left( \inf\limits_{\lambda_{12}\in \Lambda(\lambda_1,\lambda_2)} \int_{\tdomain\times\tdomain} (s-t)^2 d\lambda_{12}(s,t)  \right)^{1/2}.
	\end{equation}
	The Wasserstein space of order 2 on $\tdomain$ is then defined as the space of  distribution functions on $\mI$ with finite second moments
	$$
	\ws=\bigg\{\lambda 
	\text{ is a distribution function on } \mI:\int_{\tdomain} t^2\diffop \lambda(t)<\infty\bigg\}
	$$
	endowed with the 2-Wasserstein distance.

	The Wasserstein distance can be motivated by the problem of optimal transport. Consider a pile of mass  on space $\mI$ with a distribution $\lambda_1$. We wish to transport the mass in such a way that the new mass distribution  is $\lambda_2$. Assume also that the cost of transporting a unit mass from point $s$ to point $t$ is  $(s-t)^2.$ A transport plan to move $\lambda_1$ to $\lambda_2$ can be described by the function $\lambda_{12}$ such that  $d\lambda_{12}(s,t)$ denotes the amount of mass to move from $s$ to $t$. Since the amount of mass to be moved out of $s$ must match $d\lambda_1(s)$, and the amount of mass to be moved into $t$ must match $d\lambda_2(t),$ we have $
	\int_{t\in \mI} d\lambda_{12}(s,t)  = d\lambda_1(s), \int_{s\in \mI} d\lambda_{12}(s,t)  = d\lambda_2(t).
	$ In other words, $\lambda_{12} \in \Lambda(\lambda_1, \lambda_2)$. The Wasserstein distance then corresponds to the minimum effort that is required in order to transport the mass of $\lambda_1$ to produce the mass distribution of $\lambda_2$. {The minimizer $\lambda_{12}^*$ to the problem in \eqref{eqn:w2} always exists \citep[][Theorems 1.7 and 1.22]{Santambrogio2015} and is known as the optimal transport plan. 
		
		If $\lambda_1$ is continuous, then there exists a unique function $T(\cdot): \mathcal{I } \rightarrow \mI$ such that $d\lambda_{12}^*(s, T(s)) = d\lambda_1(s)$ \citep[][Theorems 1.7 \& 1.22]{Santambrogio2015}. Intuitively,  in this case, the optimal transport plan moves all the mass at $s$ to $T(s)$. The function $T(\cdot)$ is known as the optimal transport map.} Let $\lambda^{-1}$ be the quantile function of the distribution  $\lambda$.
	It can be shown   that  
	\citep[][Theorem 6.0.2]{Ambrosio2005}
	the optimal transport map  {$T(s) = \lambda_2^{-1}(\lambda_1(s))$, so that it moves mass between corresponding quantiles of $\lambda_1$ and $\lambda_2.$} The following proposition, summarizing the above discussion, shows that given a fixed continuous distribution $\lambda_1$, the distribution $\lambda_2$ can be defined via the optimal transport map from $\lambda_1$ to $\lambda_2.$ The proof of Proposition \ref{prop:unique} is straightforward and hence omitted.
	
	\begin{proposition}
		\label{prop:unique}
		Given a continuous distribution  function $\lambda_1$, there is a one-to-one correspondence between a distribution function $\lambda_2$ and the optimal transport map from $\lambda_1$ to $\lambda_2$.
	\end{proposition}

	Based on the notion of Wasserstein distance $W_2$, we can define the mean of a set of distributions $\lambda_1, \ldots, \lambda_n$ in the  Wasserstein space  by the so-called Wasserstein barycentre $\bar{\lambda}$, defined as the distribution ${\lambda}$ that minimizes $\dfrac{1}{n}\sum\limits_{i=1}^n W_2^2(\lambda_i, {\lambda})$. {In Lemma \ref{lem:expect-Y-inverse-inveser-mu} in the Supplementary Material, it is shown that $\bar{\lambda}^{-1} = \dfrac{1}{n} \sum\limits_{i=1}^n \lambda_i^{-1},$ so that the quantile function corresponding to the Wasserstein barycentre equals the (Euclidean) averages of the individual quantile functions.}
	The mean of distributions can also be defined in various other ways, such as the mean under the Euclidean distance $\dfrac{1}{n} \sum\limits_{i=1}^n \lambda_i$. 
	Compared to alternative center measures, the Wasserstein barycentre typically provides a better summary that captures the structure of the random objects represented by distribution functions, such as shapes, curves, and images; see for example, Figure \ref{fig:wasserstein} and \citet[][Figure 1]{cuturi2014fast} for illustrations.

	\section{Causal Inference on Distribution Functions}\label{sec:W}
	
	\subsection{Definition of causal effects}\label{subsec:def-ce}
	
	We now introduce a definition of average causal effect for outcomes taking value in the   Wasserstein space $\ws$.  
	In parallel to the definition of causal effects introduced in Section \ref{sec:background-causal}, we first define the mean potential outcomes.  As both $Y_i(1)$ and $Y_i(0)$ take value  in the  Wasserstein space, 
	we define their means using their Wasserstein barycentres:
	\begin{equation}\label{eq:FM} 
		\mu_a = \FM \Y(a) \equiv\underset{\upsilon \in \ws}{\arg\min} \ \expect \left\{W_2^2\left(Y(a),\upsilon\right) \right\}, \quad a=0,1.
	\end{equation} 
	Intuitively, $\mu_a$ is a ``typical'' potential distribution under treatment $A=a.$ Let $\delta_{t}$ denote the Dirac delta function with point mass at $t$.  
	Ideally, a causal effect definition in the Wasserstein space should satisfy the following desiderata:
	\begin{enumerate}[label=\textup{(\alph*)}]
	\item \label{desi:1} When $\FM \Y(1) = \FM \Y(0),$ the causal effect  equals zero;
		\item\label{desi:2} In the degenerate case where $Y_i(a) = \delta_{y_i(a)}$, corresponding to the classical scenario where the outcome resides in $\real$, the causal effect corresponds to the usual average causal effect $\ACE$ defined in \eqref{eqn:ace}; 
		\item \label{desi:3} The average causal effect is a contrast between the averages of potential outcomes in two hypothetical populations, $Y(1)$ and $Y(0)$, and thus can be interpreted on the population level;
		\item \label{desi:4} The average causal effect equals the average of individual causal effects, thus maintaining the individual-level interpretation of the $\ACE$ for real-valued outcomes discussed at the end of Section \ref{sec:background-causal}.
	\end{enumerate}
Desideratum (a) is natural, given the causal effect is defined as a comparison between two (hypothetical) populations. Desideratum (b) ensures that the definition is a generalization of the standard definition of the ACE \eqref{eqn:ace} in the Euclidean space.  Desiderata (c) and (d) are in place to ensure that the definition can be interpreted at both  population and individual levels.

	From an optimal transport point of view, it may be  tempting to define the causal effect as the Wasserstein distance between $\mu_1$ and $\mu_0$; see Section \ref{sec:S1} in the Supplementary Material for more discussions on causal effect defined in this way. Although causal effect defined in this way satisfies desiderata (a)--(c), in general, it fails to satisfy desideratum (d).  Instead, we introduce a novel definition of average causal effect, called the \emph{causal effect map}.  In Section \ref{sec:properties}, we shall see that the causal effect map satisfies all the desiderata, and  contains richer information  than the single summary measure $W_2(\mu_1, \mu_0)$; in particular, one can compute $W_2(\mu_1, \mu_0)$ based on the causal effect map.
	
	\begin{definition}
		\label{definition:causal}
		Let $\lambda$ be a continuous distribution function. 
		The individual causal effect map of $A$ on $\Y$ is defined as 
		\begin{equation*}
			\Delta_i^\lambda(\cdot) = Y_i(1)^{-1} \circ \lambda(\cdot) - Y_i(0)^{-1} \circ \lambda(\cdot),
		\end{equation*}
		where we say $\lambda$ is a reference distribution; for $a=0,1,$ $ Y_i(a)^{-1}$ is the quantile function of the distribution $Y_i(a)$ as defined in \eqref{eq:quantile}. 
		The (average) causal effect map  of $A$ on $\Y$ is defined as
		\begin{equation*}
			\Delta^\lambda(\cdot) = (\FM \Y(1))^{-1} \circ \lambda (\cdot) - (\FM \Y(0))^{-1} \circ \lambda (\cdot) =  (\mu_1^{-1} - \mu_0^{-1}) \circ \lambda (\cdot).
		\end{equation*}
	\end{definition}

	A crucial component in Definition \ref{definition:causal} is the choice of reference distribution $\lambda$ {that is allowed to have a domain different from that of $Y$.} In general, a different reference distribution leads to a different interpretation of the causal effect maps. Hence, one should choose 	the reference distribution based on the desired interpretation. We now illustrate some common choices of $\lambda$ with their interpretations.

	\begin{interpretation}[Difference in quantiles] \label{interpretaion:quantiles}
	    If the reference distribution $\lambda$ is the uniform distribution on $[0,1]$ so that $\lambda(t) = t, t\in[0,1]$, then 
		the causal effect map $\Delta(\cdot)$ can  be interpreted as difference in quantiles.
	\end{interpretation}

	\begin{remark}
	\label{remark:diq}
	    The interpretation in terms of difference in quantiles is not to be confused with the quantile treatment effect defined in \eqref{eqn:qte}. In our setting, the potential outcomes are random distribution functions, and $\mu_1^{-1}(\alpha)$ is the $\alpha-$quantile of \emph{the mean potential outcome} under treatment. In contrast, in the quantile treatment effect setting, the realizations of potential outcomes are real numbers, and $F^{-1}_{Y(1)} (\alpha)$ is the $\alpha-$quantile of \emph{the distribution of potential outcomes} under treatment.
	\end{remark}

	To discuss the second interpretation, analogous to the concepts of optimal transport map, we define the \emph{individual causal transport map} as $
		T_i(\cdot) = Y_i(1)^{-1} \circ Y_i(0) (\cdot)
		$ and the \emph{(population) causal transport map} as
		$
		T(\cdot) =  \mu_1^{-1} \circ \mu_0 (\cdot).
		$
	The causal transport maps are of natural interest in some applications. For example,  biological experiments  \citep[][]{Schiebinger2019} have found that cellular differentiation follows the shortest path under the Wasserstein geometry.
So in  Example \ref{ex:mp}, the causal transport map $T$ describes how a group of cells would differentiate after being exposed to an intervention, measured using gene expression levels.
	
	When the potential outcomes $Y(1), Y(0)$, and hence the barycentres $\mu_1$ and $\mu_0$ \citep[e.g.][Proposition 4.1]{Bigot2017}, are continuous distributions, with certain choices of the reference distribution, the causal effect maps can be interpreted as the (inverse of) causal transport maps up to an identity function.

	\begin{interpretation}[Causal transport maps]
	\label{interpretation:transport}
		Consider the case where the potential outcomes are random continuous distribution functions. If the reference distribution is chosen to be the barycentres $\mu_0$ or $\mu_1$, then  
			$$
			\Delta^{\mu_0} = T - \idf  = \mu_1^{-1} \circ \mu_0  - \idf,  \quad  \Delta^{\mu_1} =  \idf - T^{-1} = \idf - \mu_{0}^{-1} \circ \mu_1.
			$$
		If the reference distribution is chosen to be $Y_i$, then 
		\begin{equation}
		    \label{eqn:ict}
		    \Delta_i^{Y_i} = \left\{
		    \begin{array}{cc}
		    T_i  -\idf = Y_i(1)^{-1} \circ Y_i(0) - \idf     &  \text{if } A_i = 0 \\
		      \idf - T_i^{-1} = \idf - Y_i(0)^{-1} \circ Y_i(1)   & \text{if } A_i = 1 
		    \end{array}\right..
		\end{equation}    
	\end{interpretation}

	\subsection{Properties of causal effect maps}
	\label{sec:properties}
	
	We now describe several desirable properties of the causal effect maps. 
	First, it is easy to verify that for any choice of reference distribution $\lambda,$ the causal effect map satisfies desiderata (a)--(c). 
	The following theorem, which is crucial for  identification of the average causal effect map as we shall see later in Section \ref{sec:identifiability}, shows that  the causal effect map also satisfies desideratum (d). This theorem can be proved using Lemma \ref{lem:expect-Y-inverse-inveser-mu} in the Supplementary Material.

	\begin{theorem}
		\label{prop:equal} The average causal effect map corresponds to the average of individual causal effect maps with respect to the same reference distribution $\lambda$:
		$$
		\Delta^\lambda(\cdot) = \expect \Delta_i^\lambda(\cdot) = \expect \left\{Y(1)^{-1} \circ \lambda(\cdot) - Y(0)^{-1} \circ \lambda(\cdot)  \right\}.
		$$
	\end{theorem}
	
	\begin{remark}
	\label{remark:individual}
	    In contrast to Theorem \ref{prop:equal},  the population causal transport map is generally different from the average of individual causal transport maps: $T \neq \expect [T_i]$.  To see this, note that as shown in Interpretation \ref{interpretation:transport}, $T = \Delta^{\mu_0} + \idf,$ while $T_i = \Delta_i^{Y_i(0)} + \idf.$  Although as we show in Theorem \ref{prop:equal}, under the same reference distribution $\lambda,$ $\Delta^\lambda = \expect{\Delta_i^\lambda},$ in general, $\Delta^{\mu_0} \neq \expect{\Delta_i^{Y_i(0)}}.$
	    Instead, as we illustrate later in Section \ref{sec:app}, to estimate the (expectation of) causal transport map for a particular individual $i$,  one first estimates the average causal effect map with reference distribution $Y_i$, and then applies \eqref{eqn:ict}.

	\end{remark}

	{In some scenarios, practitioners may also want a scalar quantity that measures the magnitude of the causal effect. The Wasserstein distance is a natural choice from an optimal transport point of view.}
	The following proposition shows that  one may compute the Wasserstein distance  $W_2(\mu_1, \mu_0)$ from the causal effect map $\Delta^\lambda(\cdot)$. In particular,  for any $U$ that follows the reference distribution $\lambda$,   $W_2(\mu_1, \mu_0)$ equals the $\ell_2$-norm of $\Delta^\lambda(U)$.  It can be proved using \Cref{lem:dist-norm-W} in the Supplementary Material.
	\begin{proposition}
		\label{prop:compute}
		The Wasserstein distance  $W_2(\mu_1, \mu_0)$ is determined by the causal effect map $\Delta^\lambda(\cdot)$ via
		$$
		W_2(\mu_1,\mu_0) = \|\Delta^\lambda\|_\lambda \define \left\{\expect_{U\sim \lambda} \ (\Delta^\lambda)^2(U)  \right\}^{1/2}=  \left\{\int (\Delta^\lambda)^2(u) d\lambda (u)\right\}^{1/2}.
		$$
	\end{proposition}
	
	\subsection{Identification and estimation}
	\label{sec:identifiability}

	Under the ignorability and positivity assumptions, similar to \eqref{eqn:identify}, we can identify the causal effect map $\Delta^\lambda.$
	
	\begin{theorem}
		\label{thm:identification}
		Under Assumptions \ref{assu:ignorability} and \ref{assu:positivity},   the average causal effect map $\Delta^\lambda$ is identifiable and given by 
		\begin{equation*}
			\Delta^\lambda=\mu_1^{-1,\lambda}-\mu_0^{-1,\lambda}
		\end{equation*}
	with $$\mu_a^{-1,\lambda} = \expect_X\{\expect[(\Y^{-1} \circ \lambda) \mid A=a, X]\} =   \expect\left\{ \textstyle\dfrac{I(A=a) (\Y^{-1} \circ \lambda)}{P(A=a\mid X)}  \right\}.$$
	\end{theorem}

	To see the above, note that
	\begin{align*}
		\Delta^\lambda &= \sum\limits_{a=0,1} (2a-1) \expect\left\{ Y(a)^{-1} \circ \lambda \right\}\quad \text{(due to Theorem \ref{prop:equal})} \\ 
		&=  \sum\limits_{a=0,1} (2a-1) \expect_X [\expect \{ Y(a)^{-1} \circ \lambda \mid A=a, X\}] \quad  \text{(due to ignorability)} \\
		&=  \sum\limits_{a=0,1} (2a-1) \expect_X [\expect \{Y^{-1} \circ \lambda\mid A=a, X\}] \quad  \text{(due to the SUTVA)}\\
		& = \mu_1^{-1,\lambda}-\mu_0^{-1,\lambda}.
	\end{align*} 
	
	\begin{remark}
	In general, when $Y(a)$ is a random distribution function,  $$
	    \mu_a \neq \expect_X\{\expect[\Y   \mid A=a, X]\} = \expect\left\{ \textstyle\dfrac{I(A=a) \Y}{P(A=a\mid X)}  \right\}.
	 $$
	So  $\mu_a, a=0,1$ may not be directly identified using the identification formula \eqref{eqn:identify} for real-valued $Y$. Instead, following Theorem \ref{thm:identification}, the mean potential outcomes $\mu_a, a=0,1$ may be identified as the inverse of $\mu_a^{-1,\lambda}  \circ \lambda^{-1}.$  
	\end{remark}
	
	\begin{remark}
	\label{remark:ignorability}
	    When the outcome is a distribution function, Assumption \ref{assu:ignorability} is stronger than the same assumption applied to a summary measure of the outcome.
	\end{remark}

	Similar to Section \ref{background-identification}, we let
	$
	m_a^\lambda(X)=\expect\{ \Y^{-1}\circ \lambda \mid A=a,X\},
	$
	so that 
	$\Delta^\lambda = \sum\limits_{a=0,1} (2a-1) \expect_X m_a^\lambda(X).$ We  consider a regression model $m_a^\lambda(X;\beta)$, where $\beta$ may be finite or infinite dimensional.
	In practice,  the outcome  $\Y_i$, which is a distribution function itself, might not be fully observed. Instead, we typically only observe $k_i$ samples from $\Y_i$. We hence propose to construct a doubly robust estimator for $\Delta^\lambda$ in the following steps: (i) Use standard nonparametric methods such as the nonparametric maximum likelihood estimation, or local polynomial smoothing to obtain the estimates $\widehat{Y_i}$; (ii) If  $\lambda$ is unknown, obtain an estimate of $\lambda$, denoted as $\widehat{\lambda}$; for notational convenience,  we let $\widehat{\lambda} = \lambda$ if $\lambda$ is a fixed or fully observed reference distribution; (iii) Regress $\widehat{Y}^{-1}\circ \widehat{\lambda}$ on $X$ 
	and $A$ to obtain an estimate of $m_a^{\lambda}(X)$ for each individual in the sample,
	denoted as $\widehat{m}_a^{\widehat\lambda}(X_i), i=1,\ldots,n;$  {this may be done  using a standard functional regression model \citep[see e.g.][Chapter 13]{Ramsay2005};}
(iv) Construct an estimate for $f(A\mid X),$ denoted as $\widehat{f}(A \mid X)$;
(v) Construct a doubly robust estimator via
	\begin{equation}\label{eq:dr}
	\widehat{\Delta}_{DR}^{\widehat{\lambda}} = \widehat\mu_1^{-1,\widehat\lambda}-\widehat\mu_0^{-1,\widehat\lambda}\quad \text{with}\quad \widehat\mu_a^{-1,\widehat\lambda} = \mathbb{P}_n \left[\widehat{m}_a^{\widehat\lambda}(X)+\dfrac{I(A=a) }{\widehat{f}(A \mid X)}\left\{\widehat{\Y}^{-1}\circ \widehat{\lambda} - \widehat{m}_a^{\widehat\lambda}(X)\right\}\right].
	\end{equation}
	 This estimator,
{motivated  by the doubly robust estimator discussed in Section \ref{background-identification}},  combines an outcome regression estimator with an inverse probability weighting estimator. We establish its asymptotic properties in Section \ref{sec:theory}.

	In the doubly robust estimating procedure described above, the data are used twice: once for estimating $\widehat\pi$, $\widehat f$ and $\widehat m_a^{\widehat\lambda}$, and once for estimating the causal effect $\widehat\Delta_{DR}^{\widehat\lambda}$. Theoretical analysis of such estimators is rather challenging and requires some complex conditions that may fail in settings involving machine learning methods; see Assumption \ref{assu:DR-fixed-ref} and Remark \ref{remark:donsker}.  To overcome this difficulty, following \cite{chernozhukov2018double}, we propose the following cross-fitting estimators.  The entire data are randomly partitioned into $K$ parts of roughly equal sizes, denoted by $\mathscr D_1,\ldots, \mathscr D_K.$  For $k=1,\ldots, K$, we use  $\mathscr D_{-k} = \cup_{m\neq k} \mathscr D_m$  to obtain estimates  $\widehat f_k$, $\widehat \pi_k$, $\widehat m_{a,k}^{\widehat\lambda_k}$ and $\widehat\lambda_k$ (if $\lambda$ is chosen to be an unknown distribution),
and use $\mathscr D_{k}$ to estimate the causal effect by $\widehat{\Delta}_{CF,k}^{\widehat{\lambda}_k} =\widehat\mu_{1,k}^{-1,\widehat\lambda_k}-\widehat\mu_{0,k}^{-1,\widehat\lambda_k}$, where
	\begin{equation*}
		\widehat\mu_{a,k}^{-1,\widehat\lambda_k}=  n_k^{-1}\sum_{i\in\mathscr D_{k}}\left[\widehat{m}_{a,k}^{\widehat\lambda_k}(X_i)+\dfrac{I(A_i=a) }{\widehat{f}_k(A_i \mid X_i)}\left\{\widehat{\Y}^{-1}_i\circ \widehat{\lambda}_k - \widehat{m}_{a,k}^{\widehat\lambda_k}(X_i)\right\}\right],
	\end{equation*}
	where $n_k$ is the sample size of the $k$th partition. 
    Finally, we combine the effects from different partitions via
    \begin{equation}\label{eq:dr-cf}
    \widehat\mu_{a, CF}^{-1,\widehat\lambda} =	 \sum\limits_{k=1}^K  \dfrac{n_k}{n}  	\widehat\mu_{a,k}^{-1,\widehat\lambda_k} \circ\widehat\lambda_k^{-1}\circ\widehat\lambda \quad \text{and} \quad \widehat\Delta_{CF}^{\widehat\lambda} =\widehat\mu_{1, CF}^{-1,\widehat\lambda}  - \widehat\mu_{0, CF}^{-1,\widehat\lambda}, 
    \end{equation}
      where $\widehat\lambda$ is estimated  using the entire dataset. In the above, if each reference distribution $\widehat\lambda_k$ is fixed to a common distribution $\widehat\lambda$, then \eqref{eq:dr-cf} is reduced to $\widehat\Delta_{CF}^{\widehat\lambda}=n^{-1}\sum\limits_{k=1}^K {n_k}\widehat\Delta_{CF,k}^{\widehat\lambda}$. Otherwise,  the objects $\widehat\Delta_{CF,k}^{\widehat\lambda_k}$  reside in  distinct spaces $L^2(\ldomain;\widehat\lambda_k)$ for $k=1,\ldots,K$, where $\ldomain$ is the domain of $\lambda$. 
    In this case, to combine $\widehat\Delta_{CF,k}^{\widehat\lambda_k}$, in \eqref{eq:dr-cf} we apply the optimal transport {between the measures $\widehat\lambda_k$ and $\widehat\lambda$} to move it from the space $L^2(\ldomain;\widehat\lambda_k)$ into the space $L^2(\ldomain;\widehat\lambda)$.
	
	{To reduce the sensitivity of the cross-fitting estimator to partitioning, as suggested by \cite{chernozhukov2018double}, one may repeat the estimator $\widehat{\Delta}_{CF}^{\widehat{\lambda}}$ for $R$ times over independent partitioning. This results in estimates $\widehat{\Delta}_{CF}^{\widehat{\lambda},r}$ for $r=1,\ldots,R$. We then estimate $\Delta^\lambda$ by
	\begin{equation}\label{eq:CF-median}
	    \widehat{\Delta}_{CF}^{\widehat{\lambda},med}(\cdot)=median\{\widehat{\Delta}_{CF}^{\widehat{\lambda},r}(\cdot)\}_{r=1}^R.
	\end{equation}}

	\section{Asymptotic Properties}\label{sec:theory}
	
	We study the asymptotic properties of the proposed estimators, including $\widehat\Delta^{\widehat{\lambda}}_{DR}$ and $\widehat\Delta^{\widehat{\lambda}}_{CF}$, in this section. 
 Let $\ldomain$, potentially coincides with $\tdomain$, be the domain of the reference distribution $\lambda$.
	{For simplicity, we assume $\tdomain$ and $\ldomain$ to be a bounded interval of $\real$.
This condition may be replaced by some moment conditions on $Y_i$ and other relevant quantities; see Remarks \ref{rem:S1} and \ref{rem:S2} in the Supplementary Material for details.}
	
In the following, we shall first introduce assumptions on the variability from estimating $\Y_i, i=1,\ldots, n$ and ${\lambda}$.

	\begin{assumption}\label{assu:eY}
		The estimates $\eY_1,\ldots,\eY_n$ are independent, and there are two sequences of constants $\alpha_n=o(1)$ and $\nu_n=o(1)$ such that
		\begin{equation}\label{eqn:2}
			\begin{aligned}
				\sup_{1\leq i\leq n}\sup_{\y\in\ws}\expect\{W_2^2(\eY_i,\Y_i)\mid \Y_i=\y\} & =O(\alpha_n^2),\\
				\sup_{1\leq i\leq n}\sup_{\y\in\ws}\var\{W_2^2(\eY_i,\Y_i)\mid \Y_i=\y\} & =O(\nu_n^4).
			\end{aligned}
		\end{equation}
	\end{assumption}
	
The conditional expectation $\expect[W_2^2(\hat Y_i,Y_i)|Y_i]$ in Assumption \ref{assu:eY} is a real-valued measurable function defined on $\ws$, the precise definition of which is given in Section \ref{sec:S-remark} of the Supplementary Material. 
Assumption \ref{assu:eY} requires that  $\widehat{\Y}_i, i=1,\ldots,n$, converge to their population counterparts at certain rates. Suppose that the number of observations for unit $i$, $k_i\asymp n^\zeta$ for some constant $\zeta>0$. If $\widehat{\Y}_i$ is obtained using the corresponding empirical distribution function, then under some additional moment assumptions,
	condition \eqref{eqn:2} holds with $\alpha_n^2= \nu_n^4=n^{-\zeta/2}$. 
	This can be shown via a combination of \citet[][Theorem 1]{Fournier2015} and  \citet[][Proposition 2.17]{Santambrogio2015}. Assumption \ref{assu:eY} also holds with many other standard non-parametric estimators for $\Y_i$. For example, under some regularity conditions on {the distribution of} $\Y$, 
	the estimator by \cite{Petersen2016} satisfies condition \eqref{eqn:2}  with a faster rate: $\alpha_n^2=n^{-2\zeta/3}$ and $\nu_n^4=n^{-4\zeta/3}$.

The following Assumption \ref{assu:DR-lambda} imposes a condition on the rate of convergence for $\widehat\lambda$. It is satisfied, for example, when $\lambda$ is a fixed and known distribution so that $\widehat\lambda=\lambda$. {Under 
Assumption \ref{assu:eY}, it holds for $\lambda=Y_i$. 
It also holds when $\lambda$ is the Fr\'echet mean of $Y_i$ and  $\widehat{\lambda} =\underset{\upsilon \in \ws}{\arg\min}\sum_{i=1}^n W_2^2(\upsilon,\eY_i)$ is the sample Fr\'echet mean; see Lemma \ref{lem:mu-rate-noisy-W} in the Supplementary Material.
\begin{assumption}\label{assu:DR-lambda} $W_2^2(\widehat\lambda,\lambda) = \Op\big(n^{-1}+\alpha_n^2+\nu_n^2\big)$.
\end{assumption}	

Let $\tilde{m}_a^\lambda$ be an estimate of $m_a^\lambda$ by using the outcome $Y_i, i=1,\ldots, n$ and $\lambda$.
To study the asymptotic properties of the causal effect map estimator $\widehat\Delta_{DR}^{\widehat\lambda}$,  we introduce Assumption \ref{assu:DR-std} that is standard in causal effect estimation \citep[e.g.][]{hernan2020}. {Part \ref{assu:DR-std-pi-bound} of Assumption \ref{assu:DR-std} assumes positivity of the estimated propensity scores, while part \ref{assu:DR-std-convergence-fixed-ref} assumes that $\tilde m_a^\lambda$ and $\widehat\pi$  converge to their limits uniformly.} 
\begin{assumption}\label{assu:DR-std}\,  
	\begin{enumerate}[label=\textup{(\alph*)}]
		\item\label{assu:DR-std-pi-bound}  The  estimated propensity scores are bounded away from zero:  for some $\epsilon>0$,  $\pr\{\inf_x\widehat\pi(x)>\epsilon \text{ and } \sup_x\widehat\pi(x)<1-\epsilon\}=1$.
		\item\label{assu:DR-std-convergence-fixed-ref} The outcome regression and propensity score estimates converge: $\sup_x\| \tilde m_a^\lambda(x)-m_a^{\lambda,\ast}(x)\|_\lambda=\op(1), a=0,1$ and $\sup_x|\widehat\pi(x)-\pi^\ast(x)|=\op(1)$ for some  functions $m_a^{\lambda,\ast}$  and $\pi^\ast$. 
	\end{enumerate}	
\end{assumption}

The following Assumption \ref{assu:DR-D-rate-fixed-ref} assumes that the outcome regression estimates $\widehat{m}_a^{\widehat\lambda}$ obtained using $\widehat\lambda$ and $\widehat{Y}_i, i=1,\ldots, n$  are not too far away from the quantities $\tilde{m}_a^\lambda$.  It holds for estimators $\widehat m_a^\lambda(x)$ that are Lipschitz continuous functions of a weighted average of $(\eY_1)^{-1} \circ \lambda,\ldots,(\eY_n)^{-1} \circ \lambda$. Examples include local polynomial estimators and parametric estimators satisfying certain regularity conditions. Note that the optimal transport $\widehat{\lambda}^{-1}\circ\lambda$ in the assumption is needed to transport $\widehat m_a^{\widehat\lambda}$, which resides in the space $L^2(\ldomain;\widehat{\lambda})$, into the  space $L^2(\ldomain;\lambda)$ where  $\tilde m_a^\lambda$ resides. 

\begin{assumption}\label{assu:DR-D-rate-fixed-ref} 
	$\mathbb{P}_n \|\widehat m_a^{\widehat\lambda}(X)\circ  \widehat{\lambda}^{-1} \circ \lambda - \tilde m_a^\lambda(X)\|_{\lambda}^2 = \Op\big(W_2^2(\widehat{\lambda},\lambda)+\alpha_n^2+\nu_n^2\big)$ for $a=0,1$.
\end{assumption}

To state the last condition, we first define the concept of Donsker class. 
Consider a fixed $a \in \{0,1\}.$ For a real number $t\in\ldomain$, a function $\breve m^\lambda:\xdomain\rightarrow \wwlog_\lambda(\ws):=\{\y^{-1}\circ\lambda:\y\in \ws\}$  and a function $\breve \pi:\tdomain\rightarrow\real$,  we view the element $t\times \breve m^\lambda \times \breve \pi$ as a real-valued function defined on $(A,X,Y)$ by 
$$(t\times \breve m^\lambda \times \breve \pi)(A,X,Y)=\textstyle\dfrac{(aA+(1-a)(1-A))\{(Y^{-1}\circ\lambda)(t)-\breve m_a^\lambda(X)(t)\}}{a\breve\pi(X)+(1-a)(1-\breve\pi(X))}+\breve m_a^\lambda(X)(t),$$
and for a family $\mathscr F_a^\lambda$ of functions in the form of $t\times \breve m^\lambda \times \breve \pi$, we view  $\sqrt n\mathcal G_a (t\times \breve m^\lambda\times \breve \pi) :=\sqrt{n}(\splavg-\expect)\{(t\times \breve m^\lambda \times \breve \pi)(A,X,Y)\}$ as a real-valued random process indexed by functions in $\mathscr F_a^\lambda$. Let $L^\infty (\mathcal F_a^\lambda)$ denote the collection of real-valued functions $H$ defined on $\mathscr F_a^\lambda$ and satisfying  $\sup_{q\in\mathcal F_a^\lambda}|H(q)|<\infty$. We say $\mathscr F_a^\lambda$ is a Donsker class if $\sqrt{n}\mathbb G_a$ converges to a tight Gaussian measure on $L^\infty (\mathscr F_a^\lambda)$.

Let $\mathscr M_a^\lambda$ be a class of functions containing $m_a^\lambda$. For $g,h\in\mathscr M_a^\lambda$, define $\eta_i(g,h)=\|g(X_i)-h(X_i)\|_\lambda$ and $\eta^2(g,h)=n^{-1}\sum_{i=1}^n \eta_i^2(g,h)$. Conditional on $\mathbb X=(X_1,\ldots,X_n)$, $\eta$ defines a pseudo-distance function on $\mathscr M_a^\lambda$. Let $B_a(r;g)=\{h\in\mathscr M_a^\lambda:\eta(g,h)<r\}$ be a ball in $\mathscr M_a^\lambda$ with radius $r$, and $N_a(\delta,r,\eta)$  denote the smallest number of $\delta$-balls in the pseudo-metric space $(\mathscr M_a^\lambda ,\eta)$ that are required to cover $B_a(r;m_a)$.  Without loss of generality, we assume $N_a(\delta,r,\eta)$ is continuous in $\delta$ and $r$. Otherwise, we just redefine it with its continuous upper bound.  

The following Assumption \ref{assu:DR-fixed-ref} imposes some technical conditions on the estimators $\widehat\pi$ and $\widehat m_a^\lambda$. {Part \ref{assu:DR-stability-fixed-ref} assumes that each data point has an asymptotically equal contribution to the estimators, i.e., {there is no outlier in the sense that as the sample size grows, the influence of a single data point on the estimates $\widehat{\pi}$ and $\tilde{m}_{a ^\lambda}$ is negligible.} Part \ref{assu:DR-donsker-fixed-ref} restricts $\mathscr F_a^\lambda$ to be a Donsker class,  and part \ref{assu:DR-ep} limits the complexity of $\mathscr M_a^\lambda$ via a bound on the metric entropy, which enables us to employ empirical process theory to provide an upper bound on the convergence rate of $\widehat\Delta_{DR}^{\widehat\lambda}$.  These conditions are satisfied, for example, by a logistic model for $\pi$ and a simple linear regression model for $m_a^\lambda$.}

\begin{assumption}\label{assu:DR-fixed-ref}\,  
	\begin{enumerate}[label=\textup{(\alph*)}]
		\item\label{assu:DR-stability-fixed-ref} Stability of the estimators: For a constant $C_4>0$, $\sup_{1\leq i\leq n} \expect |\widehat\pi(X_i)-\widehat\pi_{-i}(X_i)|^2 \leq C_4n^{-1}$ and $\sup_{a}\sup_{1\leq i\leq n} \expect \|\tilde m_{a}^\lambda(X_i)-\tilde m_{a,-i}^\lambda(X_i)\|^2_\lambda \leq C_4n^{-1}$, where $\widehat\pi_{-i}$ and $\tilde m_{a,-i}^\lambda$ are the estimates of $\pi$ and $m_a^\lambda$ without using the $i$th subject, respectively.
		\item \label{assu:DR-donsker-fixed-ref}
		For $a=0,1$, the class $\mathscr F_a^\lambda$ is a Donsker class containing $t\times m^{\lambda,\ast}_a\times \pi^\ast$ for all $t$,  and with probability tending to one, $t\times\tilde m_a^\lambda\times \widehat\pi\in \mathscr F_a^\lambda$ for all $t$.
		\item \label{assu:DR-ep}{For $a=0,1$, for some fixed $K>0$, $\tilde m_a^\lambda\in \mathscr M_a^\lambda$,  $\widehat m_a^{\widehat\lambda}(\cdot)\circ\widehat\lambda^{-1}\circ\lambda\in\mathscr M_a^\lambda$,  and $\log N_a(\delta,r,\eta)\leq Kr\delta^{-1}$ for all $r,\delta>0$ almost surely.} 
	\end{enumerate}	
\end{assumption}

\begin{remark}
\label{remark:donsker}
    Assumption \ref{assu:DR-fixed-ref}   may fail in settings invoking machine learning methods, in which case the dimension of covariates $X$ is modelled as an increasing function of the sample size \citep{chernozhukov2018double}. Notably Theorem \ref{thm:AN-W-fixed-ref-cf} for the cross-fitting estimator   $\widehat\Delta_{CF}^{\widehat\lambda}$  does not require this assumption and thus can accommodate machine learning methods for modeling $\pi$ and $m_a^\lambda$.
\end{remark}

	Let $\vertiii{\tilde m_a^\lambda-m_a^\lambda}^2_\lambda=\int \|\tilde m_a^\lambda(x)-m_a^\lambda(x)\|_{\lambda}^2\diffop F_X(x)$,  the integrated squared error of  $\tilde m_a^\lambda$ for estimating $m_a^\lambda$, where $F_X$ denotes the probability measure induced by $X$. Similarly, the integrated squared error of  $\widehat\pi$ for estimating $\pi$ is denoted by $\|\widehat\pi-\pi\|_{2}^2=\int |\widehat\pi(x)-\pi(x)|^2\diffop F_X(x)$. Define $\varrho_\pi^4:=\varrho_\pi^4(n):=\expect \|\widehat\pi-\pi\|_{2}^4$ and  $\varrho_{m}^4:=\varrho_{m}^4(n):=\max\{\expect \vertiii{\tilde m_0^\lambda-m_0^\lambda}^4_\lambda,\expect \vertiii{\tilde m_1^\lambda-m_1^\lambda}^4_\lambda\}$. Let $L^2(\lambda)=\{h\in\real^{\ldomain}:\int_{\ldomain}|h|^2\diffop\lambda<\infty\}$ be endowed with the inner product $\langle h_1,h_2\rangle_\lambda=\int_{\ldomain} h_1 h_2\diffop\lambda$ and the induced norm $\|h\|_\lambda=\langle h,h\rangle_\lambda^{1/2}$. Finally, let 
	\begin{equation*}
	    \varphi(A,X,Y)=\textstyle\dfrac{A\{ Y^{-1}\circ \lambda -m_1^\lambda(X)\}}{\pi(X)}+m_1^\lambda(X) -\textstyle\dfrac{(1-A)\{Y^{-1}\circ \lambda -m_0^\lambda(X)\}}{1-\pi(X)}-m_0^\lambda(X).
	\end{equation*}

	The following Theorem \ref{thm:AN-W-fixed-ref}  shows that the estimator $\widehat\Delta_{DR}^{\widehat\lambda}$  enjoys the double robustness property whether $\widehat\lambda$ is a fixed and known distribution or is estimated from data, so that the convergence rate is $n^{-1/2}$ when either of $\varrho_\pi$ and $\varrho_m$ is of the order $n^{-1/2}$ and the other one is bounded. 
Moreover, one may use flexible non-parametric methods for estimating $m_a^{\widehat\lambda}, a=0,1$ and $\pi$, provided that the nonparametric convergence rates satisfy $\varrho_m\varrho_\pi=o(n^{-1/2})$. 
{Theorem \ref{thm:AN-W-fixed-ref}  also shows that $\widehat\Delta_{DR}^{\widehat\lambda}$ is an asymptotically linear estimator with influence function $\varphi(A,X,Y)-\expect\varphi(A,X,Y)$}. 
	
	\begin{theorem}\label{thm:AN-W-fixed-ref}{Suppose that both $\widehat\lambda$ and $\lambda$ are continuous distribution functions.} 
	If Assumptions \ref{assu:ignorability}--\ref{assu:DR-fixed-ref} hold with $\alpha_n=o(n^{-1/2})$ and $\nu_n=o(n^{-1/2})$,  then 
		\begin{enumerate}[label={(\roman*)}]
			\item\label{thm:AN-W-rate-fixed-ref}   	$\|\widehat\Delta_{DR}^{\lambda}\circ\widehat\lambda^{-1}\circ\lambda  - \Delta^\lambda \|_{\lambda}=\Op(n^{-1/2}+n^{-1/2}\varrho_m^{1/2}+n^{-1/2}\varrho_\pi+\varrho_{m}\varrho_\pi)$;
			\item\label{thm:AN-W-AN-fixed-ref} if $\varrho_m\varrho_\pi=o(n^{-1/2}), \varrho_m = o(1), \varrho_\pi = o(1)$, then  
			$\sqrt n \left(\widehat\Delta_{DR}^{\lambda}\circ\widehat\lambda^{-1}\circ\lambda - \Delta^\lambda \right)=\sqrt{n}(\mPn-\expect)\{\varphi(A,X,Y)\}+\op(1)$, and consequently $\sqrt n \left(\widehat\Delta_{DR}^{\lambda}\circ\widehat\lambda^{-1}\circ\lambda - \Delta^\lambda \right)$
			converges weakly to a centered Gaussian process in the space {$L^2(\ldomain;\lambda)$} with the same asymptotic distribution as $\sqrt{n}(\mPn-\expect)\{\varphi(A,X,Y)\}$.
		\end{enumerate}
	\end{theorem}

\begin{remark}\label{rem:cov}
    The covariance function of the limit distribution of $\sqrt{n}(\mPn-\expect)\{\varphi(A,X,Y)\}$ can be estimated from data, as follows. 
    Let $V=\textstyle\dfrac{A\{ Y^{-1}\circ \lambda -m_1^\lambda(X)\}}{\pi(X)}+m_1^\lambda(X) -\textstyle\dfrac{(1-A)\{Y^{-1}\circ \lambda -m_0^\lambda(X)\}}{1-\pi(X)}-m_0^\lambda(X)$ and $G=\sqrt{n}(\mPn-\expect)\{\varphi(A,X,Y)\}$. Then $G$ and its limit distribution share the same covariance function of $V$. Letting $\widehat V_i=\textstyle\dfrac{A_i\{ \eY^{-1}_i\circ \widehat\lambda -\widehat m_1^{\widehat\lambda}(X_i)\}}{\widehat\pi(X_i)}+\widehat m_1^{\widehat\lambda}(X_i) -\textstyle\dfrac{(1-A_i)\{\eY^{-1}_i\circ \widehat\lambda -\widehat m_0^{\widehat\lambda}(X_i)\}}{1-\widehat\pi(X_i)}-\widehat m_0^{\widehat\lambda}(X_i)$, we can use the sample covariance $\widehat C(s,t)=n^{-1}\sum_{i=1}^n \{V_i(s)-\bar V(s)\}\{V_i(t)-\bar V(t)\}$ as an estimate of the covariance function of $V$, where $s,t\in\ldomain$ and $\bar V=n^{-1}\sum_{i=1}^n V_i$.
\end{remark}
\begin{remark}\label{rem:scb}
    The above asymptotic result enables inference on $\Delta^{\widehat\lambda}$. Take the case that $\lambda$ is fixed and known so that $\widehat{\lambda}=\lambda$ for example. An approximate simultaneous confidence band (SCB) in the form of $[\widehat\Delta_{DR}^{\lambda}(t) -q_{\alpha/2}n^{-1/2},\widehat\Delta_{DR}^{\lambda}(t) + q_{\alpha/2}n^{-1/2}]$ with a constant $q_{\alpha/2}$ for all $t\in\ldomain$, i.e., $P\{\sup_{t\in\ldomain}\sqrt{n}|\widehat\Delta_{DR}^{\lambda}(t) - \Delta^\lambda(t)|\leq q_{\alpha/2}\}\approx 1-\alpha$ for a significance level $\alpha$, can be derived for $\Delta^\lambda$, by estimating $q_{\alpha/2}$ via a resampling strategy. Specifically, we draw $B$ realizations $G_1,\ldots,G_B$, for example, $B=1000$, from the centered Gaussian process with the covariance $\widehat C$, and for each realization $G_j$ we compute $g_j=\sup_{t\in\ldomain}|G_j(t)|$. Then $q_{\alpha/2}$ is estimated by the $1-\alpha/2$ empirical quantile $\widehat q_{\alpha/2}$ of $g_1,\ldots,g_B$, and the approximate SCB is given by $[\widehat\Delta_{DR}^{\lambda}(t)-\widehat q_{\alpha/2}n^{-1/2},\widehat\Delta_{DR}^{\lambda}(t)+\widehat q_{\alpha/2}n^{-1/2}]$. With the derived SCB, one can also test the null hypothesis $\Delta^\lambda\equiv 0$, for example, by rejecting the null at the significance level $\alpha$ if $0\not \in [\widehat\Delta_{DR}^{\lambda}(t)-\widehat q_{\alpha/2}n^{-1/2},\widehat\Delta_{DR}^{\lambda}(t)+\widehat q_{\alpha/2}n^{-1/2}]$ for some $t\in\ldomain$. Note that the same procedure applies to the estimator $\widehat\Delta^{\widehat\lambda}$ and the cross-fitting estimators in light of the theorems in the sequel.
\end{remark}

	{Next we turn to the cross-fitting estimator \eqref{eq:dr-cf} and show that it enjoys  double robustness and asymptotic normality properties without the technical assumption \ref{assu:DR-fixed-ref} \citep[e.g.][]{chernozhukov2018double}.	For this, we require that the sizes of the partitions are of the same order, quantified by Assumption \ref{assu:cf-sample-size}, and we tailor Assumption \ref{assu:DR-D-rate-fixed-ref} to the cross-fitting estimator by Assumption \ref{assu:DR-D-rate-fixed-ref-cf}. Let $\tilde m_{a,k}^\lambda$, the counterpart of $\widehat m_{a,k}^{\widehat\lambda}$, be estimated by using the outcomes $Y_i$ (instead of $\eY_i$) and the reference distribution $\lambda$ (instead of $\widehat\lambda$). 
	}

    \begin{assumption}\label{assu:DR-D-rate-fixed-ref-cf} 
    	$\vertiii{\widehat m_{a,k}^{\widehat\lambda_k}(\cdot)\circ  \widehat{\lambda}_k^{-1} \circ \lambda - \tilde m_{a,k}^\lambda}_{\lambda}^2 = \Op\big(W_2^2(\widehat{\lambda}_k,\lambda)+\alpha_n^2+\nu_n^2\big)$ for $a=0,1$ and $k=1,\ldots,K$.
	\end{assumption}

	\begin{assumption}\label{assu:cf-sample-size}
		There exist constants $c_1$ and $c_2$ such that $0<c_1\leq n_k/n \leq c_2 < 1$ for all $n$ and $k=1,\ldots,K$.
	\end{assumption}

	\begin{theorem}\label{thm:AN-W-fixed-ref-cf}{Suppose that both $\widehat\lambda$ and $\lambda$ are continuous distribution functions.} 
	If Assumptions \ref{assu:ignorability}--\ref{assu:eY} and \ref{assu:DR-D-rate-fixed-ref-cf}--\ref{assu:cf-sample-size} hold with $\alpha_n=o(n^{-1/2})$ and $\nu_n=o(n^{-1/2})$, and additionally, Assumptions \ref{assu:DR-lambda}--\ref{assu:DR-std} hold for $\widehat m_{a,k}^{\widehat\lambda}$, $\tilde m_{a,k}^{\lambda}$, $\widehat\pi$ and $\widehat\lambda_k$ for $k=1,\ldots,K$, then for the estimator defined in \eqref{eq:dr-cf}, we have
	\begin{enumerate}[label={(\roman*)}]
		\item\label{thm:AN-W-rate-fixed-ref-cf} 	$\|\widehat\Delta_{CF}^{\widehat\lambda} \circ \widehat\lambda^{-1} \circ {\lambda}  - \Delta^\lambda \|_{\lambda}=\Op(n^{-1/2}+n^{-1/2}\varrho_m+n^{-1/2}\varrho_\pi+\varrho_{m}\varrho_\pi)$;
		\item\label{thm:AN-W-AN-fixed-ref-cf} if $\varrho_m\varrho_\pi=o(n^{-1/2}), \varrho_m = o(1), \varrho_\pi = o(1)$, then  
			$\sqrt n \left(\widehat\Delta_{CF}^{\widehat\lambda} \circ \widehat\lambda^{-1} \circ {\lambda} - \Delta^\lambda \right)=\sqrt{n}(\mPn-\expect)\{\varphi(A,X,Y)\}+\op(1)$, and consequently $\sqrt n \left(\widehat\Delta_{CF}^{\widehat\lambda} \circ \widehat\lambda^{-1} \circ {\lambda} - \Delta^\lambda \right)$
			converges weakly to a centered Gaussian process in the space {$L^2(\ldomain;\lambda)$} with the same asymptotic distribution as $\sqrt{n}(\mPn-\expect)\{\varphi(A,X,Y)\}$.
	\end{enumerate}
	\end{theorem}

    The above results show that the cross-fitting estimator $\widehat \Delta_{CF}^{\widehat\lambda}$ also enjoys  double robustness. A simultaneous confidence band for $\widehat \Delta_{CF}^{\widehat\lambda}$ can be constructed  using the method described in Remark \ref{rem:scb}. 
    For the estimator $\widehat{\Delta}^{\widehat{\lambda},med}_{CF}$ in \eqref{eq:CF-median}, as in \cite{chernozhukov2018double}, the covariance function $C(\cdot,\cdot)$ of the process $\sqrt n(\mPn-\expect)\{\varphi(A,X,Y)\}$ may be estimated as follows. {For $r=1,2,\ldots,R$, let $\tilde C^r(s,t)=\widehat{C}^{r}(s,t)+\{\widehat\Delta_{CF}^{\widehat{\lambda},r}(s)-\widehat\Delta_{CF}^{\widehat{\lambda},med}(s)\}\{\widehat\Delta_{CF}^{\widehat{\lambda},r}(t)-\widehat\Delta_{CF}^{\widehat{\lambda},med}(t)\}$, where $\widehat\Delta_{CF}^{\widehat{\lambda},r}$ is given at the end of Section \ref{sec:identifiability}, and $\widehat C^r$ is the estimated covariance for $\widehat\Delta_{CF}^{\widehat{\lambda},r}$ as per Remark \ref{rem:cov}. Let $\hat r$ be the index of $\tilde C^r$ whose operator norm is a median among $\tilde C^1,\ldots,\tilde C^R$. Then we use $\tilde C^{\hat r}$  as the estimate of the covariance function  of $\widehat{\Delta}^{\widehat{\lambda},med}_{CF}$}.

	\section{Simulation Studies}\label{sec:simulation}
	
	Our simulation data consist of $n$ independent samples from the joint distribution of $(X,A, Y).$
	The confounder $ X$ follows a uniform distribution on $ [-1, 1] $. Conditional on $X$, the treatment $A$ follows a Bernoulli distribution with mean $P(A=1\mid X)=\textnormal{expit}(1+ X)$, where $\textnormal{expit}(\cdot) = \exp(\cdot)/(1+\exp(\cdot))$.
	The outcome $Y$, which is a random distribution function, is generated through the corresponding quantile function $Y^{-1}(\alpha) = (- \expect(A) + A  + X + \epsilon) \sin(\pi \alpha)/8 + \alpha, \alpha\in  [0,1]$, 
	where $\epsilon$ is independently generated from a uniform distribution on $ [-0.5, 0.5] $.  Note that for any realization of $X$ and $A$, $Y^{-1}(0) = 0, Y^{-1}(1) = 1,$ and   $Y^{-1}(\alpha)$ is a continuous and strictly increasing function of $\alpha,$ so that $Y^{-1}$ is a quantile function for a continuous random variable taking values in $\mI = [0,1].$  It follows that $Y \in \ws.$ We are interested in estimating the causal effect at the reference distribution $\mu_0$,  whose true value is $\Delta^{\mu_0}(t) = \sin(\pi t) / 8, t\in \mI.$

	The sample sizes are $ n=50, 200, 1000 $. {For each subject  $i = 1,\ldots, n$, we assume that we have access to 1001 independent and identically distributed observations sampled from the distribution  function $ \Y_i(t), t\in [0,1]$.  We estimated the outcomes $Y_i(t)$ by the empirical cumulative distribution function $\widehat{Y}_i(t)$ based on these observations.}

	We considered two  specifications for the outcome regression model: a linear regression model with predictor  $ X $  in $m_A(X)$ (correct) and one with predictor  $ X^2 $  (incorrect), and 
	two specifications for the
	propensity score model: a logistic regression model with predictor $X$ (correct) and one with predictor $X^2$ (incorrect). The estimation error is quantified using two measures: (1) bias of difference in medians, i.e., $\widehat\Delta^{\widehat{\mu}_0}(\widehat{\mu}_0^{-1}(0.5)) - \Delta^{\mu_0}(\mu_0^{-1}(0.5))$; (2) root mean integrated squared error   under the reference distribution $\mu_0$, i.e.,  $\|\widehat\Delta^{\widehat{\mu}_0}\circ \widehat{\mu}_0^{-1}\circ \mu_0 -\Delta^{\mu_0}\|_{\mu_0}$.

	{To illustrate  double robustness of the estimators $\widehat\Delta_{DR}^{\widehat\mu_0}$ and $\widehat\Delta_{CF}^{\widehat\mu_0,med}$, we compare them with the inverse probability weighting (IPW) and outcome regression (OR) estimators, defined by $
		\widehat{\Delta}^{\widehat{\mu}_0}_{OR} =  \sum\limits_{a=0,1} (2a-1) \mathbb{P}_n \widehat{m}_a^{\widehat\mu_0}(X)
		$
		and 
		$
		\widehat{\Delta}_{IPW}^{\widehat{\mu}_0} =  \sum\limits_{a=0,1} (2a-1) \mathbb{P}_n {I(A=a) (\widehat{\Y}^{-1}\circ \widehat{\mu_0})}/{\widehat{f}(A \mid X)}.
		$
	}
{The cross-fitting estimator $\widehat{\Delta}^{\widehat\mu_0,med}_{CF}$ is based on the median of cross-fitting estimators from $100$ random splits, where for each random split, we consider the 5-fold cross-fitting.}

	\Cref{table2} summarizes the simulation results based on 1000 Monte Carlo replicates.
	The bias of difference in medians  of the doubly robust estimator becomes closer to zero as the sample size increases when either the outcome regression or propensity score model is correct, thus confirming ``double robustness.'' In comparison, neither the OR nor IPW estimator has the double robustness property: When the corresponding model is misspecified, their bias of difference in medians can be large even with a sample size of 1000. {Similar results hold for root mean integrated squared error.}  We further note that when the outcome resides in a Euclidean space, it is well-known that when both models are correct, the standard error of the OR estimator is no larger than that of the DR estimator, which is in turn no larger than that of the IPW estimator. One can see a similar phenomenon  from Table \ref{table2}, where the outcome is a random distribution function residing in $\ws$.  Although the cross-fitting estimator $\widehat{\Delta}^{\widehat{\mu}_{0},med}_{CF}$ is  more appealing theoretically, for the setting we consider here, $\widehat{\Delta}^{\widehat{\mu}_{0}}_{DR}$ has better finite sample performance, especially when the sample size is small. 
	
	We assess the finite-sample coverage of the proposed confidence bands in Remark \ref{rem:scb} when both models are correctly specified. For the DR method, the coverage probabilities of the $95\%$ SCB are $88.4\%$, $91.5\%$, and $ 92.4\%$, respectively for $n=50,200,1000$, based on $1000$ Monte Carlo replicates. For the CF method, the coverage probabilities are $80.8\%$, $91.4\%$, and $ 92.4\%$, respectively for $n=50,200,1000$. {These coverage probabilities are reasonably close to the nominal level considering the  difficulty to
derive effective confidence bands for functional objects.}

	\begin{table}
		\begin{center}
			\caption{Bias of difference in medians $\times 100$ (standard error $\times 100$) and root mean integrated squared error $\times 100$ (standard error $\times 100$)   of four proposed estimators: the outcome regression  estimator $\widehat{\Delta}^{\widehat{\mu}_{0}}_{OR}$, the inverse probability weighting   estimator $\widehat{\Delta}^{\widehat{\mu}_{0}}_{IPW}$, the doubly robust   estimator $\widehat{\Delta}^{\widehat{\mu}_{0}}_{DR}$, and the cross-fitting   estimator $\widehat{\Delta}^{\widehat{\mu}_{0},med}_{CF}$}
			\bigskip
			\label{table2}
			\begin{tabular}{rcccccc}
				\toprule
				Estimator  &       \multicolumn{2}{c} { Model} &  \multicolumn{3}{c} {Sample size} \\
				\cmidrule(r){2-3} \cmidrule(l){4-6}
				& \multicolumn{1}{c}{PS} & \multicolumn{1}{c}{OR} &   \multicolumn{1}{c}{{$n=50$}} & \multicolumn{1}{c}{$n=200$} & {$n=1000$} &  \\
				\midrule
				\multicolumn{4}{l}{Bias of difference in medians $\times 100$  }   & & \\[5pt]
				$\widehat{\Delta}^{\widehat{\mu}_{0}}_{OR}$  &  $-$ & $\checkmark$	&0.010(0.037) &$-$0.036(0.018) &0.008(0.008)  \\
				$\widehat{\Delta}^{\widehat{\mu}_{0}}_{OR}$  &  $-$ & $\times$ &3.898(0.075) &3.787(0.037) &3.819(0.017)  \\ [5pt]  
				$\widehat{\Delta}^{\widehat{\mu}_{0}}_{IPW}$  & $\checkmark$  & $-$ &0.433(0.157) &0.008(0.045) &0.021(0.018)  \\
				$\widehat{\Delta}^{\widehat{\mu}_{0}}_{IPW}$  & $\times$ & $-$ &3.975(0.084) &3.706(0.036) &3.738(0.016)  \\[5pt]  
				$\widehat{\Delta}^{\widehat{\mu}_{0}}_{DR}$  & $\checkmark$   & $\checkmark$ &0.005(0.04) &$-$0.033(0.019) &0.013(0.008)  \\
				$\widehat{\Delta}^{\widehat{\mu}_{0}}_{DR}$  & $\checkmark$   & $\times$ &0.194(0.057) &$-$0.001(0.022) &0.020(0.010)  \\
				$\widehat{\Delta}^{\widehat{\mu}_{0}}_{DR}$  & $\times$ & $\checkmark$  &0.015(0.038) &$-$0.038(0.018) &0.009(0.008)  \\
				$\widehat{\Delta}^{\widehat{\mu}_{0}}_{DR}$  & $\times$ & $\times$ &3.881(0.074) &3.704(0.036) &3.737(0.016)  \\ [5pt]
				$\widehat{\Delta}^{\widehat{\mu}_{0},med}_{CF}$ & $\checkmark$   & $\checkmark$ &$-$0.066(0.091) &-0.032(0.019) & 0.013(0.008) \\ 
				$\widehat{\Delta}^{\widehat{\mu}_{0},med}_{CF}$  & $\checkmark$   & $\times$ &$-$2.151(0.141) &$-$0.383(0.025) & $-$0.051(0.010) \\
				$\widehat{\Delta}^{\widehat{\mu}_{0},med}_{CF}$  & $\times$ & $\checkmark$  &0.014(0.047) &$-$0.038(0.018) & 0.009(0.008) \\
				$\widehat{\Delta}^{\widehat{\mu}_{0},med}_{CF}$  & $\times$ & $\times$ &4.193(0.108) &3.726(0.036) & 3.741(0.016) \\[15pt]
			
				\multicolumn{4}{l}{Root mean integrated squared error $\times 100$  }   & & \\[5pt]
				$\widehat{\Delta}^{\widehat{\mu}_{0}}_{OR}$  & $-$  & $\checkmark$ &0.695(0.016) &0.339(0.008) &0.150(0.003)  \\
				$\widehat{\Delta}^{\widehat{\mu}_{0}}_{OR}$  &  $-$ & $\times$ &2.941(0.05) &2.773(0.027) &2.796(0.012)  \\[5pt]
				$\widehat{\Delta}^{\widehat{\mu}_{0}}_{IPW}$  & $\checkmark$  & $-$ &2.912(0.108) &0.981(0.029) &0.425(0.009)  \\ 
				$\widehat{\Delta}^{\widehat{\mu}_{0}}_{IPW}$  & $\times$ & $-$ &3.176(0.057) &2.715(0.026) &2.736(0.012)  \\[5pt] 
				$\widehat{\Delta}^{\widehat{\mu}_{0}}_{DR}$  & $\checkmark$   & $\checkmark$ &0.740(0.017) &0.348(0.008) &0.156(0.004)  \\ 
				$\widehat{\Delta}^{\widehat{\mu}_{0}}_{DR}$  & $\checkmark$   & $\times$  &0.991(0.028) &0.409(0.010) &0.183(0.004)  \\
				$\widehat{\Delta}^{\widehat{\mu}_{0}}_{DR}$  & $\times$ & $\checkmark$  &0.709(0.016) &0.342(0.008) &0.150(0.003)  \\
				$\widehat{\Delta}^{\widehat{\mu}_{0}}_{DR}$  & $\times$ & $\times$  &2.942(0.048) &2.712(0.026) &2.736(0.012)  \\ [5pt]
				$\widehat{\Delta}^{\widehat{\mu}_{0},med}_{CF}$ & $\checkmark$   & $\checkmark$ &0.981(0.054) &0.352(0.008) & 0.156(0.004) \\ 
				$\widehat{\Delta}^{\widehat{\mu}_{0},med}_{CF}$  & $\checkmark$ & $\times$   &2.163(0.107) &0.495(0.013) & 0.188(0.004)  \\
				$\widehat{\Delta}^{\widehat{\mu}_{0},med}_{CF}$  & $\times$ & $\checkmark$  &0.834(0.047) &0.343(0.008) & 0.150(0.003) \\
				$\widehat{\Delta}^{\widehat{\mu}_{0},med}_{CF}$ & $\times$ & $\times$ &3.289(0.070) &2.728(0.026) & 2.738(0.012) \\
				\bottomrule
			\end{tabular}
		\end{center}
	\end{table}  
	
	We further consider the data-adaptive DR and CF estimators for two scenarios. Scenario 1 is exactly the same as the previous simulation. In Scenario 2, both the propensity score model and the outcome regression model are nonlinear functions of $X$. In particular, we set $P(A=1\mid X)=\textnormal{expit}(1+ \sin(\pi X))$ and $Y^{-1}(\alpha) = (- \expect(A) + A  + \sin(\pi X) + \epsilon) \sin(\pi \alpha)/8 + \alpha, \alpha\in  [0,1]$, i.e. we replace the term $X$ in these two models by $ \sin(\pi X)$. The remaining settings are the same as those in Scenario 1. In both scenarios, we fit the outcome regression model using smoothing spline, implemented by the {\tt R} function {\tt smooth.spline}, while for the propensity score model, we consider the logistic smoothing spline fit, implemented by the {\tt R} function {\tt gssanova} in package {\tt gss}. \label{page:details} The default tuning methods of these {\tt R} functions are adopted, namely, generalized cross-validation for {\tt smooth.spline} and cross-validation for {\tt gssanova}; see the corresponding {\tt R} packages for more details.   The results based on 1000 Monte Carlo replicates are summarized in Table \ref{tabledataadaptive}. From the results, one can see that for Scenario 1, the data-adaptive methods work reasonably well, although not as good as the method where we correctly specify both the outcome and propensity score models parametrically. {When both  the underlying true outcome and propensity score models are nonlinear, the root mean integrated squared errors of the data-adaptive methods decay as the sample size increases, suggesting  consistency of these methods for estimating the average treatment effect.}

	\begin{table}
		\begin{center}
			\caption{Bias of difference in medians $\times 100$ (standard error $\times 100$) and root mean integrated squared error $\times 100$ (standard error $\times 100$)   of data-adaptive  doubly robust estimator $\widehat{\Delta}^{\widehat{\mu}_{0}}_{DRA}$ and the data-adaptive cross-fitting   estimator $\widehat{\Delta}^{\widehat{\mu}_{0},med}_{CFA}$ }
			\bigskip
			\label{tabledataadaptive}
			\begin{tabular}{rccccc}
				\toprule
				Estimator  & Scenario &  \multicolumn{3}{c} {Sample size} \\
				\cmidrule(l){3-5}
				& &   \multicolumn{1}{c}{{$n=50$}} & \multicolumn{1}{c}{$n=200$} & {$n=1000$} &  \\
				\midrule
				\multicolumn{4}{l}{Bias of difference in medians $\times 100$  }   & & \\[5pt]
				$\widehat{\Delta}^{\widehat{\mu}_{0}}_{DRA}$  & 1 & $-$0.047(0.127) &$-$0.096(0.046) &0.012(0.008)  \\
				$\widehat{\Delta}^{\widehat{\mu}_{0}}_{DRA}$  & 2 &$-$0.205(0.139) &$-$0.010(0.042) &0.030(0.008)  \\ [5pt]  
				$\widehat{\Delta}^{\widehat{\mu}_{0}}_{CFA}$  & 1 &$-$2.505(0.534) &$-$0.469(0.035) &$-$0.102(0.011)  \\
				$\widehat{\Delta}^{\widehat{\mu}_{0}}_{CFA}$  & 2  &$-$1.206(0.788) &$-$0.808(0.122) &0.180(0.021)  \\
\multicolumn{4}{l}{Root mean integrated squared error $\times 100$  }   & & \\[5pt]
				$\widehat{\Delta}^{\widehat{\mu}_{0}}_{DRA}$  & 1 &1.604(0.091) &0.545(0.044) &0.156(0.004)  \\
				$\widehat{\Delta}^{\widehat{\mu}_{0}}_{DRA}$  & 2 &1.867(0.096) &0.686(0.153) &0.161(0.004)  \\ [5pt]  
				$\widehat{\Delta}^{\widehat{\mu}_{0}}_{CFA}$  & 1 &3.648(0.355) &0.767(0.024) &0.244(0.006)  \\
				$\widehat{\Delta}^{\widehat{\mu}_{0}}_{CFA}$  & 2 &5.111(0.894) &2.202(0.183) &0.568(0.012)  \\
				\bottomrule
			\end{tabular}
		\end{center}
	\end{table}

	\section{Data Application}\label{sec:app}

	 Behavioral scientists are often interested in evaluating the effects of potential risk factors, such as marriage, on physical activity patterns \citep[e.g.][]{king1998effects}. In this section, we apply our proposed method to estimate the causal effect of marriage on physical activity levels, with data 
	obtained from the National Health and Nutrition Examination Survey (NHANES) 2005-2006\footnote{
		\href{https://wwwn.cdc.gov/nchs/nhanes/ContinuousNhanes/Default.aspx?BeginYear=2005}{https://wwwn.cdc.gov/nchs/nhanes/ContinuousNhanes/Default.aspx?BeginYear=2005}.}\label{page:footnote}. The NHANES is a program of studies designed to assess the health and nutritional status of adults and children in the United States. The survey is unique in that it combines interviews and physical examinations. The NHANES interview includes demographic, socioeconomic, dietary, and health-related questions. The examination component consists of medical, dental, and physiological measurements, as well as laboratory tests administered by highly trained medical personnel.
	
	In the 2005-2006 cycle of NHANES, participants of ages six years and older were asked to wear an Actigraph 7164 on a waist belt during all non-sleeping hours for seven days. The technology and application of current accelerometer-based devices in physical activity research allow the capture and storage or transmission of large volumes of raw acceleration signal data \citep{troiano2014evolution}. The NHANES accelerometer data have been widely used by researchers to explore relationships among accelerometer measures and a
	variety of other measures \citep[e.g.][]{tudor2012peer}. The monitors were programmed to begin recording activity information for successive 1-minute intervals (epochs) beginning at 12:01 a.m. the day after the health examination. The device was placed on an elasticized fabric belt, custom-fitted for each subject, and worn on the right hip. Subjects were told to keep the device dry (i.e. remove it before swimming or bathing) and to remove the device at bedtime. For each participant, the physical activity intensity, ranging from $0$ to $32767$ counts per minute (cpm), was recorded every minute for $24$ hours, $7$ days, where 32767 is the maximum value that the wearable device can record.

	In our analysis, the exposure of interest is marriage, coded as a  binary variable, with $1$ being married or living with a partner, and $0$ being otherwise. To define the outcome variable, we note that the trajectory of activity intensity  is not directly comparable across different subjects as different individuals might have different circadian rhythms. Instead,  the distribution of activity intensity is invariant to circadian rhythms and hence can be compared between groups of individuals. Specifically, our outcome of interest is $\Y(s)=\mathrm{Leb}(\{t:Z(t)\leq s\})/7$, the distribution of physical activity intensity over 7 days, where $\mathrm{Leb}$ denotes the Lebesgue measure.
	
	To obtain robust and reliable results, we applied the following preprocessing steps. Firstly, we excluded all observations that data reliability is questionable following NHANES protocol, after which there were $7170$ subjects left. Secondly, following \cite{Chang2020}, for each subject, 
	we removed observations with intensity values higher than 1000 or equal to 0.  In the data set,	 most intensity values are between 0 and $1000$ cpm.  Observations with zero intensity value were removed as they could represent activities with very different intensities, such as sleeping,  bathing and swimming. 
 Thirdly, we removed subjects with no more than 100 observations left, which further reduced the sample size to $7014$. Lastly,  for illustrative purposes, we removed $1490$ participants for whom we do not have information on their marital status.

After the preprocessing steps, we are left with  $5524$ participants  in the data set, among which $2682$ are in the married group and $2842$ participants are in the unmarried group.  The average age was $40.2$ years old with a standard deviation of $21.3$, and $52.3\%$ of them were female. As an example of the outcome data, in Figure \ref{real_f_1}(a), we plot the empirical cumulative
	distribution function for a randomly selected participant (subject ID 31144) who was 21 years old, male, and unmarried. 
In Figure \ref{real_f_1}(b), we  plot the Wasserstein barycentres of the empirical cumulative distribution functions in the married group and unmarried group, respectively. {One can see that the Wasserstein barycentres retain the key structural information in the individual empirical CDFs. For example, their derivatives decrease with the intensity level, suggesting that the physical intensity level is low most of the time. }
The Wasserstein barycentre in the married group is stochastically greater than that in the unmarried group, suggesting a positive association between marriage and physical activity level. However, this crude association may be subject to potential confounders such as age and gender.
	
	A simple approach to answering our question of interest is to first summarize the distribution functions with their means and then apply standard approaches such as the doubly robust estimator of \cite{robins1994estimation} to estimate the causal effect. With a linear outcome regression model and a logistic propensity score model, this simple doubly robust approach suggests that marriage increases the average physical intensity by $21.7$ ($95\%$ CI = $[17.1, 26.3]$) cpm.

	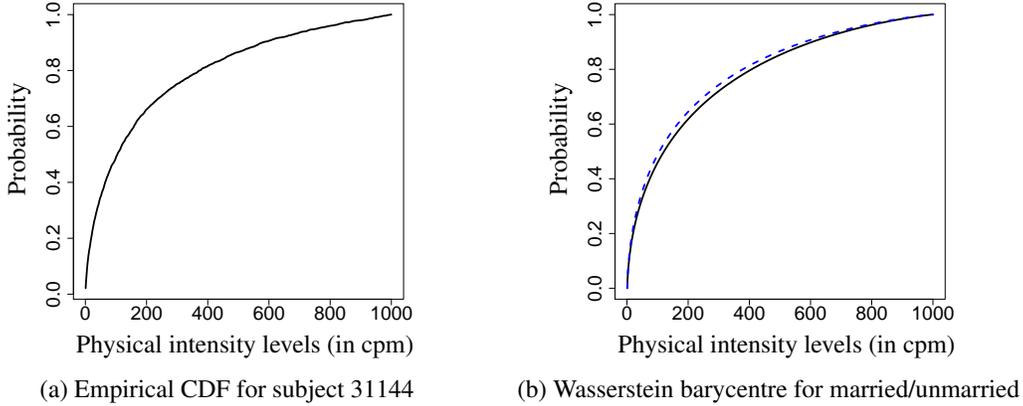
\begin{figure}
		\begin{center}
			\begin{tikzpicture}[scale=1.2, every node/.style={scale=1.2}]
				\newcommand\x{0.8}
				\node at (1+\x,0) {\includegraphics[scale=0.25]{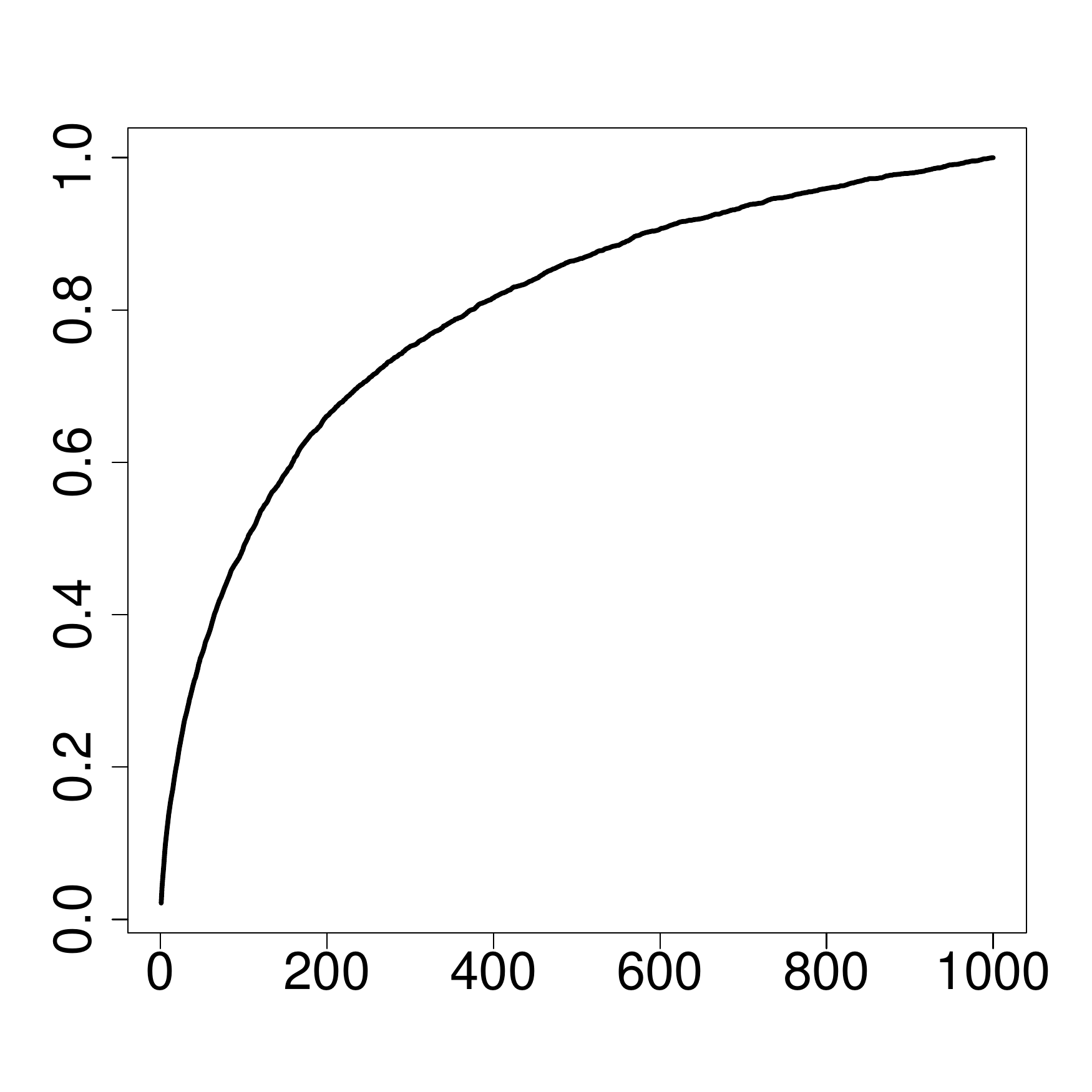}};
				\node at (7+\x,0) {\includegraphics[scale=0.25]{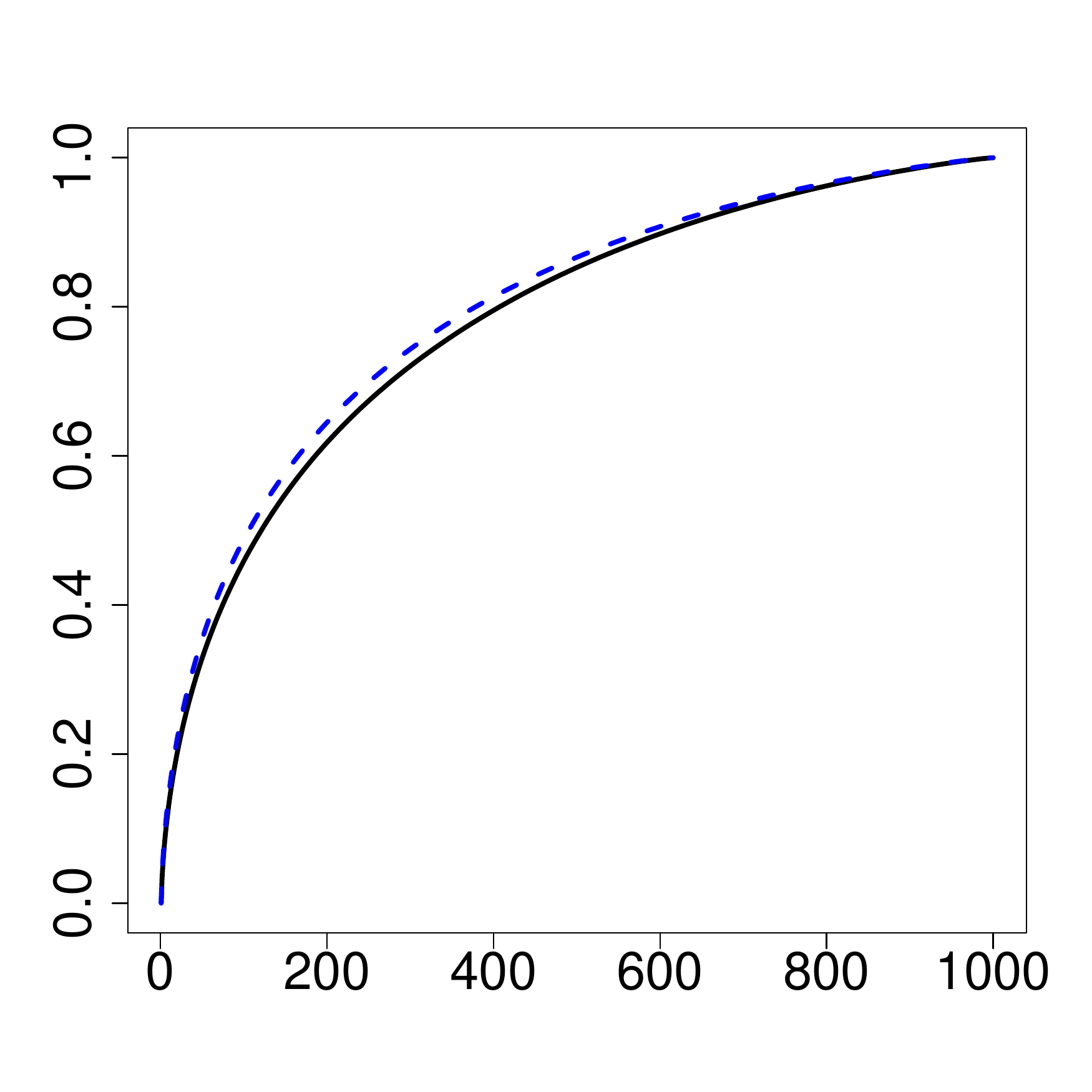}};
				
				\node[rotate=90] at (-0.5,0.2) {\scriptsize Probability};
				\node[rotate=90] at (5.5,0.2) {\scriptsize Probability};
				\node at (2,-2.1) {\scriptsize Physical intensity levels (in cpm)};
				\node at (8,-2.1) {\scriptsize Physical intensity levels (in cpm)};
				\node at (1.8,-2.6) {\scriptsize (a) Empirical CDF for subject 31144};
				\node at (7.8,-2.6) {\scriptsize (b) Wasserstein barycentre for married/unmarried};
			\end{tikzpicture} 
		\end{center}
		\vspace{-0.2in}
		\caption{Panel (a) plots the preprocessed data for a randomly selected subject, where we transform the raw trajectory data to a cumulative distribution function in $\mathcal{W}_2([1,1000])$. Panel (b) plots the Wasserstein barycentre of the empirical cumulative distribution functions in the married group (black solid line) and unmarried group (blue dashed line). 
		}
		\label{real_f_1}
	\end{figure}
	
	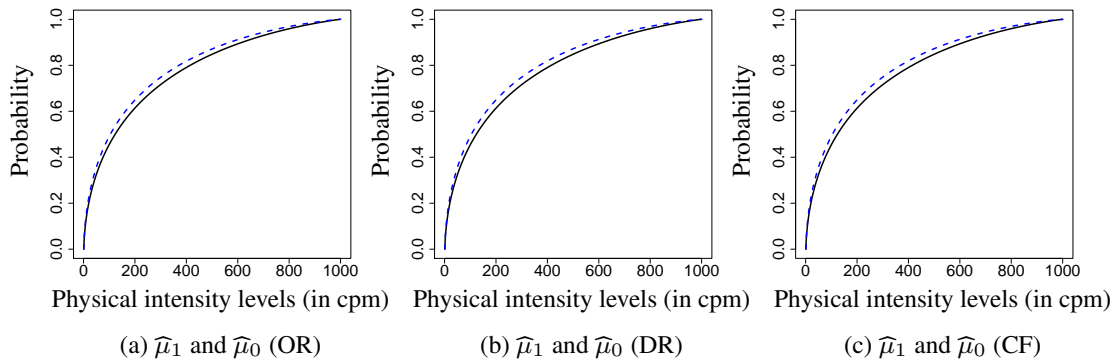
\begin{figure}
		\begin{center}
			\begin{tikzpicture}[scale=1.2, every node/.style={scale=1.2}]
				\newcommand\x{0.8}
				\node at (0+\x,0) {\includegraphics[scale=0.21]{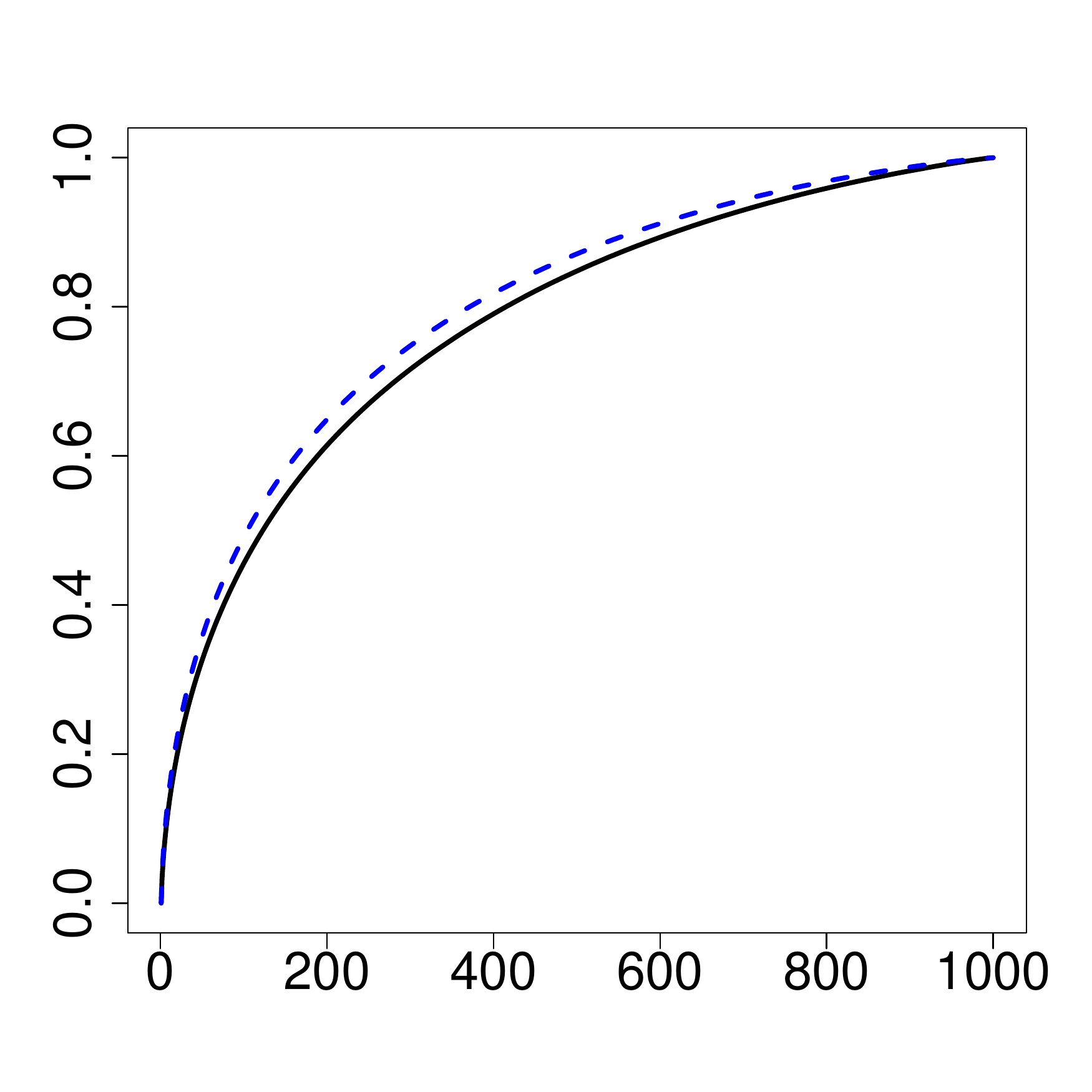}};
				\node at (4+\x,0) {\includegraphics[scale=0.21]{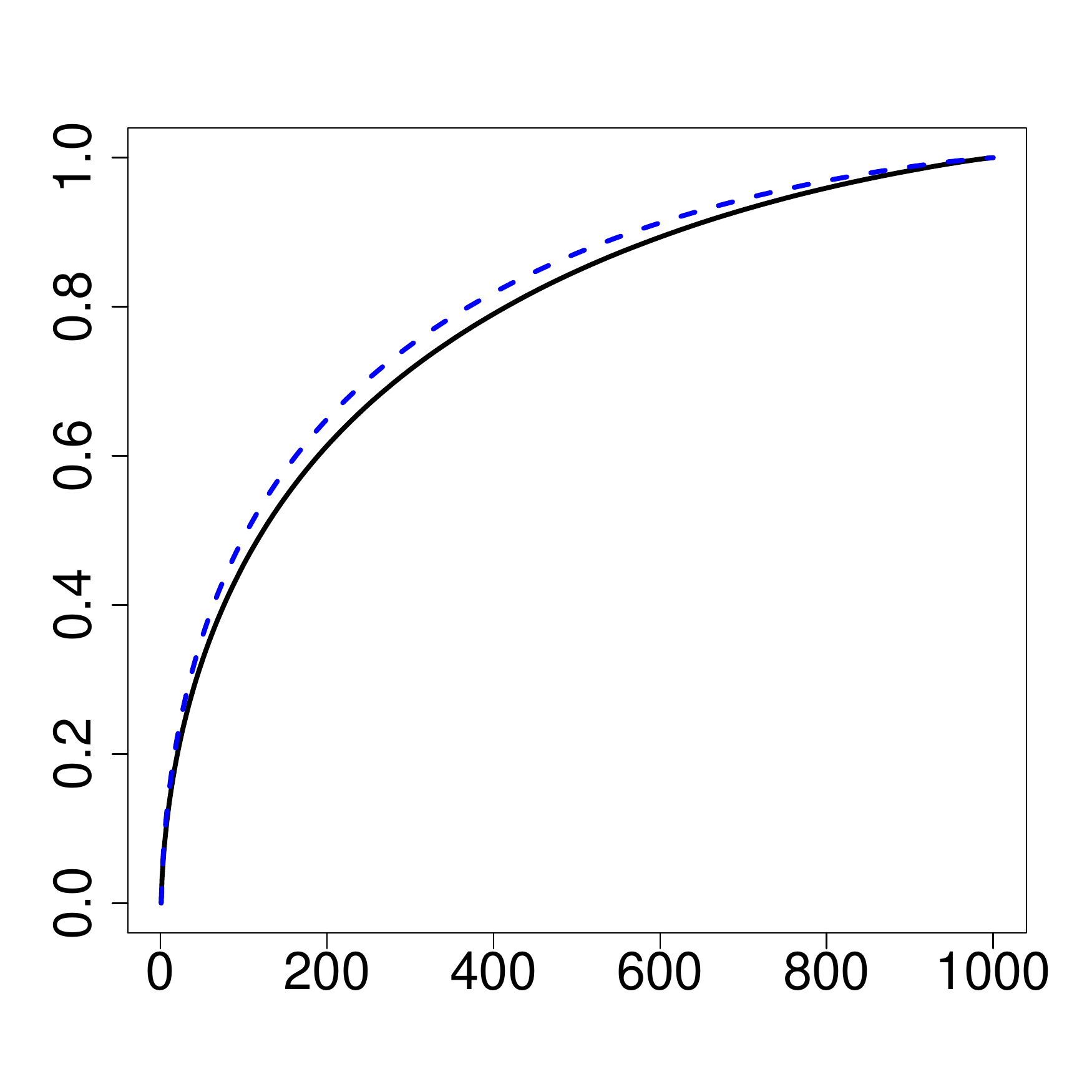}};
				\node at (8+\x,0)
				{\includegraphics[scale=0.21]{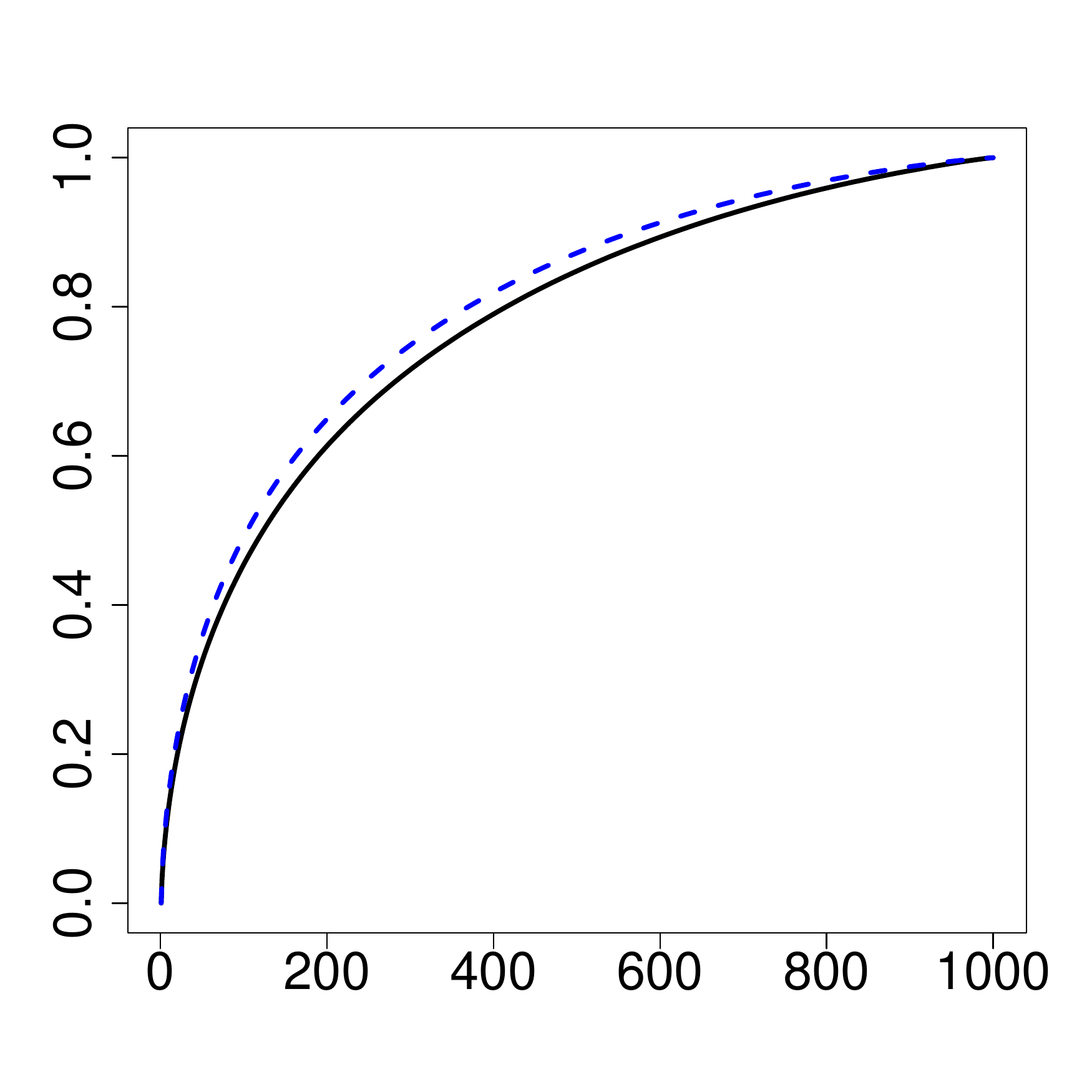}};
				\node[rotate=90] at (-1.2,0.2) {\scriptsize Probability};
				\node[rotate=90] at (2.8,0.2) {\scriptsize Probability};
				\node[rotate=90] at (6.8,0.2) {\scriptsize Probability};
				\node at (1,-1.8) {\scriptsize Physical intensity levels (in cpm)};
				\node at (5,-1.8) {\scriptsize Physical intensity levels (in cpm)};
				\node at (9,-1.8) {\scriptsize Physical intensity levels (in cpm)};
				\node at (1,-2.3) {\scriptsize (a) $\widehat\mu_1$ and $ \widehat\mu_0$ (OR) };
				\node at (5,-2.3) {\scriptsize (b) $\widehat\mu_1$ and $ \widehat\mu_0 $ (DR)};
				\node at (9,-2.3) {\scriptsize (c) $\widehat\mu_1$ and $ \widehat\mu_0 $ (CF)};
			\end{tikzpicture} 
		\end{center}
		\vspace{-0.2in}
		\caption{Panel (a), (b) and (c) plots $\widehat\mu_1$ (black solid line) and $\widehat\mu_0 $ (blue dashed line) estimated by OR, DR and CF, respectively. }
		\label{real_f_2}
	\end{figure}

		We then present a finer analysis of these data with the proposed approaches. 
	    We first plot the estimates of causal Wasserstein barycentres, $\widehat{\mu}_1$ and $\widehat{\mu}_0$ by OR, DR, and CF in Figure \ref{real_f_2}. Here, for $a=0,1,$ $\widehat\mu_a$ is the corresponding cumulative distribution function of the estimate $\widehat\mu_a^{-1}=\widehat\mu^{-1,\widehat\lambda}_a\circ\widehat\lambda^{-1}$ for $\mu_a^{-1}$,  where for the OR method, 
	    $\widehat\mu_a^{-1,\widehat\lambda} = \mathbb{P}_n \widehat{m}_a^{\widehat{\lambda}}(X),$ for the DR method,  $\widehat\mu_a^{-1,\widehat\lambda}$ is defined in \eqref{eq:dr}, and for the CF method, in the spirit of \eqref{eq:CF-median}, {$\widehat\mu_a^{-1}= median\{\widehat\mu_{a,CF}^{-1,\widehat\lambda,r}\}_{r=1}^R$ with $\widehat\mu_{a,CF}^{-1,\widehat\lambda,r}$ being the estimator in \eqref{eq:dr-cf} for the $r$th partitioning.}  	With these estimators, the former is stochastically greater than the latter, suggesting that  marriage improves the entire distribution of physical intensity  level at every quantile.

	    To quantify the size of the treatment effect, we take the difference between $\widehat{\mu}_1^{-1}$ and $\widehat{\mu}_0^{-1}$, corresponding to the difference in quantiles of the mean potential outcomes $\mu_1$ and $\mu_0$. 
	    To quantify the uncertainty of these estimates, we plot the corresponding $ 95\% $ confidence bands in Figure \ref{real_f_3} (b)--(d), where the confidence bands for the doubly robust and cross-fitting estimators were obtained using our asymptotic results presented in Theorems \ref{thm:AN-W-fixed-ref} {and \ref{thm:AN-W-fixed-ref-cf}}, 
	    and the confidence band for the outcome regression estimator was obtained using the conventional linear regression confidence interval for the slope coefficient corresponding to the exposure $A$. As expected, the estimation results from the three methods are very close to each other, and the OR method has the tightest confidence band among the three methods.

	    One can see from Figure \ref{real_f_3} that the effect of marriage on physical activity level is significant at the 0.05 level. 
	    One can also get the causal effect on individual quantiles from these plots. For example, 
	    according to the DR estimation results, on average, marriage improves the median  physical intensity level by 19.1 (95\% CI = [12.1, 26.1]) cpm. {Due to Theorem \ref{prop:equal}, this effect may be interpreted on both the population and individual levels. On the population level, this means that marriage improves the median of ``average'' (in the sense of Wasserstein barycentre) physical intensity level by 19.1 cpm. On the individual level, this means that the average improvement on the median physical intensity level is 19.1 cpm.} We also estimate the Wasserstein distance between $\hat{\mu}_1^{DR}, \hat{\mu}_0^{DR}$, $W_2(\hat{\mu}_1^{DR}, \hat{\mu}_0^{DR})=\|\hat{\Delta}^{\hat{\mu}_0^{DR}, DR}\|_{\hat{\mu}_0^{DR}} $, and the estimate is $27.6$ cpm (95\% CI: [24.0, 31.2]).

		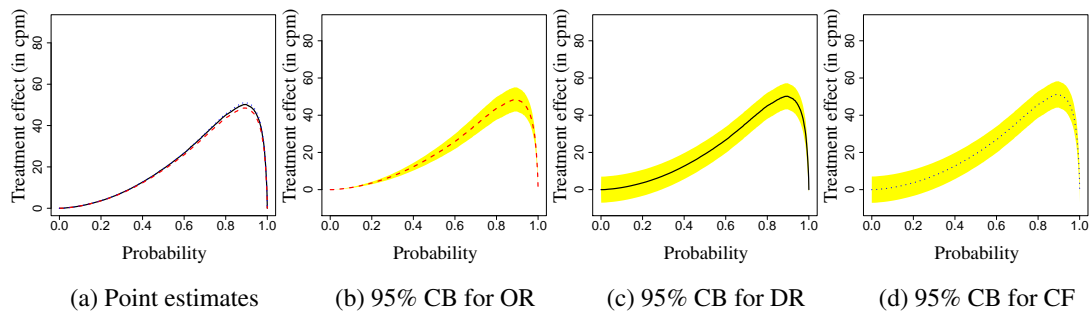
\begin{figure}
		\begin{center}
			\begin{tikzpicture}[scale=1.2, every node/.style={scale=1.2}]
				\newcommand\x{0.8}
				\node at (0+\x,0) {\includegraphics[scale=0.17]{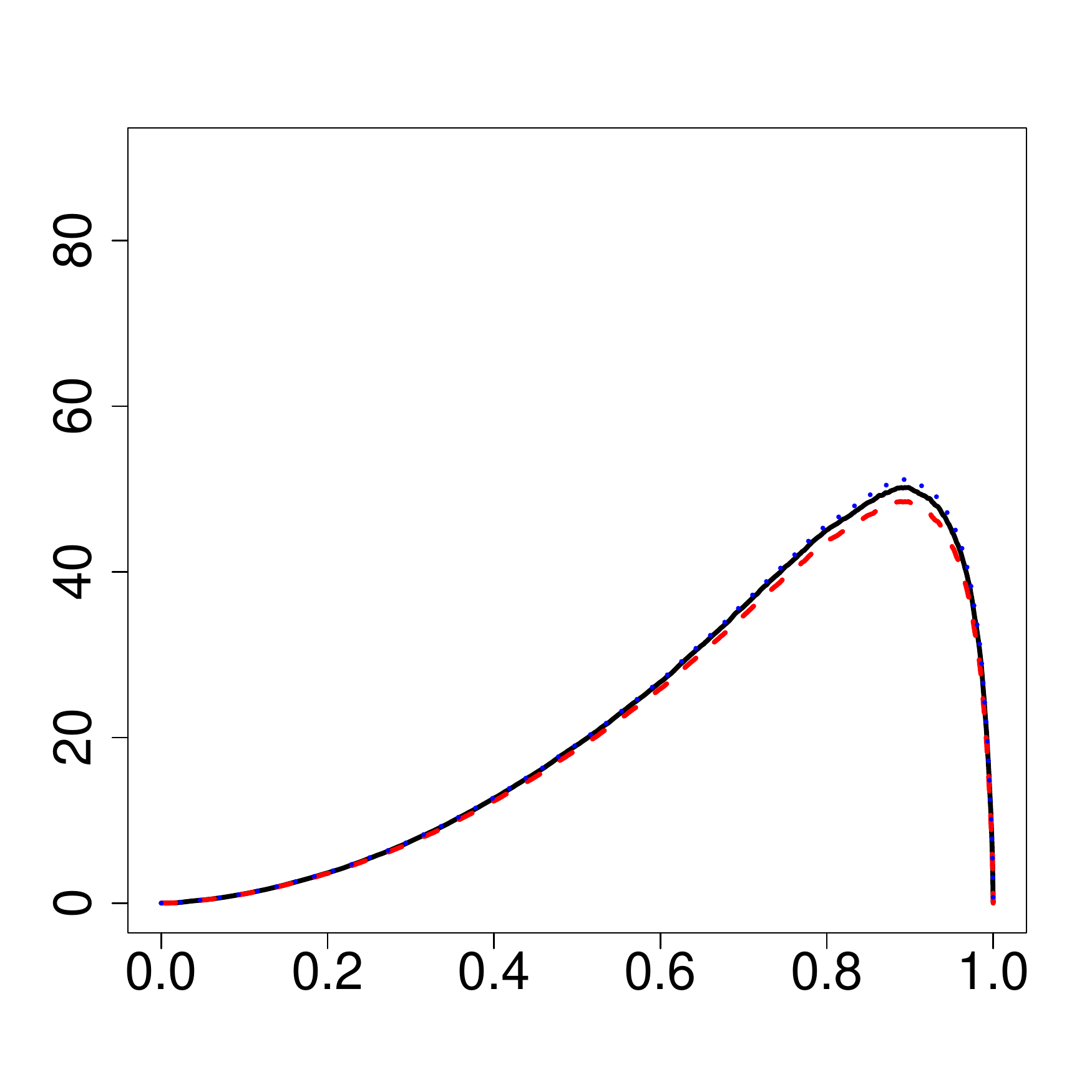}};
				\node at (3+\x,0) {\includegraphics[scale=0.17]{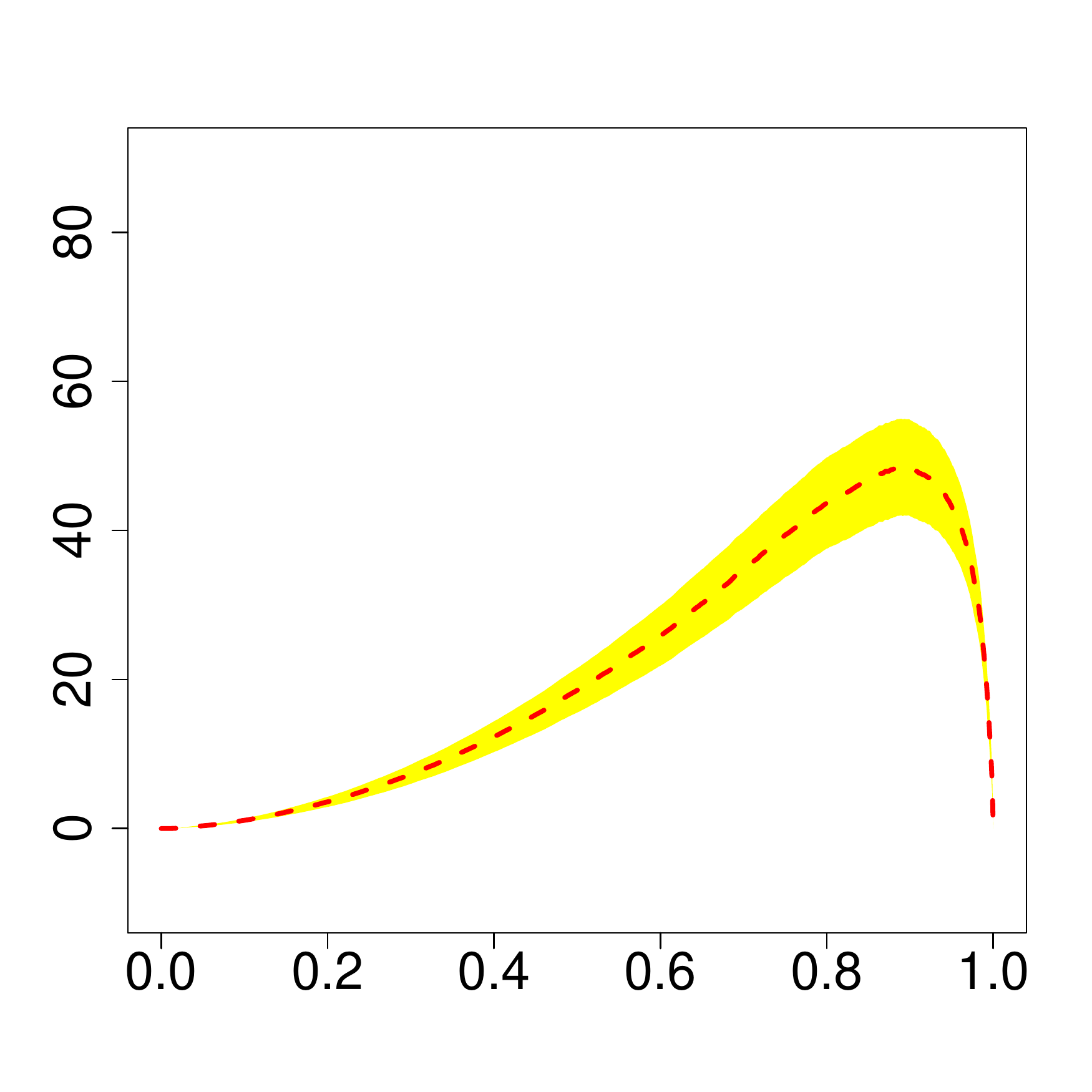}};
				\node at (6+\x,0) {\includegraphics[scale=0.17]{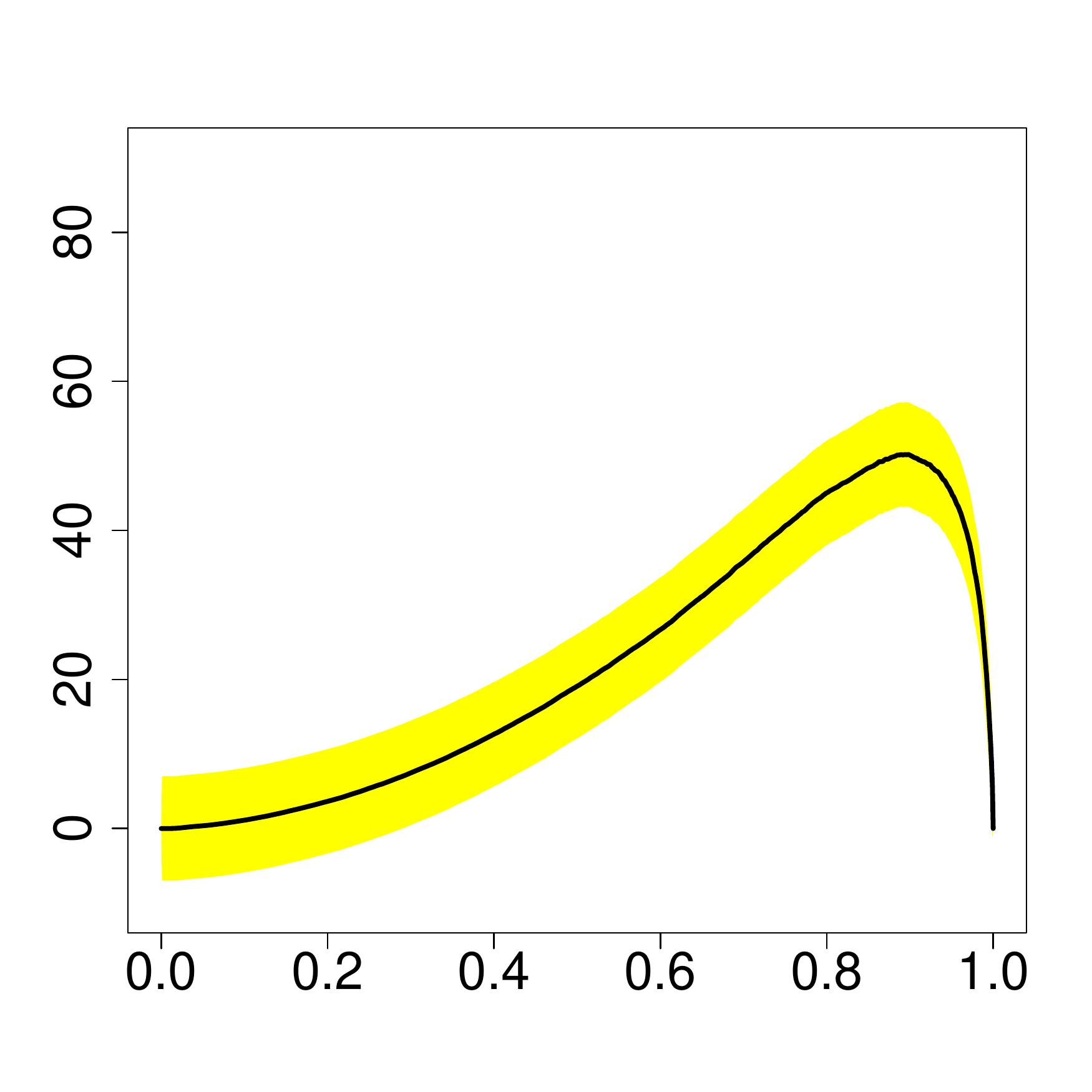}};
				\node at (9+\x,0) {\includegraphics[scale=0.17]{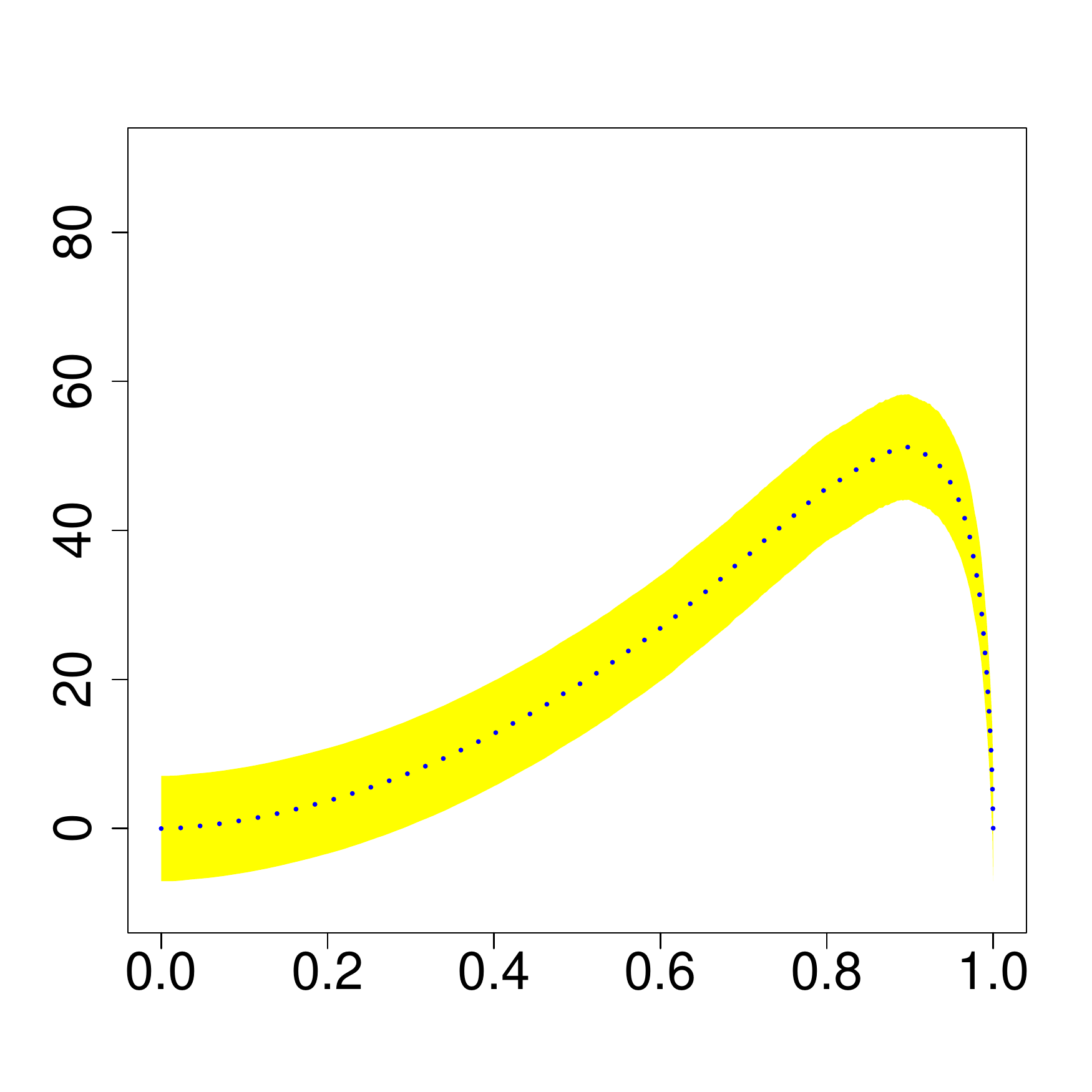}};
				\node[rotate=90] at (-0.7,0.15) {\tiny Treatment effect (in cpm)};
				\node[rotate=90] at (2.3,0.15) {\tiny Treatment effect (in cpm)};
				\node[rotate=90] at (5.3,0.15) {\tiny Treatment effect (in cpm)};
				\node[rotate=90] at (8.3,0.15) {\tiny Treatment effect (in cpm)};
				\node at (0.9,-1.5) {\tiny Probability};
				\node at (3.9,-1.5) {\tiny Probability};
				\node at (6.9,-1.5) {\tiny Probability};
				\node at (9.9,-1.5) {\tiny Probability};
				\node at (0.9,-2) {\scriptsize (a) Point estimates};
				\node at (3.9,-2) {\scriptsize (b) 95\% CB for OR};
				\node at (6.9,-2) {\scriptsize (c) 95\% CB for DR};
				\node at (9.9,-2) {\scriptsize (d) 95\% CB for CF};
			\end{tikzpicture} 
		\end{center}
		\vspace{-0.2in}
		\caption{Difference in quantiles estimates with the OR (red dashed), DR (black solid) and CF (blue dotted) estimators:  point estimates and 95\% confidence bands. }
		\label{real_f_3}
	\end{figure}

One may also be interested in predicting the unobserved potential outcome for a particular individual.
As an illustration, we estimate ${Y}_i(1)$ for Subject 31144, who was unmarried so $Y_{31144} = Y_{31144}(0).$
Recall that   $\Delta_i^{Y_i(0)} = T_i-\idf $. We first estimate the individual causal effect map for Subject 31144 using the average causal effect map with reference distribution  $\widehat{Y}_{31144}:$ $\widehat{T}_{31144}=\widehat{\Delta}^{\widehat{Y}_{31144}}+\idf,$ {the latter being estimated using the DR method.} We then apply the individual causal transport map to his empirical CDF  to obtain $\widehat{Y}_{31144}(1)$ plotted in Figure \ref{real_f_4}(b). From this, one may obtain, for example, getting married would raise his mean physical activity from $144.5$ cpm to $166.3$ cpm (95\% CI: [159.5, 173.7]).

		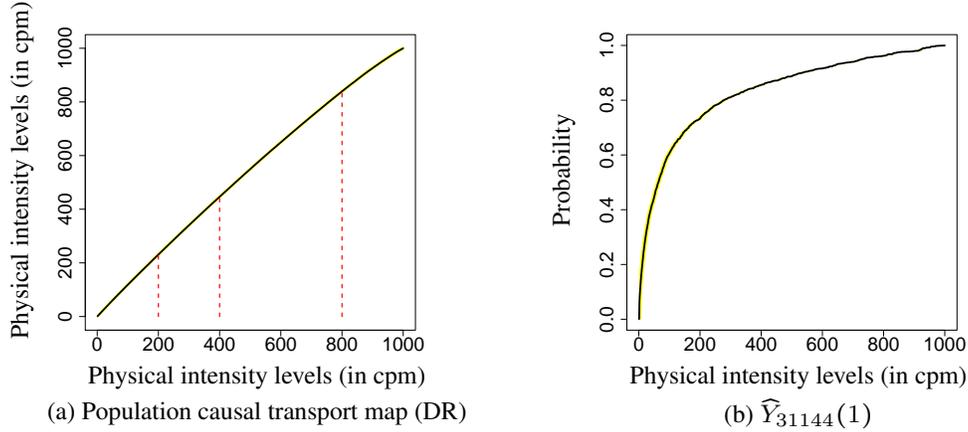
\begin{figure}
		\begin{center}
			\begin{tikzpicture}[scale=1.2, every node/.style={scale=1.2}]
				\newcommand\x{0.8}
				\node at (2+\x,0) {\includegraphics[scale=0.25]{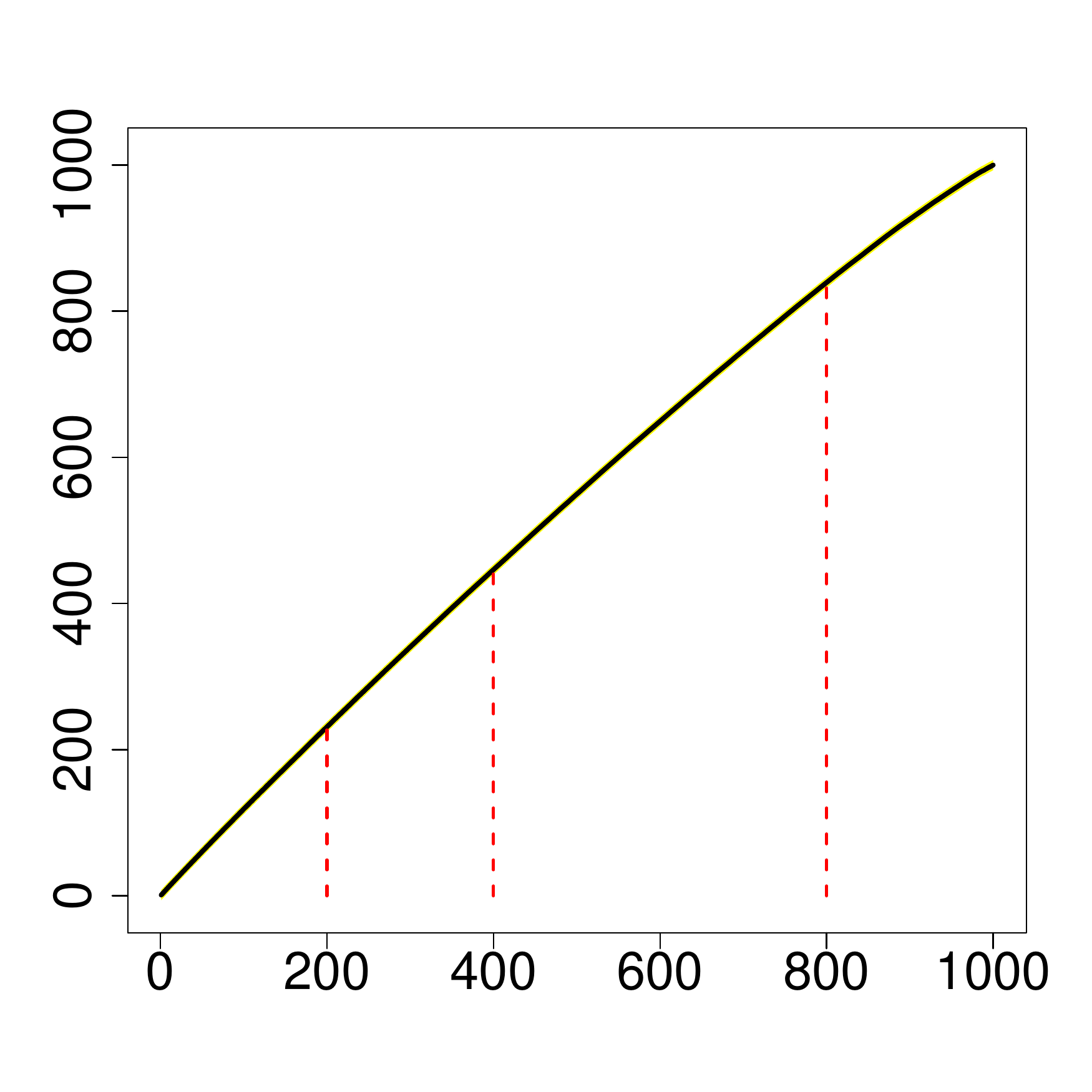}};
				\node at (8+\x,0) {\includegraphics[scale=0.25]{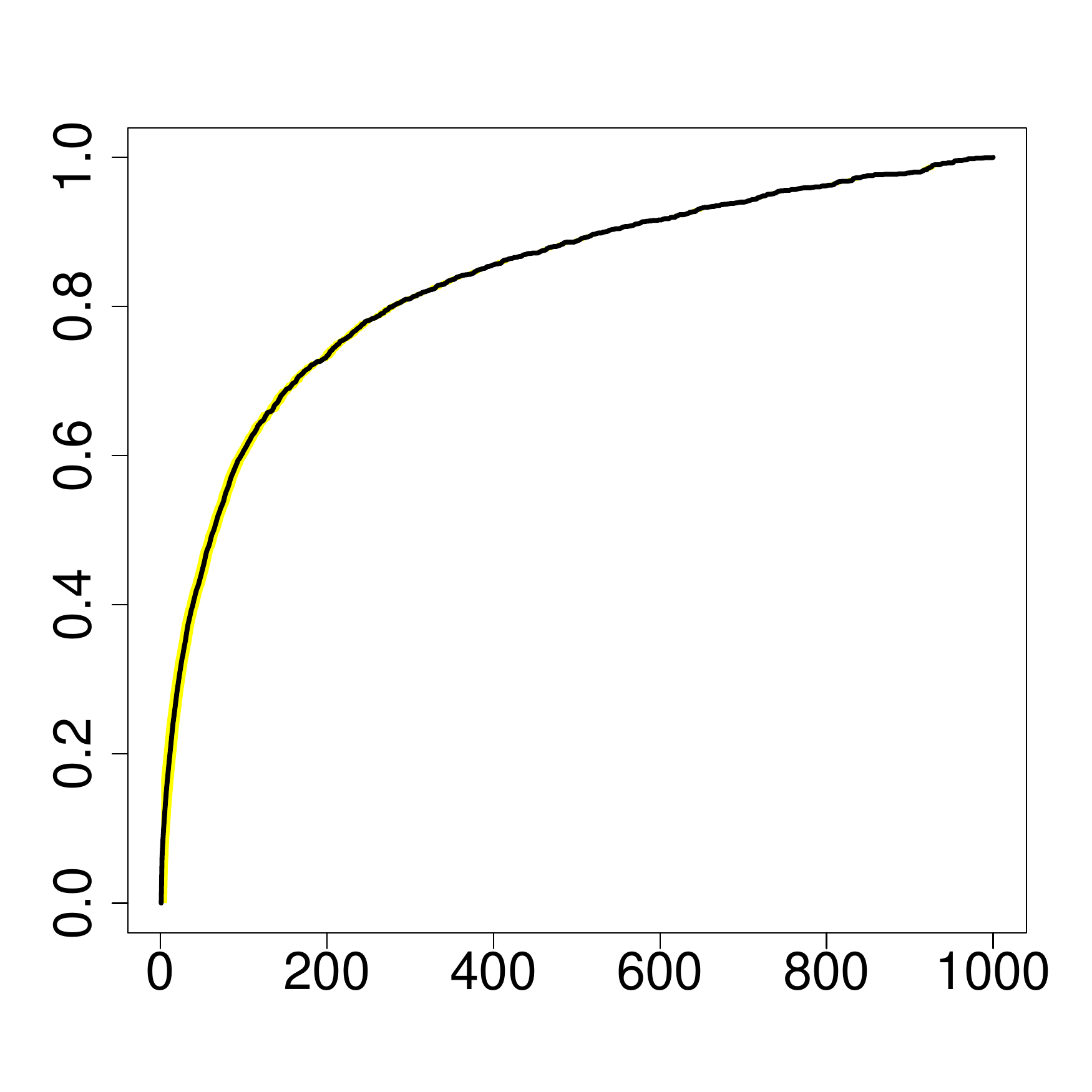}};
				\node[rotate=90] at (0.4,0.2) {\scriptsize Physical intensity levels (in cpm)};
				\node[rotate=90] at (6.4,0.2) {\scriptsize Probability};
				\node at (9,-2.1) {\scriptsize Physical intensity levels (in cpm)};
				\node at (3,-2.1) {\scriptsize Physical intensity levels (in cpm)};
				\node at (3,-2.5) {\scriptsize (a) Population causal transport map (DR)};
				\node at (9,-2.5) {\scriptsize (b) $\widehat{Y}_{31144}(1)$};
			\end{tikzpicture} 
		\end{center}
		\vspace{-0.2in}
		\caption{Panel (a) plots {the DR estimate of} population causal transport map, superimposed with the 95\% confidence band; {Panel (b) plots the DR estimate of counterfactual outcome $Y(1)$ with 95\% confidence band for individual 31144, who was unmarried so that $Y_{31144} = Y_{31144}(0)$.}}
		\label{real_f_4}
	\end{figure}

We also compare our adjusted estimates with the results where we do not adjust for the observed confounders age and gender. In particular, we apply the OR,  DR, and CF estimators for estimating the average treatment effect $\Delta^\lambda$. We plot these estimates and the corresponding $ 95\% $ confidence bands in Figure \ref{noconfoundertxeffect}. One can see the treatment effect is attenuated without adjusting for   age and gender.

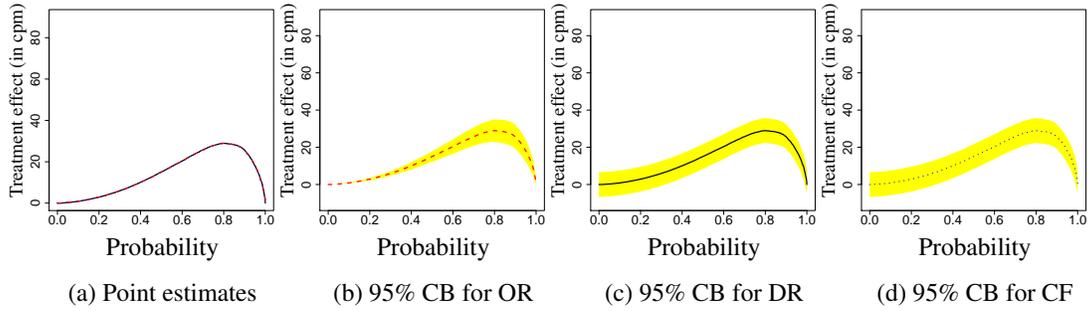
\begin{figure}[h!]
	\begin{center}
		\begin{tikzpicture}[scale=1.2, every node/.style={scale=1.2}]
			\newcommand\x{0.8}
			\node at (0+\x,0) {\includegraphics[scale=0.17]{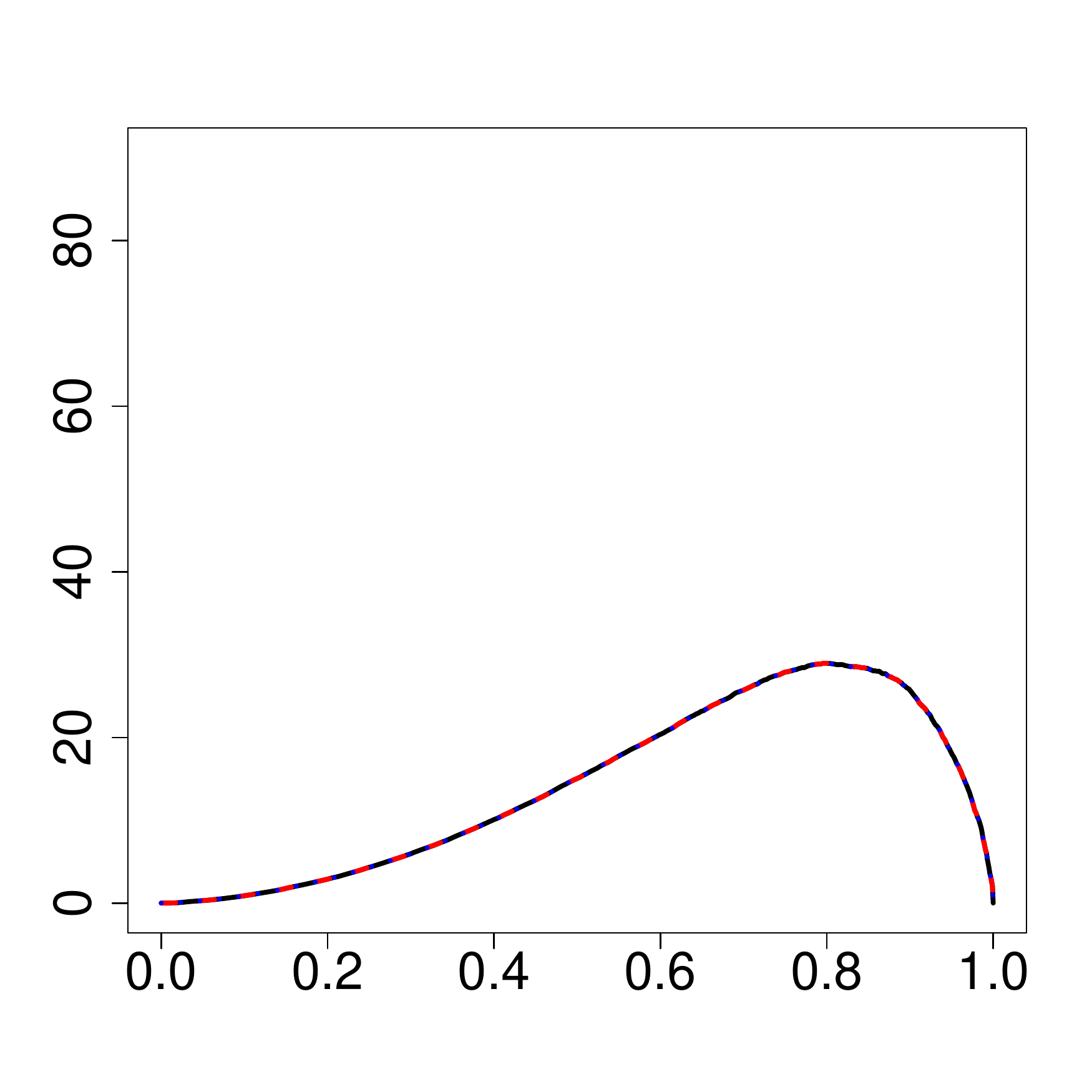}};
			\node at (3+\x,0) {\includegraphics[scale=0.17]{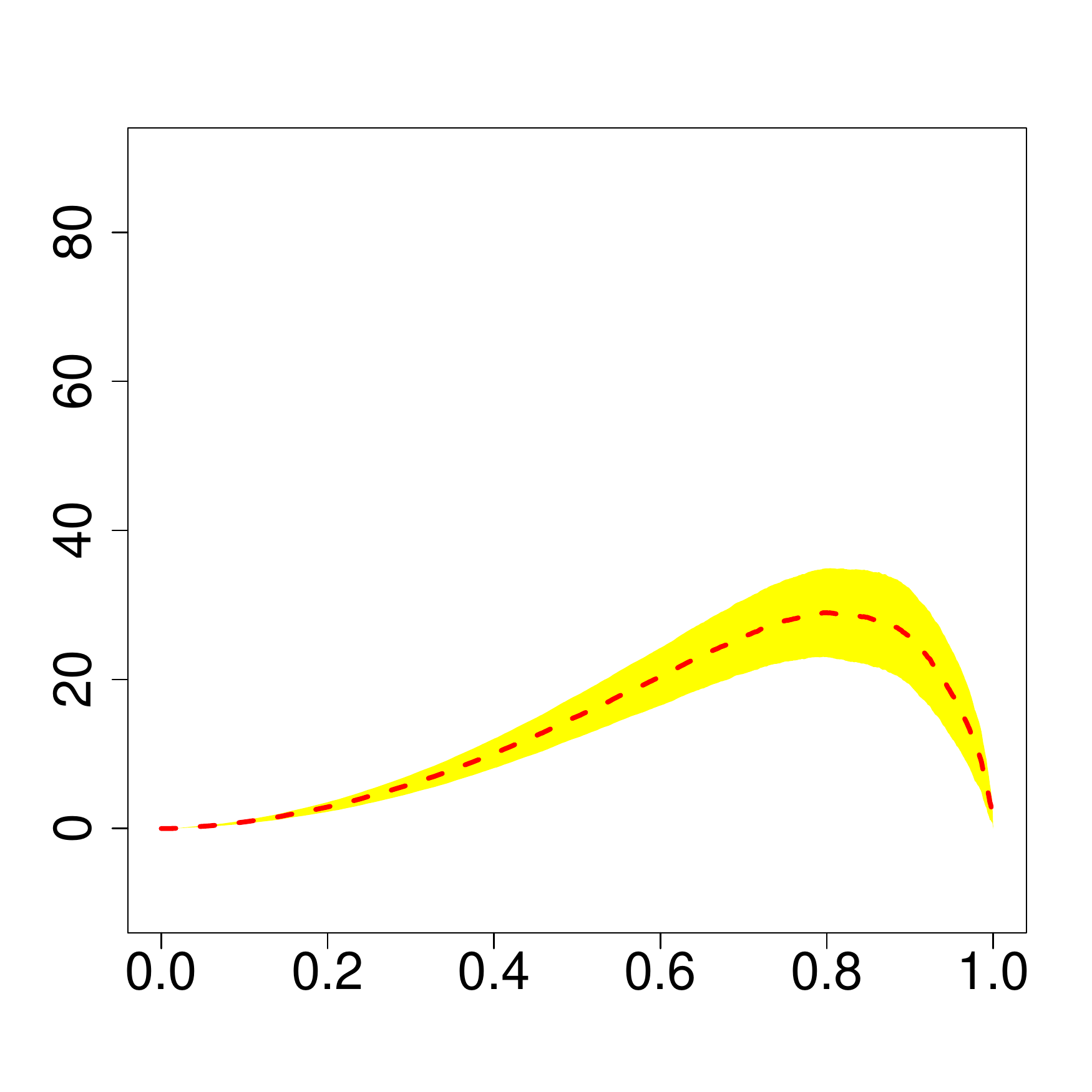}};
			\node at (6+\x,0) {\includegraphics[scale=0.17]{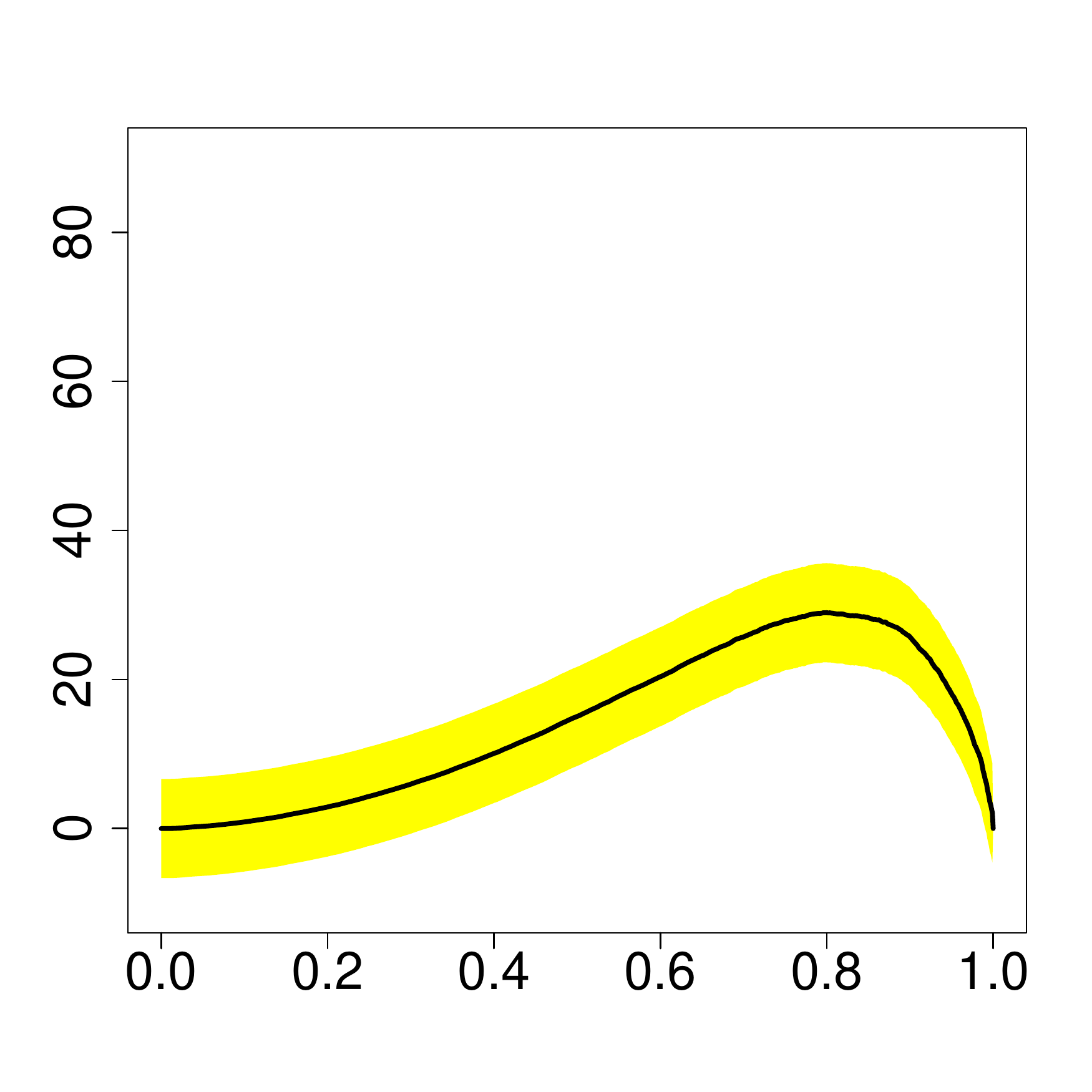}};
			\node at (9+\x,0) {\includegraphics[scale=0.17]{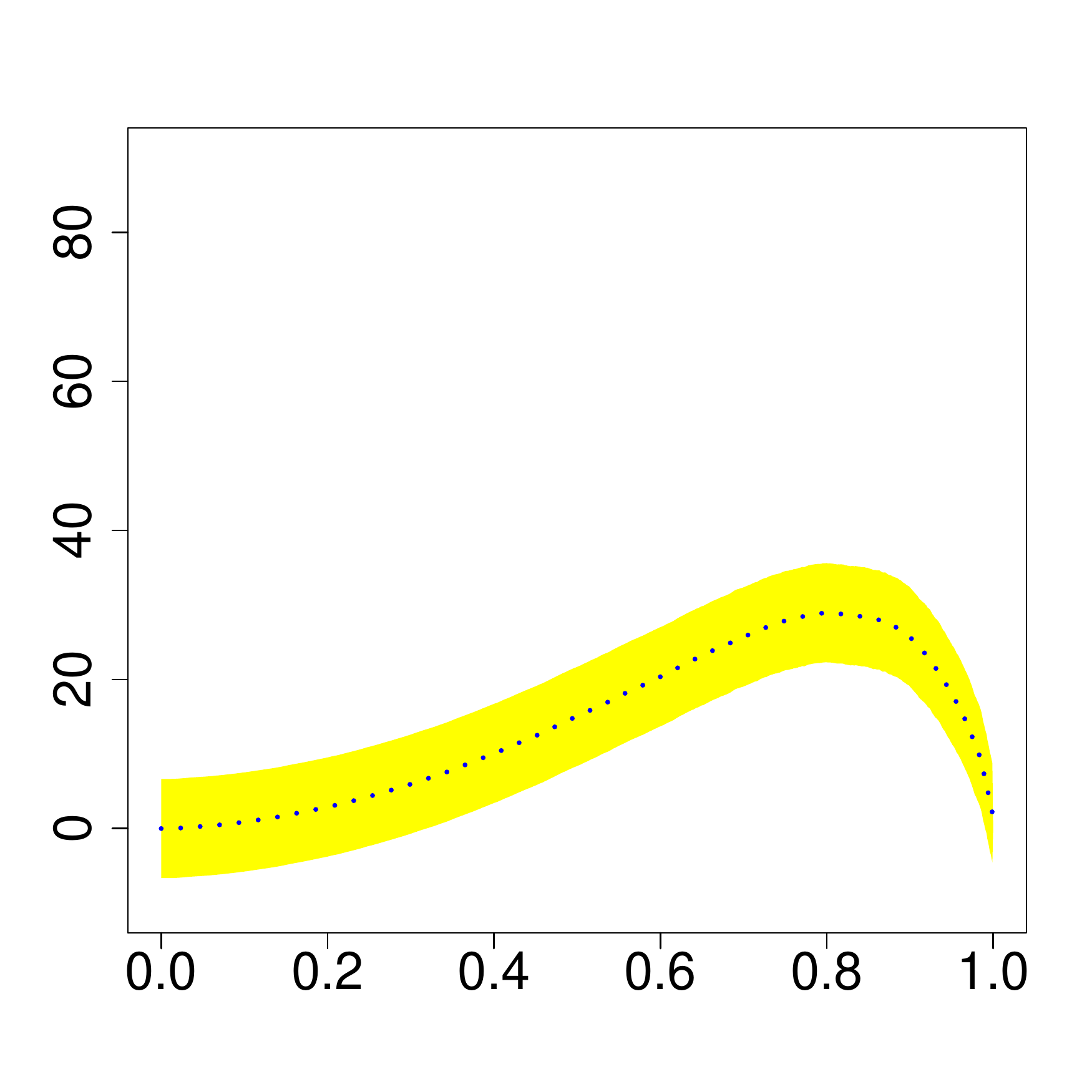}};
			\node[rotate=90] at (-0.7,0.15) {\tiny Treatment effect (in cpm)};
			\node[rotate=90] at (2.3,0.15) {\tiny  Treatment effect (in cpm)};
			\node[rotate=90] at (5.3,0.15) {\tiny Treatment effect (in cpm)};
			\node[rotate=90] at (8.3,0.15) {\tiny Treatment effect (in cpm)};
			\node at (0.9,-1.5) {\scriptsize Probability};
			\node at (3.9,-1.5) {\scriptsize Probability};
			\node at (6.9,-1.5) {\scriptsize Probability};
			\node at (9.9,-1.5) {\scriptsize Probability};
			\node at (0.9,-2) {\scriptsize (a) Point estimates};
			\node at (3.9,-2) {\scriptsize (b) 95\% CB for OR};
			\node at (6.9,-2) {\scriptsize (c) 95\% CB for DR};
			\node at (9.9,-2) {\scriptsize (d) 95\% CB for CF};
		\end{tikzpicture} 
	\end{center}
	\vspace{-0.2in}
	\caption{Difference in quantiles estimates with the OR (red dashed), DR (black solid) and CF (blue dotted) estimators when we do not adjust for observed confounders age and gender:  point estimates and 95\% confidence bands. 
}
	\label{noconfoundertxeffect}
\end{figure}

	\section{Discussion}
	\label{sec:discussion}
	
	In this paper, we study causal inference for distribution functions that reside in a Wasserstein space.  We propose novel definitions of causal effects and develop doubly robust estimation procedures for estimating these effects under the assumption of no unmeasured confounding. It would be interesting to extend classical causal inference methods for dealing with unmeasured confounding, such as the instrumental variable methods \citep[e.g.][]{ogburn2015doubly,wang2018bounded} to this setting. 	
	
	To the best of our knowledge, ours is the first formal study of causal effects for outcomes defined in a non-linear space. As such, we have only considered a leading special case of non-linear spaces. There are many other data objects from non-linear spaces that we do not consider in this paper. For example, the Wasserstein spaces of probability distributions on higher dimensional Euclidean spaces exhibit structures different from $\ws$ and thus pose new challenges for causal inference on such spaces.
	We also note that although the Wasserstein space $\ws$ is not a Riemannian manifold \citep{Bigot2017}, it can be endowed with a Riemannian structure, including the tangent space and Riemannian logarithmic map.
	In particular, let $cl(S)$ denote the closure of set $S$.  With a {continuous} reference distribution $\lambda,$  the space $T_\lambda \ws=cl\{k(\y^{-1}\circ\lambda-\idf):\y\in\ws, k\in \mR^+\}$ can be viewed as the tangent space of $\ws$ at $\lambda$, and the mapping  $\y \mapsto \y^{-1}\circ \lambda-\idf$    can be viewed as the Riemannian logarithmic map at $\lambda$ \citep{Ambrosio2004}.  	From this perspective, with the notation $\wwlog_\lambda \y=\y^{-1}\circ \lambda$, the individual causal effect maps can be written as
$
			\Delta_i^\lambda =   \wwlog_\lambda \Y_i(1)-\wwlog_\lambda \Y_i(0) = \{\wwlog_\lambda \Y_i(1)-\idf\} - \{\wwlog_\lambda \Y_i(0) - \idf\},
		$		so they  may be equivalently defined as the contrasts between the Riemannian logarithmic maps of distribution functions $Y_i(1)$ and  $Y_i(0)$. By Theorem \ref{prop:equal}, the average causal effect map may then be equivalently defined as
		$ \Delta^\lambda =\expect(\wwlog_\lambda Y(1) - \idf)-\expect(\wwlog_\lambda \Y(0) - \idf).$ 
	These connections allow one to extend the proposed definition of causal effect from random distributions to random elements residing on a Riemannian manifold; see \cite{srivastava2016functional} for concepts and tools of Riemannian manifolds that are relevant to statistics.

	Another interesting venue for future research is the study of efficiency theory with distribution-valued outcomes. 
	It is well-known that the classical doubly robust and cross-fitting estimators \citep{robins1994estimation,chernozhukov2018double} are both doubly robust and locally semiparametric efficient. In Theorems \ref{thm:AN-W-fixed-ref} and \ref{thm:AN-W-fixed-ref-cf} we 
	establish double robustness of our proposed doubly robust and cross-fitting estimators. On the other hand, to establish semiparametric efficiency of these proposed estimators, one needs to extend semiparametric efficiency theory to accommodate distribution-valued outcomes that reside in infinite-dimensional functional spaces.  This will be developed in a separate paper.
	
	\section*{Supplementary Material}\label{sec:supple}
	The supplementary file contains  some auxiliary results,  technical lemmas, and proofs for all the  theorems.
	 \texttt{R} code to reproduce the simulation studies and data analysis can be found in the repository 
	\url{https://github.com/kongdehanstat/causaldistributionfunction}. The data analyzed in Section \ref{sec:app} is available at \url{https://wwwn.cdc.gov/nchs/nhanes/ContinuousNhanes/Default.aspx?BeginYear=2005}.

	\bibliographystyle{asa}  
	\bibliography{causalmanifold,causal}

\clearpage

\begin{center}
	
	{\Large Supplementary Material for ``Causal Inference on Distribution Functions''}
	
	
	
	
\end{center}
\setcounter{equation}{0}
\setcounter{figure}{0}
\setcounter{table}{0}
\makeatletter
\renewcommand{\theequation}{S\arabic{equation}}
\renewcommand{\thefigure}{S\arabic{figure}}
\renewcommand{\thetable}{S\arabic{table}}
\setcounter{section}{0}

		\begin{abstract}
		{The supplementary file contains additional examples, some auxiliary results,  technical lemmas and proofs for all the  theorems.}
	\end{abstract}
	
\section{Additional Motivating Examples}\label{sec:add-ex}

\begin{example}[Cellular Differentiation]\label{ex:mp}
	In developmental biology, scientists are often interested in how an exposure influences the cellular differentiation process. In these studies, multiple samples of tissues 
	may be collected in the exposure and control groups. 
	For each sample, one randomly selects a set of cells and measures the  expression level for $p\geq 1$ genes in these cells. This process is  then repeated over a  period of time.
	To understand molecular programs related to cell differentiation, 
	in a high-profile work, \cite{Schiebinger2019} developed a so-called Waddington-OT framework based on the Wasserstein geometry. 
	Under this framework, at each time point $t$, a sample is represented by a p-dimensional distribution
 	$\mathscr P_t$ 
	of  gene expression level over a population of cells.  Typically these distributions are multi-modal, corresponding to different cell types in the samples.
Experiment results in \cite{Schiebinger2019} suggest that cellular differentiation follows the shortest path under the Wasserstein geometry. 
In other words, one may use the Wasserstein geometry to reconstruct the differentiation path $\mathscr P_t,  0\leq t \leq T$ with observations  at time $0$ and $T.$
	\end{example}

	\begin{example}[Metagenomics]\label{ex:mg}
	In microbial ecology, it is of interest to study whether an exposure changes the microbiome system in an environmental site (e.g., a human gut or acid mine drainage).  	
	To study this problem,   multiple samples of microorganisms are collected from the environmental site in both control and exposure cohorts. For each sample, scientists use shotgun sequencing to obtain DNA sequences, and map each of the DNA sequence onto a node of a reference phylogenetic tree \citep{vonMering2007}.  Consequently, a sample of microorganisms can be represented by an empirical distribution on a phylogenetic tree \citep{evans2012phylogenetic}. 
	Such a distribution encodes both relative gene abundance and taxonomic information, which together characterize a microbiome system.
	\end{example}

\section{Inference Based on Wasserstein Distance}\label{sec:S1}	
It was suggested by a reviewer to provide estimation and inference details also for the Wasserstein distance $W_2(\mu_1,\mu_0)$ due to its simplicity, even though this quantity does not satisfy the desiderata (d). For this, by simple calculation or Lemma \ref{lem:dist-norm-W},  we note that $W_2(\mu_1,\mu_0)=\|\Delta^\lambda\|_\lambda$ for any fixed continuous reference distribution $\lambda$. Therefore, it is intuitive to estimate $W_2(\mu_1,\mu_0)$ by $\|\hat\Delta_{DR}^\lambda\|_\lambda$ whose  asymptotic distribution is provided below.

According to Theorem \ref{thm:AN-W-fixed-ref}, $\sqrt n(\hat\Delta_{DR}^\lambda-\Delta_\lambda)$ converges weakly to a centered Gaussian process $G$ in $L^2(\ldomain;\lambda)$. Let $\phi(h)=\|h\|^2_\lambda$ for $h\in L^2(\ldomain;\lambda)$. It is seen that its Hadamard derivative $\phi^\prime_h$ \citep[Section 20.2,][]{vanderVaart1998} at $h$ is $\phi_h(g)=2\langle g,h \rangle_\lambda$ for $g\in L^2(\ldomain;\lambda)$. Then, according to Theorem 20.8 of \cite{vanderVaart1998}, $\sqrt n(\|\hat \Delta^\lambda_{DR}\|_\lambda^2-\|\Delta^\lambda\|_\lambda^2)$ converges weakly to $2\langle \Delta^\lambda, G\rangle_\lambda$ which is a centered Gaussian (real-valued) random variable when $\Delta_{DR}^\lambda\neq 0$. In case of $\Delta_{DR}^\lambda=0$,  by continuous mapping theorem, $n \|\hat \Delta^\lambda_{DR}\|_\lambda^2$ converges weakly to $\|G\|_\lambda^2$. Based on these results, we can also derive the asymptotic distribution of $\sqrt n(\|\hat \Delta^\lambda_{DR}\|_\lambda-\|\Delta^\lambda\|_\lambda)$ by applying the delta method or continuous mapping theorem again. For example,  $\sqrt n(\|\hat \Delta^\lambda_{DR}\|_\lambda-\|\Delta^\lambda\|_\lambda)$ converges weakly to the $N(0,\sigma^2\|\Delta^\lambda\|_\lambda^{-2})$ with $\sigma^2=\var(\langle \Delta^\lambda, G\rangle_\lambda)$ by the classic delta theorem when $\Delta^\lambda\neq 0$, and  converges to $\|G\|_\lambda$ weakly by the continuous mapping theorem when $\Delta^\lambda=0$.

The above results can be used to perform inference such as hypothesis tests on $W_2(\mu_1,\mu_0)$. For instance, to test the null hypothesis $W_2(\mu_1,\mu_0)=0$, which is equivalent to $\|\Delta^\lambda\|_\lambda=0$, we can use the test statistic $\sqrt n \|\hat \Delta_{DR}^\lambda\|_\lambda$, and reject the null hypothesis at the significance level $\alpha$ if it exceeds the $1-\alpha$ quantile of $\|G\|_\lambda$. Such quantile can be estimated via resampling, as follows. As in Remark \ref{rem:cov}, we can obtain an estimate $\hat C$ of the covariance function of $G$, and as in Remark \ref{rem:scb}, resample $B$ (e.g., $B=1000$) realizations $G_1,\ldots,G_B$ from the centered Gaussian process with the covariance function $\hat C$. For each realization we compute the norm $\|G_j\|_\lambda$ and finally estimate the $1-\alpha$  quantile by the empirical $1-\alpha$ quantile of $\|G_1\|_\lambda,\ldots,\|G_B\|_\lambda$.

\section{Remark on Condition Expectation in Assumption \ref{assu:eY}}\label{sec:S-remark}
The equation \eqref{eqn:2}  involves conditional expectation of a random variable $W_2^2(\eY_i,Y_i)$ given a random distribution $Y_i$. Such conditional expectation is well defined, as follows. Underlying all random quantities is a probability space $(\Omega,\mathscr E,P)$ with a sample space $\Omega$, an event space ($\sigma$-field) $\mathscr E$ and a probability measure $P$. Both $Y_i:\Omega\rightarrow \ws$ and $\widehat Y_i:\Omega\rightarrow \ws$ are  measurable maps taking values in $\ws$, while $W_2^2(\widehat Y_i,Y_i):\Omega\rightarrow \real$ is a real-valued measurable map. Note that $\expect\{W_2^2(\widehat Y_i,Y_i)|Y_i\}$ is understood to be the conditional expectation   $\expect\{W_2^2(\widehat Y_i,Y_i)|\sigma(Y_i)\}$ of the real random variable $W_2^2(\widehat Y_i,Y_i)$ given the $\sigma$-field $\sigma(Y_i)$, where $\sigma(Y_i)\subset \mathscr E$ is the smallest sub-$\sigma$-field that makes $Y_i$ measurable. By the definition of conditional expectation given a sub-$\sigma$-field, $\expect\{W_2^2(\widehat Y_i,Y_i)|\sigma(Y_i)\}$ is a measurable function, and further by the Doob--Dynkin lemma, there exists a measurable function $g:\ws\rightarrow \real$ such that $\expect\{W_2^2(\widehat Y_i,Y_i)|Y_i\}=\expect\{W_2^2(\widehat Y_i,Y_i)|\sigma(Y_i)\}=g(Y_i)$. In the equation \eqref{eqn:2},   $\expect\{W_2^2(\widehat Y_i,Y_i)|Y_i=\y\}$ represents $g(\y)$, and thus is well defined.

\section{Proof of \Cref{thm:AN-W-fixed-ref}}

To simplify notation and unify the proofs, for probability distributions $\lambda,\nu$, write $\wwlog_\lambda \nu=\nu^{-1}\circ \lambda$ and 
\begin{equation}\label{eq:w-pt}
	\pt_{\lambda}^\nu g=g\circ \lambda^{-1}\circ\nu \quad\text{ for } g\in L^2(\ldomain;\lambda).
\end{equation}
Also, let $Z_i=\wwlog_\lambda \Y_i$, $\hat Z_i=\wwlog_{\hat \lambda} \eY_i$, $R_i=\pt_{\hat\lambda}^\lambda \hat Z_i-Z_i$, and $D_a(x)=\pt_{\hat\lambda}^\lambda \hat m_a^{\hat\lambda}(x)-\tilde m_a^\lambda(x)$. The quantity $R_i$ can be viewed as the residual due to the discrepancy between $\hat\lambda$ and $\lambda$, and between $\eY_i$ and $\Y_i$. 
Define \begin{align*}
	\psi_1 & := \expect\left[\textstyle\dfrac{A\wwlog_{\lambda} Y}{\pi(X)}-\left\{\textstyle\dfrac{A}{\pi(X)}-1\right\}m_1^\lambda(X)\right],\\
	\psi_0 & := \expect\left[\textstyle\dfrac{(1-A)\wwlog_{\lambda}Y}{1-\pi(X)}-\left\{\textstyle\dfrac{1-A}{1-\pi(X)}-1\right\} m_0^\lambda(X)\right]
\end{align*}
and their sample versions
\begin{align*}
	\hat\psi_1 & := \splavg\left[\textstyle\dfrac{A\wwlog_{\hat\lambda}\hat Y}{\hat\pi(X)}-\left\{\textstyle\dfrac{A}{\hat\pi(X)}-1\right\}\hat m_1^{\hat\lambda}(X)\right],\\
	\hat\psi_0 & := \splavg\left[\textstyle\dfrac{(1-A)\wwlog_{\hat\lambda}\hat Y}{1-\hat\pi(X)}-\left\{\textstyle\dfrac{1-A}{1-\hat\pi(X)}-1\right\}\hat m_0^{\hat\lambda}(X)\right].
\end{align*}
Then we have $\Delta^{\lambda}=\psi_1-\psi_0$ and $\hat\Delta_{DR}^{\hat\lambda}=\hat\psi_1-\hat\psi_0$. 
In the above and in what follows, when $\hat\lambda=\lambda$, the operator $\tau_{\hat\lambda}^\lambda$ is an identity operator and has no effect. The proof will based on the following decomposition for $\tau_{\hat\lambda}^{\lambda}\hat\psi_{1}$:
\begin{align*}
	\tau_{\hat\lambda}^{\lambda}\hat\psi_{1} -\psi_1 & = \splavg \left[ \textstyle\dfrac{A Z + AR}{\hat\pi(X)}-\left\{\textstyle\dfrac{A}{\hat\pi(X)}-1\right\}\{\tilde m_{1}^\lambda(X)+D_{1}(X)\} \right]  -\psi_1 \\
	& =\underbrace{(\splavg-\expect_{n}) \left[ \textstyle\dfrac{A\{Z-\tilde m_{1}^\lambda(X)\}}{\hat\pi(X)}+\tilde m_{1}^\lambda(X) -\textstyle\dfrac{A\{Z-m_{1}^{\lambda,\ast}(X)\}}{\pi^\ast(X)}-m_1^{\lambda,\ast}(X) \right]}_{\textup{I}} \nonumber\\
	& \,\,\,\,\,\,+ \underbrace{(\splavg-\expect_{n}) \left[\textstyle\dfrac{A\{Z-m_1^{\lambda,\ast}(X)\}}{\pi^\ast(X)}+m_1^{\lambda,\ast}(X) \right]}_{\textup{II}} \nonumber \\
	& \,\,\,\,\,\,+  \underbrace{\expect_{n} \left[ \textstyle\dfrac{\{\tilde m_{1}^\lambda(X)-m_1^\lambda(X)\}\{\hat{\pi}(X)-A\}}{\hat\pi(X)} \right]}_{\textup{III}} \\
	& \,\,\,\,\,\, +  \underbrace{\mathbb P_{n} \left[ \left\{1-\textstyle\dfrac{A}{\hat\pi(X)}\right\}D_{1}(X) \right]}_{\textup{IV}}\\
	& \,\,\,\,\,\, +  \underbrace{\mathbb P_{n} \left\{\textstyle\dfrac{ AR}{\hat\pi(X)} \right\}}_{\textup{V}}.
\end{align*}
Here, $\expect_n(O)=n^{-1}\sum_{i=1}^n \expect (O_i)$ for generic random quantities $O_1,\ldots,O_n$. 
The decomposition for the other term is similar and thus omitted. In the sequel, we use $c$ to denote a positive constant and allow its value to vary in different occurrences.  

\begin{proof}[Proof of part \ref{thm:AN-W-rate-fixed-ref}] This is a direct consequence of Claims \ref{claim:f:I}--\ref{claim:f:V} and the assumed rates of $\alpha_n$ and $\nu_n$.
\end{proof}
\begin{proof}[Proof of part \ref{thm:AN-W-AN-fixed-ref}]
	Under the assumed conditions, the terms $\textup{I}$ and $\textup{III}$--$\textup{V}$ are of order $\op(n^{-1/2})$. Consequently,
	\begin{equation}
		\sqrt n (\pt_{\hat\lambda}^\lambda\hat\psi_1-\psi_1) =  \sqrt{n}(\splavg-\expect) \left[\textstyle\dfrac{A\{Z-m_1^{\lambda,\ast}(X)\}}{\pi^{\ast}(X)}+m_1^{\lambda,\ast}(X)\right] +  \op(1). \label{eq:AN-W-a=1}
	\end{equation}
	Similar deviations for the case $a=0$ lead to 
	\begin{equation}\label{eq:AN-W-a=0}
		\sqrt n (\pt_{\hat\lambda}^\lambda\hat\psi_0-\psi_0) =  \sqrt{n}(\splavg-\expect) \left[\textstyle\dfrac{(1-A)\{Z-m_0^{\lambda,\ast}(X)\}}{1-\pi^\ast(X)}+m_0^{\lambda,\ast}(X) \right]  +  \op(1).
	\end{equation}
	By combining \Cref{eq:AN-W-a=0,eq:AN-W-a=1}, the asymptotic normality of $\pt_{\hat\lambda}^\lambda\hat\Delta_{DR}^{\hat\lambda}-\Delta^\lambda$ follows from a central limit theorem and  Slutsky's lemma, with the fact that $\ws$ has a bounded diameter (since $\tdomain$ is assumed to be a bounded interval of $\real$) and thus $Z$, $m_0^{\lambda,\ast}(X)$ and  $m_1^{\lambda,\ast}(X)$ have finite variance.
\end{proof}

\begin{claim}\label{claim:f:I} 
	$\textup{I}=\op(n^{-1/2})$. 
\end{claim}
This claim is due to Assumptions \ref{assu:DR-std}\ref{assu:DR-std-convergence-fixed-ref} and \ref{assu:DR-fixed-ref}\ref{assu:DR-donsker-fixed-ref}.

\begin{claim}\label{claim:f:II}
	$\textup{II}=\Op(n^{-1/2})$.
\end{claim}
This is a direct consequence of a central limit theorem, with the fact that $\ws$ has a bounded diameter and thus $Z$ and $m_1^{\lambda,\ast}(X)$ have finite variance.

\begin{claim}\label{claim:f:III}
	$\textup{III}=O\left(\varrho_\pi\varrho_m+n^{-1/2}\varrho_m+n^{-1/2}\varrho_\pi\right)+O( n^{-1/2}\varrho_m^{1/2})$.
\end{claim}
By Cauchy--Schwartz inequality, with Assumption \ref{assu:DR-std}\ref{assu:DR-std-pi-bound}, we have
\begin{align*}
	\textup{III} & =\left\|\expect_n \left[ \textstyle\dfrac{\{\tilde m_1^\lambda(X)-m_1^\lambda(X)\}\{\hat{\pi}(X)-A\}}{\hat\pi(X)}\right]\right\|_{\lambda} \\
	& \leq c\left\|\expect_n \left[ \{\tilde m_1^\lambda(X)-m_1^\lambda(X)\}\{\hat{\pi}(X)-A\}\right]\right\|_{\lambda} \\
	& \leq c\expect_n \| \{\tilde m_1^\lambda(X)-m_1^\lambda(X)\}\{\hat{\pi}(X)-\pi(X)\}\|_{\lambda} + c\left\|\expect_n \left[ \{\tilde m_1^\lambda(X)-m_1^\lambda(X)\}\{{\pi}(X)-A\}\right]\right\|_{\lambda} \\
	& \leq cn^{-1}\sum_{i=1}^n\sqrt{\expect\{ |\hat\pi(X_i)-\pi(X_i)|^2\}\expect\{\|\tilde m_1^\lambda(X_i)-m_1^\lambda(X_i)\|^2_{\lambda}\}}+ O( n^{-1/2}\varrho_m^{1/2})\\ 
	& = O\left(\varrho_\pi\varrho_m+n^{-1/2}\varrho_m+n^{-1/2}\varrho_\pi\right)+O( n^{-1/2}\varrho_m^{1/2}),
\end{align*}
where $c$ is a constant depending on the constant $\epsilon$ in Assumption \ref{assu:DR-std}\ref{assu:DR-std-pi-bound}, and the last equality is obtained by using Assumption \ref{assu:DR-fixed-ref}\ref{assu:DR-stability-fixed-ref}. In the above, the third inequality relies on the bound 
\begin{equation}\label{eq:IV}
	\left\|\expect_n \left[ \{\tilde m_1^\lambda(X)-m_1^\lambda(X)\}\{{\pi}(X)-A\}\right]\right\|_{\lambda}=O( n^{-1/2}\varrho_m^{1/2}),
\end{equation}
which we establish below. Let $V_i(g)=\{g(X_i)-m_1^\lambda(X_i)\}\{\pi(X_i)-A_i\}$, $\eta_i(g,h)=\|g(X_i)-h(X_i)\|_\lambda$, $\eta^2(g,h)=\frac{1}{n}\sum_{i=1}^n \eta_i^2(g,h)$, and $S_n(g)=\frac{1}{\sqrt n}\sum_{i=1}^n V_i(g)$. Then $S_n(m_1^\lambda)=0$, $\expect\{V_i(g)\}=0$ and $\expect\{S_n(g)\}=0$ for all $g$. In addition, $\|V_i(g)-V_i(h)\|_\lambda\leq |\pi(X_i)-A|\eta_i(g,h)\leq 2\eta_i(g,h)$. Then, according to Assumption \ref{assu:DR-fixed-ref}\ref{assu:DR-ep} and \Cref{thm:hilbertian-empirical}, holding $\mathbb X$ fixed, we deduce that, for some universal constants $c_0,b_0,r>0$,  $$\pr\left(\sup_{g\in B_r(m_1^\lambda)}\frac{\|S_n(g)\|_\lambda}{\eta(g,m_1^\lambda)^{1/2}}\geq b\sqrt{K}\,\bigg|\, \mathbb X\right)\leq \exp\big(-\frac{c_0bK}{r}\big)$$ holds for all $b\geq b_0$. This further implies that  $$\expect\left(\sup_{g\in B_r(m_1^\lambda)}\frac{\|S_n(g)\|_\lambda}{\eta(g,m_1^\lambda)^{1/2}}\,\bigg|\, \mathbb X\right)\leq c_1$$ for a fixed constant $c_1$ for all $\mathbb{X}$, and further,  $\expect\left(\|S_n(g)\|_\lambda\,|\, \mathbb X\right)\leq c_1 {\eta(g,m_1^\lambda)^{1/2}}$ for all $g\in B_r(m_1^\lambda)$. By assumption, $\tilde m_1^\lambda\in B_r(m_1^\lambda)$ almost surely, and thus $\expect\left(\|S_n(\tilde m_1^\lambda)\|_\lambda\,|\, \mathbb X\right)\leq c_1 {\eta(\tilde m_1^\lambda,m_1^\lambda)^{1/2}}$. Consequently,
\begin{align*}
	\left\|\expect_n \left[ \{\tilde m_1^\lambda(X)-m_1^\lambda(X)\}\{{\pi}(X)-A\}\right]\right\|_{\lambda} & = n^{-1/2}\|\expect S_n(\tilde m_1^\lambda)\|_\lambda \leq  n^{-1/2}\expect\left(\|S_n(\tilde m_1^\lambda)\|_\lambda\right) \\
	& \leq c_1  n^{-1/2}\expect\eta(\tilde m_1^\lambda,m_1^\lambda)^{1/2} \leq c_1  n^{-1/2} \{\expect\eta^2(\tilde m_1^\lambda,m_1^\lambda)\}^{1/4} \\
	& = O(n^{-1/2}\varrho_m^{1/2}).
\end{align*}

\begin{claim}\label{claim:f:IV}
	$\textup{IV} = \Op(n^{-1/2}\varrho_\pi+\varrho_\pi\nu_n+\varrho_\pi\alpha_n+n^{-1}) + \op(n^{-1/2})$.
\end{claim}
	We  first observe that
\begin{align*}
	\splavg\left[\left\{\textstyle\dfrac{A}{\hat\pi(X)}-1\right\}D_1(X) \right] & = \underbrace{\splavg\left[\left\{\textstyle\dfrac{A}{\pi(X)}-1\right\}D_1(X) \right]}_{\textup{IV}_1} + \underbrace{\splavg\left[\left\{\textstyle\dfrac{A\{\pi(X)-\hat\pi(X)\}}{\hat\pi(X)\pi(X)}\right\}D_1(X) \right]}_{\textup{IV}_2}.
\end{align*}
 The term  $\textup{IV}_1$ can be shown to have the order $\op(n^{-1/2})$ by an argument that is similar to the derivation of \eqref{eq:IV}. For the second term, we have
 \begin{align*}
 	\|\textup{IV}_2\|_\lambda & \leq \splavg \|\{\pi(X)-\hat\pi(X)\}D_1(X)\|_\lambda\\
 	& \leq \sqrt{\splavg |\pi(X)-\hat\pi(X)|^2 \splavg \|D_1(X)\|_\lambda^2   }\\
 	& = \Op\big((\varrho_\pi + n^{-1/2})(W_2(\hat\lambda,\lambda)+\nu_n+\alpha_n)\big)\\
 	& = \Op(n^{-1/2}\varrho_\pi+\varrho_\pi\nu_n+\varrho_\pi\alpha_n+n^{-1})
 \end{align*}
where, the first inequality is due to Assumptions \ref{assu:positivity} and \ref{assu:DR-std}\ref{assu:DR-std-pi-bound} on $\pi$ and $\hat\pi$, and the last two equalities are derived by using Assumptions  \ref{assu:DR-lambda} and \ref{assu:DR-D-rate-fixed-ref}, as well as the assumed rates of $\nu_n$ and $\alpha_n$.

\begin{claim}\label{claim:f:V}
$\textup{V}=\Op(\alpha_n+\nu_n)$.
\end{claim}
We observe that
\begin{align*}
	\splavg \left[ \textstyle\dfrac{ AR}{\hat\pi(X)}\right] & = \splavg \left[ \textstyle\dfrac{ AR}{\pi(X)}\right] + \splavg \left[ \textstyle\dfrac{ AR}{\hat\pi(X)}-\textstyle\dfrac{ AR}{\pi(X)}\right],
\end{align*}
where the second term is dominated by the first one. Moreover, 
\begin{equation}\label{eq:pf:AR-piX}
	\splavg \left[ \textstyle\dfrac{ AR}{\pi(X)}\right]  = \splavg \left[ \textstyle\dfrac{ AU}{\pi(X)}\right] + \splavg \left[ \textstyle\dfrac{ AV}{\pi(X)}\right],
\end{equation}
where $U_i=\pt_{\hat\lambda}^\lambda \rlog_{\hat \lambda} \Y_i-\rlog_\lambda \Y_i=0$ and $V_i=\pt_{\hat\lambda}^\lambda\rlog_{\hat\lambda} \eY_i-\pt_{\hat\lambda}^\lambda\rlog_{\hat\lambda} \Y_i=\rlog_\lambda \eY_i-\rlog_\lambda Y_i$. The claim is then proved by using Lemma \ref{lem:eY-Y}.

\begin{remark}\label{rem:S1}
In the paper, $\tdomain$ is assumed to be a {bounded} interval of $\real$, which implies that $\ws$ has a bounded diameter. This boundedness assumption, however, can be dropped if we require $\expect W_2^2(Y,y)<\infty$ for some $y\in\ws$ and $\vertiii{m_a^{\lambda,\ast}}_\lambda<\infty$ for $a=0,1$, so that Claim \ref{claim:f:II} remains valid.
\end{remark}

	\section{Proof of Theorem \ref{thm:AN-W-fixed-ref-cf}}
	For simplicity, we assume $K=2$; the general case can be proved in a similar fashion. 
	Let $\mathbb P_{n_k}O$ denote $n_{k}^{-1}\sum_{i\in\mathscr D_k}O(A_i,X_i,\eY_i)$ and $\expect_{n_k}O=n_k^{-1}\sum_{i\in\mathscr D_k} \expect \{O(A_i,X_i,\eY_i)\}$, where $O=O(A,X,\eY)$ is a random quantity dependent on $(A,X,\eY)$. Similarly, we use $\mathbb P_{n}O$  to denote $n^{-1}\sum_{i=1}^nO(A_i,X_i,Y_i)$.
	As in the previous section,  let $Z_i=\wwlog_\lambda \Y_i$, where we recall the notation $\wwlog_\lambda \nu=\nu^{-1}\circ \lambda$. If the $i$th subject belongs to the $k$ partition, then $\hat Z_i=\wwlog_{\hat \lambda_k} \eY_i$ and $R_i=\pt_{\hat\lambda_k}^\lambda \hat Z_i-Z_i$, and define  $D_{a,k}(x)=\pt_{\hat\lambda_k}^\lambda \hat m_{a,k}^{\hat\lambda_k}(x)-\tilde m_{a,k}^\lambda(x)$. The quantity $R_i$ can be viewed as the residual due to the discrepancy between $\hat\lambda_k$ and $\lambda$, and between $\eY_i$ and $\Y_i$. 
	Define \begin{align*}
		\psi_1 & := \expect\left[\textstyle\dfrac{A\wwlog_{\lambda} Y}{\pi(X)}-\left\{\textstyle\dfrac{A}{\pi(X)}-1\right\}m_1^\lambda(X)\right],\\
		\psi_0 & := \expect\left[\textstyle\dfrac{(1-A)\wwlog_{\lambda}Y}{1-\pi(X)}-\left\{\textstyle\dfrac{1-A}{1-\pi(X)}-1\right\} m_0^\lambda(X)\right]
	\end{align*}
	and their sample versions in each data partition $\mathscr D_k$
	\begin{align*}
		\hat\psi_{1,k} & := \mathbb P_{n_{3-k}}\left[\textstyle\dfrac{A\wwlog_{\hat\lambda_k}\hat Y}{\hat\pi_k(X)}-\left\{\textstyle\dfrac{A}{\hat\pi_k(X)}-1\right\}\hat m_{1,k}^{\hat\lambda_k}(X)\right],\\
		\hat\psi_{0,k} & := \mathbb P_{n_{3-k}}\left[\textstyle\dfrac{(1-A)\wwlog_{\hat\lambda_k}\hat Y}{1-\hat\pi_k(X)}-\left\{\textstyle\dfrac{1-A}{1-\hat\pi_k(X)}-1\right\}\hat m_{0,k}^{\hat\lambda_k}(X)\right].
	\end{align*}
	Then we have $\Delta^{\lambda}=\psi_1-\psi_0$ and $\hat\Delta_{DR}^{\hat\lambda}=n^{-1}(n_1\tau_{\hat\lambda_1}^{\hat\lambda}\hat\psi_{1,1}+n_2\tau_{\hat\lambda_2}^{\hat\lambda}\hat\psi_{1,2})-n^{-1}(n_1\tau_{\hat\lambda_1}^{\hat\lambda}\hat\psi_{0,1}+n_2\tau_{\hat\lambda_2}^{\hat\lambda}\hat\psi_{0,2})$ for the cross-fitting estimator defined in \eqref{eq:dr-cf}, and consequently, $$\tau_{\hat\lambda}^\lambda\hat\Delta_{DR}^{\hat\lambda}=n^{-1}(n_1\tau_{\hat\lambda_1}^{\lambda}\hat\psi_{1,1}+n_2\tau_{\hat\lambda_2}^{\lambda}\hat\psi_{1,2})-n^{-1}(n_1\tau_{\hat\lambda_1}^{\lambda}\hat\psi_{0,1}+n_2\tau_{\hat\lambda_2}^{\lambda}\hat\psi_{0,2}).$$
	In the above, when $\hat\lambda_k=\lambda$, the operator $\tau_{\hat\lambda_k}^\lambda$ is an identity operator and has no effect. The proof will based on the following decomposition for $n^{-1}(n_1\tau_{\hat\lambda_1}^{\lambda}\hat\psi_{1,1}+n_2\tau_{\hat\lambda_2}^{\lambda}\hat\psi_{1,2})$:
		\begin{align*}
			& n^{-1}(n_1\tau_{\hat\lambda_1}^{\lambda}\hat\psi_{1,1}+n_2\tau_{\hat\lambda_2}^{\lambda}\hat\psi_{1,2}) -\psi_1 \\
			& = n^{-1}\sum_{k=1,2}n_{3-k}\mathbb P_{n_{3-k}} \left[ \textstyle\dfrac{A Z + AR}{\hat\pi_k(X)}-\left\{\textstyle\dfrac{A}{\hat\pi_k(X)}-1\right\}\{\tilde m_{1,k}^\lambda(X)+D_{1,k}(X)\} \right]  -\psi_1 \\
			& = n^{-1}\sum_{k=1,2}n_{3-k}\underbrace{(\mathbb P_{n_{3-k}}-\expect_{n_{3-k}}) \left[ \textstyle\dfrac{A\{Z-\tilde m_{1,k}^\lambda(X)\}}{\hat\pi_k(X)}+\tilde m_{1,k}^\lambda(X) -\textstyle\dfrac{A\{Z-m_{1}^{\lambda,\ast}(X)\}}{\pi^\ast(X)}-m_1^{\lambda,\ast}(X) \right]}_{\textup{I}} \nonumber\\
			& \,\,\,\,\,\,+ n^{-1}\sum_{k=1,2}n_{3-k}\underbrace{(\mathbb P_{n_{3-k}}-\expect_{n_{3-k}}) \left[\textstyle\dfrac{A\{Z-m_1^{\lambda,\ast}(X)\}}{\pi^\ast(X)}+m_1^{\lambda,\ast}(X) \right]}_{\textup{II}} \nonumber \\
			& \,\,\,\,\,\,+  n^{-1}\sum_{k=1,2}n_{3-k}\underbrace{\expect_{n_{3-k}} \left[ \textstyle\dfrac{\{\tilde m_{1,k}^\lambda(X)-m_1^\lambda(X)\}\{\hat{\pi}_k(X)-A\}}{\hat\pi_k(X)} \right]}_{\textup{III}} \\
			 & \,\,\,\,\,\, +  n^{-1}\sum_{k=1,2}n_{3-k}\underbrace{\mathbb P_{n_{3-k}} \left[ \left\{1-\textstyle\dfrac{A}{\hat\pi_k(X)}\right\}D_{1,k}(X) \right]}_{\textup{IV}}\\
			 & \,\,\,\,\,\, +  n^{-1}\sum_{k=1,2}n_{3-k}\underbrace{\mathbb P_{n_{3-k}} \left\{\textstyle\dfrac{ AR}{\hat\pi_k(X)} \right\}}_{\textup{V}}.
		\end{align*}
	The decomposition for the other term is similar and thus omitted.
		
	The symbol $\mathscr D_k$ below is used to denote both the data in the $k$th partition (when it appears in a conditional expectation or probability) and their indices (when it appears in the subscript of a summation). Let $b_m=\max\{\vertiii{m_a^\lambda-m_a^{\lambda,\ast}}:a=0,1\}$ and $b_\pi=\|\pi-\pi^\ast\|_2$. Note that $b_\pi=0$ when $\varrho_\pi=\op(1)$ and $b_m=0$ when $\varrho_m=o(1)$, as in these cases, $\pi=\pi^\ast$ and $m_a^\lambda=m_a^{\lambda,\ast}$.

	\begin{proof}[Proof of part \ref{thm:AN-W-rate-fixed-ref-cf}]	This is a direct consequence of Claims \ref{claim:I}--\ref{claim:V}, given that $b_\pi=O(1)$, $b_m=O(1)$, $\nu_n=o(n^{-1/2})$ and $\alpha_n=o(n^{-1/2})$.
	\end{proof}
	
	\begin{proof}[Proof of part \ref{thm:AN-W-AN-fixed-ref-cf}]	Under the assumed conditions, $b_m=b_\pi=0$ and the terms $\textup{I}$ and $\textup{III}$--$\textup{V}$ are of order $\op(n^{-1/2})$. Consequently,
		\begin{equation}
			\sqrt n (n^{-1}(n_1\tau_{\hat\lambda_1}^{\lambda}\hat\psi_{1,1}+n_2\tau_{\hat\lambda_2}^{\lambda}\hat\psi_{1,2}) -\psi_1 ) =  \sqrt{n}(\splavg-\expect) \left[\textstyle\dfrac{A\{Z-m_1^{\lambda,\ast}(X)\}}{\pi^{\ast}(X)}+m_1^{\lambda,\ast}(X)\right] +  \op(1). \label{eq:AN-W-cf-a=1}
		\end{equation} 
		Similar deviations for the case $a=0$ lead to 
		\begin{equation}\label{eq:AN-W-cf-a=0}
			\sqrt n (n^{-1}(n_1\tau_{\hat\lambda_0}^{\lambda}\hat\psi_{0,1}+n_2\tau_{\hat\lambda_2}^{\lambda}\hat\psi_{0,2}) -\psi_0) =  \sqrt{n}(\splavg-\expect) \left[\textstyle\dfrac{(1-A)\{Z-m_0^{\lambda,\ast}(X)\}}{1-\pi^\ast(X)}+m_0^{\lambda,\ast}(X) \right]  +  \op(1).
		\end{equation}
		By combining \eqref{eq:AN-W-cf-a=1} and  \eqref{eq:AN-W-cf-a=0}, the asymptotic normality of $\pt_{\hat\lambda}^\lambda\hat\Delta_{DR}^{\hat\lambda}-\Delta^\lambda$ follows from a central limit theorem and  Slutsky's lemma, with the fact that $\ws$ has a bounded diameter (since $\tdomain$ is assumed to be a bounded interval of $\real$) and thus $Z$, $m_0^{\lambda,\ast}(X)$ and  $m_1^{\lambda,\ast}(X)$ have finite variance.
	\end{proof}

	\begin{claim}\label{claim:I}
		$\textup{I}=\Op(n^{-1/2}\varrho_\pi+n^{-1/2}\varrho_m+\varrho_\pi\varrho_m + n^{-1/2}b_\pi+n^{-1/2}b_m) $.
	\end{claim}
		To prove the claim, let 
		\begin{align*}
			G(A,X,Z) & = \textstyle\dfrac{A\{Z-\tilde m_{1,k}^\lambda(X)\}}{\hat\pi_k(X)}+\tilde m_{1,k}^\lambda(X) - m_1^\lambda(X),\\
			H_1(A,X,Z) & =\frac{AZ\{\pi^\ast(X)-\hat\pi_k(X)\}}{\hat\pi_k(X)\pi^\ast(X)},\\
				H_2(A,X,Z) & =\frac{A\{\hat\pi_k(X)m_1^{\lambda,\ast}(X)-\pi^\ast(X)\tilde m_{1,k}^\lambda(X)\}}{\hat\pi_k(X)\pi^\ast(X)}, \\
					H_3(A,X,Z) & =\tilde m_{1,k}^\lambda(X)-m_1^{\lambda,\ast}(X), \\
					H(A,X,Z) & = H_1(A,X,Z) + H_2(A,X,Z) + H_3(A,X,Z).
		\end{align*}
	Then 
		\begin{align*}
			\expect(\|\textup{\textup{I}}\|_\lambda^2) = &  \frac{1}{n_{3-k}^2}\sum_{i\in\mathscr D_{3-k}}\expect \|H(A_i,X_i,Z_i)-\expect H(A_i,X_i,Z_i)\|_\lambda^2 \\
			& +  \frac{1}{n_{3-k}^2}\sum_{\stackrel{i,j\in\mathscr D_{3-k}}{i\neq j}}\expect \langle H(A_i,X_i,Z_i)-\expect H(A_i,X_i,Z_i) ,H(A_j,X_j,Z_j) -\expect H(A_j,X_j,Z_j)\rangle_\lambda\\
			\equiv & \textup{I}_{1} + \textup{I}_{2}.
		\end{align*}
	
	   For the first term $\textup{I}_{1}$, we further have
	   \begin{align*}
	   	\textup{I}_{1} & \leq \frac{1}{n_{3-k}^2}\sum_{i\in\mathscr D_{3-k}}\expect \|H(A_i,X_i,Z_i)\|_\lambda^2\\
	   	 				& \leq \frac{1}{n_{3-k}^2}\sum_{i\in\mathscr D_{3-k}} \{ \expect\|H_1(A_i,X_i,Z_i)\|_\lambda^2 + \expect \|H_2(A_i,X_i,Z_i)\|_\lambda^2 + \expect \|H_3(A_i,X_i,Z_i)\|_\lambda^2\}.
	   \end{align*}
   		By Assumption \ref{assu:DR-std}\ref{assu:DR-std-pi-bound} and the fact that $\ws$ is bounded, for $(A,X,Z)$ whose index is in $D_{3-k}$, we deduce that
   	   \begin{align*}
   	   	\expect \|H_2(A,X,Z)\|_\lambda^2  \leq & c \expect \|\hat\pi_k(X)m_1^{\lambda,\ast}(X)-\pi^\ast(X)\tilde m_{1,k}^\lambda(X)\|_\lambda^2 \\
   	   	 \leq & c \expect \|\hat\pi_k(X)m_1^{\lambda,\ast}(X)-\pi(X)m_1^{\lambda,\ast}(X)\|_\lambda^2 \\
   	   	      & + c \expect \|\pi(X)m_1^{\lambda,\ast}(X) -\pi(X)\tilde m_{1,k}^\lambda(X)\|_\lambda^2 \\
   	   	      & + c \expect \|\pi(X)\tilde m_{1,k}^\lambda(X)-\pi^\ast(X)\tilde m_{1,k}^\lambda(X)\|_\lambda^2 \\
   	     \leq & c \expect |\hat \pi_k(X)-\pi(X)|^2 + c \expect \|m_1^{\lambda,\ast}(X) -\tilde m_{1,k}^\lambda(X)\|_\lambda^2 + c\expect |\pi(X)-\pi^\ast(X)|^2\\
   	     \leq & c\expect\|\hat\pi_k-\pi\|_2^2 + c\expect\vertiii{m_1^{\lambda,\ast}-\tilde m_{1,k}^\lambda}_\lambda^2 + c\|\pi-\pi^\ast\|_2^2\\
   	       \leq   & c \varrho_{\pi}^2 + c\varrho_m^2 + c b_\pi^2.
   	   \end{align*}
      Similarly, we have $\expect \|H_1(A,X,Z)\|_\lambda^2 \leq c b_\pi^2 + c\varrho_{\pi}^2$ and $\expect \|H_3(A,X,Z)\|_\lambda^2 \leq c\varrho_m^2 + cb_m^2$. These imply that $\expect \|H(A_i,X_i,Z_i)\|_\lambda^2\leq c \varrho_{\pi}^2 + c\varrho_m^2 + c b_\pi^2 + cb_m^2 $ and further
      $$\textup{I}_{1} = O(n^{-1}(\varrho_{\pi}^2+\varrho_m^2+b_\pi^2+b_m^2)).$$
   
       For the term $\textup{I}_{2}$,
   		simple calculation shows that 
   		\begin{align*}
   				& \expect\{ G(A_i,X_i,Z_i)\mid \mathscr D_k,X_i\} = \frac{\{\pi(X_i)-\hat\pi_k(X_i)\}\{m_1^\lambda(X_i)-\tilde m_{1,k}^\lambda(X_i)\}}{\hat\pi_k(X_i)}
   		\end{align*}
   		%
   	and consequently
   	\begin{align*}
   		& \expect[\|\expect\{ G(A_i,X_i,Z_i)\mid \mathscr D_k,X_i\}\|_\lambda\mid \mathscr D_k]  \leq c \|\pi-\hat\pi_k\|_2\vertiii{m_1^\lambda-\tilde m_{1,k}^\lambda}_\lambda
   	\end{align*}
   and
   \begin{align*}
   	\|\expect G(A_i,X_i,Z_i)\|_\lambda \leq c\expect \{|\pi(X_i)-\hat\pi_k(X_i)|\|m_1^\lambda(X_i)-\tilde m_{1,k}^\lambda(X_i)\|_\lambda\}\leq c\varrho_\pi\varrho_m.
   \end{align*}
    where we utilize Assumption \ref{assu:DR-std}\ref{assu:DR-std-pi-bound}.
    Therefore, we deduce that
   \begin{align*}
   	& |\expect \langle H(A_i,X_i,Z_i)-\expect H(A_i,X_i,Z_i) ,H(A_j,X_j,Z_j) -\expect H(A_j,X_j,Z_j)\rangle_\lambda|\\
   	& |\expect \langle G(A_i,X_i,Z_i)-\expect G(A_i,X_i,Z_i) ,G(A_j,X_j,Z_j) -\expect G(A_j,X_j,Z_j)\rangle_\lambda|\\
   	& = |\expect \langle G(A_i,X_i,Z_i),G(A_j,X_j,Z_j)\rangle_\lambda - \langle \expect G(A_i,X_i,Z_i),\expect G(A_j,X_j,Z_j)\rangle_\lambda| \\
   	& \leq \big|\expect \expect \{\langle G(A_i,X_i,Z_i),G(A_j,X_j,Z_j)\rangle_\lambda \mid \mathscr D_k\}\big| + \|\expect G(A_i,X_i,Z_i)\|_\lambda\|\expect G(A_j,X_j,Z_j)\|_\lambda\\
   	& \leq \expect \|\expect  \{G(A_i,X_i,Z_i)\mid \mathscr D_k\}\|_\lambda \|\expect  \{G(A_j,X_j,Z_j)\mid \mathscr D_k\}\|_\lambda + c\varrho_\pi^2\varrho_m^2 \\
   	& \leq c \expect \{\|\pi-\hat\pi_k\|_2^2\vertiii{m_1^\lambda-\tilde m_{1,k}^\lambda}_\lambda^2\} +  c\varrho_\pi^2\varrho_m^2\\
   	& \leq c \varrho_\pi^2\varrho_m^2.
   \end{align*} 
	This result, together with Assumption \ref{assu:cf-sample-size}, implies that $\textup{I}_{2}=O(\varrho_\pi^2\varrho_m^2)$. Combining this with the order for $\textup{I}_{1}$, we show that
	\begin{equation}\label{eq:pf:cf:I}
		\expect(\|\textup{I}\|_\lambda^2)= O(n^{-1}\varrho_\pi^2+n^{-1}\varrho_m^2+\varrho_\pi^2\varrho_m^2 + n^{-1}b_\pi^2+n^{-1}b_m^2),
	\end{equation}
   which implies \Cref{claim:I}.
   
   \begin{claim}\label{claim:II}
   	$\textup{II}=\Op(n^{-1/2})$.
   \end{claim}
	This is a direct consequence of a central limit theorem, with the fact that $\ws$ has a bounded diameter and thus both $Z$ and $m_1^{\lambda,\ast}(X)$ have finite variance.
   
   \begin{claim}\label{claim:III}
   	$\textup{III} = O(\varrho_\pi\varrho_m)$.
   \end{claim}
	By Cauchy--Schwartz inequality, with Assumption \ref{assu:DR-std}\ref{assu:DR-std-pi-bound}, we have
		\begin{align}
			\|\textup{III}\|_\lambda & =\left\|\expect_{n_{3-k}} \expect\left[ \textstyle\dfrac{\{\tilde m_{1,k}^\lambda(X)-m_1^\lambda(X)\}\{\hat{\pi}_k(X)-A\}}{\hat\pi_k(X)} \bigg| \mathscr D_k \right]\right\|_{\lambda} \nonumber \\
			& =\left\|\expect_{n_{3-k}} \left[ \textstyle\dfrac{\{\tilde m_{1,k}^\lambda(X)-m_1^\lambda(X)\}\{\hat{\pi}_k(X)-\pi(X)\}}{\hat\pi_k(X)}  \right]\right\|_{\lambda} \nonumber \\
			& \leq c\left\|\expect_{n_{3-k}} \left[ \{\tilde m_{1,k}^\lambda(X)-m_1^\lambda(X)\}\{\hat{\pi}_k(X)-\pi(X)\}\right]\right\|_{\lambda} \nonumber \\
			& \leq c\expect_{n_{3-k}} \| \{\tilde m_{1,k}^\lambda(X)-m_1^\lambda(X)\}\{\hat{\pi}(X)-\pi(X)\}\|_{\lambda} \nonumber \\
			&  = O(\varrho_\pi\varrho_m). \label{eq:pf:cf:III}
		\end{align}

		\begin{claim}\label{claim:IV}
			$\textup{IV} = \Op\big(n^{-1}+n^{-1/2}\nu_n+n^{-1/2}\alpha_n+n^{-1/2}\varrho_\pi +\varrho_\pi\nu_n+\varrho_\pi\alpha_n\big).$
		\end{claim}
		Consider 
		\begin{align*}
			\|\textup{IV}\|_\lambda^2 = & \underbrace{n^{-2}_{3-k}\sum_{i\in \mathscr D_k} \bigg\|\left\{1-\textstyle\dfrac{A_i}{\hat\pi_k(X_i)}\right\}D_{1,k}(X_i)\bigg\|_\lambda^2}_{\textup{IV}_1} \\
			& + \underbrace{n^{-2}_{3-k}\sum_{\stackrel{i,j\in \mathscr D_k}{i\neq j}}\langle \left\{1-\textstyle\dfrac{A_i}{\hat\pi_k(X_i)}\right\}D_{1,k}(X_i),\left\{1-\textstyle\dfrac{A_j}{\hat\pi_k(X_j)}\right\}D_{1,k}(X_j)\rangle_\lambda}_{\textup{IV}_2}.
		\end{align*}
		It is seen that $\textup{IV}_1\leq c n^{-2}_{3-k} \sum_{i\in \mathscr D_k}\|D_{1,k}(X_i)\|_\lambda^2$,  where Assumption \ref{assu:DR-std}\ref{assu:DR-std-pi-bound} is utilized. For any $\delta>0$,
		\begin{align*}
			& \pr\big(n^{-1}_{3-k} \sum_{i\in \mathscr D_k}\|D_{1,k}(X_i)\|_\lambda^2 \leq \delta^{-1}\vertiii{\pt_{\hat\lambda_k}^{\lambda}\hat m_{1,k}^{\hat\lambda_k}-\tilde m_{1,k}^\lambda}_\lambda^2 \mid \mathscr D_k\big) \\
			& \leq \delta \frac{ n^{-1}_{3-k} \sum_{i\in \mathscr D_k}\expect\{\|D_{1,k}(X_i)\|_\lambda^2\mid \mathscr D_k\}}{\vertiii{\pt_{\hat\lambda_k}^{\lambda}\hat m_{1,k}^{\hat\lambda_k}-\tilde m_{1,k}^\lambda}_\lambda^2}\\
			& = \delta,
		\end{align*}
	which then implies that $\pr\big(n^{-1}_{3-k} \sum_{i\in \mathscr D_k}\|D_{1,k}(X_i)\|_\lambda^2 \leq \delta^{-1}\vertiii{\pt_{\hat\lambda_k}^{\lambda}\hat m_{1,k}^{\hat\lambda_k}-\tilde m_{1,k}^\lambda}_\lambda^2\big)\leq \delta$ for arbitrary $\delta>0$, or equivalently, $n^{-1}_{3-k} \sum_{i\in \mathscr D_k}\|D_{1,k}(X_i)\|_\lambda^2 \leq \delta^{-1}\vertiii{\pt_{\hat\lambda_k}^{\lambda}\hat m_{1,k}^{\hat\lambda_k}-\tilde m_{1,k}^\lambda}_\lambda^2 = \Op\big(\vertiii{\pt_{\hat\lambda_k}^{\lambda}\hat m_{1,k}^{\hat\lambda_k}-\tilde m_{1,k}^\lambda}_\lambda^2\big)$. With Assumptions \ref{assu:DR-lambda},  \ref{assu:DR-D-rate-fixed-ref-cf} and \ref{assu:cf-sample-size}, we conclude that $\|\textup{IV}_1\|_\lambda^2=\Op(n^{-2}+n^{-1}\nu_n^2+n^{-1}\alpha^2_n)$. Similarly, we can deduce that $\|\textup{IV}_2\|_\lambda^2=\Op(\|\hat\pi_k-\pi\|_2^2\vertiii{D_{1,k}}_\lambda^2)=\Op\big(n^{-1}\varrho_\pi^2 +\varrho_\pi^2\nu_n^2+\varrho_\pi^2\alpha_n^2\big)$. Consequently, we obtain 
	\begin{equation}\label{eq:pf:cf:IV}
		\textup{IV} = \Op\big(n^{-1}+n^{-1/2}\nu_n+n^{-1/2}\alpha_n+n^{-1/2}\varrho_\pi +\varrho_\pi\nu_n+\varrho_\pi\alpha_n\big).
	\end{equation}
		
	\begin{claim}\label{claim:V}
		$V=\Op(\alpha_n+\nu_n).$
	\end{claim}	
		For a proof, we observe that
		\begin{align*}
			\mathbb P_{n_{3-k}} \left[ \textstyle\dfrac{ AR}{\hat\pi_k(X)}\right] & = \mathbb P_{n_{3-k}} \left[ \textstyle\dfrac{ AR}{\pi(X)}\right] + \splavg \left[ \textstyle\dfrac{ AR}{\hat\pi_k(X)}-\textstyle\dfrac{ AR}{\pi(X)}\right],
		\end{align*}
		where the second term is dominated by the first one. Moreover, 
		\begin{equation*}
			\mathbb P_{n_{3-k}} \left[ \textstyle\dfrac{ AR}{\pi(X)}\right]  = \mathbb P_{n_{3-k}} \left[ \textstyle\dfrac{ AU}{\pi(X)}\right] + \mathbb P_{n_{3-k}} \left[ \textstyle\dfrac{ AV}{\pi(X)}\right],
		\end{equation*}
		where $U=\pt_{\hat\lambda_k}^\lambda \rlog_{\hat \lambda_k} \Y-\rlog_\lambda \Y=0$ and $V=\pt_{\hat\lambda_k}^\lambda\rlog_{\hat\lambda_k} \eY-\pt_{\hat\lambda_k}^\lambda\rlog_{\hat\lambda_k} \Y=\rlog_\lambda \eY-\rlog_\lambda Y$. By Lemma \ref{lem:eY-Y} and the assumptions on $\alpha_n$ and $\nu_n$, we deduce that
		\begin{equation}\label{eq:pf:cf:V}
			V=\Op(\alpha_n+\nu_n).
		\end{equation}
		 
	\begin{remark}\label{rem:S2}
	    As in Remark \ref{rem:S1}, when the assumption on boundedness of $\tdomain$ is dropped,  for Claim \ref{claim:II} to hold, we require the second moment condition that $\expect W_2^2(Y,y)<\infty$ for some $y\in\ws$ and $\vertiii{m_a^{\lambda,\ast}}_\lambda<\infty$. In addition, Claim \ref{claim:I} requires the additional condition that $\expect\vertiii{\tilde m_{a,k}^\lambda}_\lambda^2<\infty$ and  $\vertiii{m_a^{\lambda,\ast}}_\lambda<\infty$, for $a=0,1$ and all $k$. This is because, under the new condition, we bound the term
	    \begin{align*}
	        \expect \|\pi(X)\tilde m_{1,k}^\lambda(X)-\pi^\ast(X)\tilde m_{1,k}^\lambda(X)\|_\lambda^2 & = \expect  \expect \{\|\pi(X)\tilde m_{1,k}^\lambda(X)-\pi^\ast(X)\tilde m_{1,k}^\lambda(X)\|_\lambda^2 \mid \mathscr D_k\}\\
	        & = \expect\{\|\pi-\pi^\ast\|_2^2 \vertiii{\tilde m_{1,k}^\lambda}_\lambda^2\}\\
	        & \leq  \|\pi-\pi^\ast\|_2^2 \expect \vertiii{\tilde m_{1,k}^\lambda}_\lambda^2 = O(b_\pi^2)
	    \end{align*}
	    when we derive the order for $\expect \|H_2(A,X,Z)\|_\lambda^2$ and $X$ is not in $\mathscr D_k$. A similar argument applies to $\expect \|\hat\pi_k(X)m_1^{\lambda,\ast}(X)-\pi(X)m_1^{\lambda,\ast}(X)\|_\lambda^2$ and $\expect \|H_1(A,X,Z)\|_\lambda^2$.
	\end{remark}	
	
	\section{Technical Lemmas}

	\begin{lemma}\label{lem:dist-norm-W}
		For $G_1,G_2\in\ws$ and a probability distribution $\lambda$, we have $\|\rlog_\lambda G_1-\rlog_\lambda G_2\|_\lambda=W_2(G_1,G_2)$. 
	\end{lemma}
	\begin{proof}[Proof of \Cref{lem:dist-norm-W}]
		The claims follows from the following observations: $\|\rlog_\lambda G_1-\rlog_\lambda G_2\|_\lambda^2=\int |G_1^{-1}\circ \lambda-G_2^{-1}\circ \lambda|^2\diffop \lambda=\int_0^1 |G^{-1}_1(t)-G^{-1}_2|^2\diffop t=W^2_2(G_1,G_2)$, where the last equality is due to Theorem 2.18 of \cite{villani2003topics}.
	\end{proof}

	\begin{lemma}\label{lem:expect-Y-inverse-inveser-mu}
		For a random element $W$ on $\ws$, $\expect W^{-1}=(\FM W)^{-1}$ and $\expect \rlog_\lambda W=\rlog_\lambda \FM W$ for any probability distribution $\lambda$.
	\end{lemma}
	\begin{proof}[Proof of \Cref{lem:expect-Y-inverse-inveser-mu}]
		The first assertion is a direct consequence of the isometry between $\ws$  and the collection of quantile functions, viewed as a subspace of space $L^2(\tdomain)$ of squared integrable functions endowed with the $L^2$ distance $\|f-g\|\define \sqrt{\int_{\tdomain}|f(x)-g(x)|^2_2\diffop x}$ \citep[Theorem 2.18,][]{villani2003topics}.
		
		For the second assertion, $\expect\rlog_\lambda W=\expect (W^{-1}\circ \lambda)=(\expect W^{-1})\circ y=(\FM W)^{-1}\circ \lambda=\rlog_\lambda \FM W$, where the third equality is due to the first assertion.
	\end{proof}

	\begin{lemma}\label{lem:growth-condition-W}
		For the  Fr\'echet function  $F(\cdot)=\expect W_2^2(y,\Y)$ of a random element $\Y$ on a Wasserstein space $\ws$, we have $F(y)-F(\mu)=W_2^2(y,\mu)$ for all $y\in\ws$, where $\mu$ is the Fr\'echet mean of $\Y$.
	\end{lemma}
	\begin{proof}[Proof of \Cref{lem:growth-condition-W}]
		Let $\langle g,h\rangle=\int_{\tdomain}g(t)h(t)\diffop t$ for two functions $g$ and $h$. 
		We first observe that
		\begin{align*}
			F(y)-F(\mu) & = \expect W_2^2(y,\Y)-\expect d^2(\mu,\Y) \\
			& = \expect\big( \langle y^{-1}-Y^{-1},y^{-1}-Y^{-1}\rangle - \langle \mu^{-1}-Y^{-1},\mu^{-1}-Y^{-1}\rangle \big) \\
			& = \langle y^{-1}-\mu^{-1},y^{-1}-\mu^{-1}\rangle - 2\expect\langle y^{-1}-\mu^{-1},\mu^{-1}-Y^{-1}\rangle\\
			& = W_2^2(y,\mu) - 2\langle y^{-1}-\mu^{-1},\mu^{-1}-\expect Y^{-1}\rangle \\
			& = W_2^2(y,\mu)
		\end{align*}
		where the second equality is due to the isometry between $\ws$  and the collection of quantile functions \citep[Theorem 2.18,][]{villani2003topics}, and the last equality is due to $\expect Y^{-1}=\mu^{-1}$ shown in \Cref{lem:expect-Y-inverse-inveser-mu}.
	\end{proof}

	\begin{lemma}\label{lem:consistency-mu-hat}
		Suppose that $\hat\mu$ is the empirical Fr\'echet mean of $\eY_1,\ldots,\eY_n$, and  $\tilde\mu$ is the empirical Fr\'echet mean of $Y_1,\ldots,Y_n$ residing on $\manifold$. Then we have $W_2^2(\tilde\mu,\mu)=\op(1)$, and under additional Assumption \ref{assu:eY} we have $W_2^2(\hat\mu,\mu)=\op(1)$.
	\end{lemma}
	\begin{proof}[Proof of \Cref{lem:consistency-mu-hat}]
		As in the proof of Lemma \ref{lem:growth-condition-W}, according to the isometry between $\ws$  and the collection of quantile functions \citep[Theorem 2.18,][]{villani2003topics}, $W^2_2(\tilde\mu,\mu)=\op(1)$ if and only if $\|\tilde{\mu}^{-1}-\mu^{-1}\|_2^2=\op(1)$, where $\|g\|_2^2=\int_{\tdomain}|g(t)|^2\diffop t$ for any measurable function $g$. By \Cref{lem:expect-Y-inverse-inveser-mu}, $\tilde\mu^{-1}=n^{-1}\sum_{i=1}^n \Y_i^{-1}$ and $\expect \Y_i^{-1}=\mu^{-1}$. Then  $\|\tilde{\mu}^{-1}-\mu^{-1}\|_2^2=\op(1)$ follows from the weak law of large numbers.
		
		To prove $W^2_2(\hat\mu,\mu)=\op(1)$, we apply \Cref{lem:expect-Y-inverse-inveser-mu} to the uniform distribution on the discrete set $\{\eY_1,\ldots,\eY_n\}$ and its Fr\'echet mean $\hat\mu$, and deduce $\hat\mu^{-1}=(n^{-1}\sum_{i=1}^n \eY_i^{-1})$. Similarly, $\tilde\mu^{-1}=(n^{-1}\sum_{i=1}^n \Y_i^{-1})$. 
		Then 
		\begin{align*}
			W_2(\hat\mu,\tilde\mu) &= \|\hat\mu^{-1}-\tilde\mu^{-1}\|_2 = \bigg\|n^{-1}\sum_{i=1}^n\eY_i^{-1}-n^{-1}\sum_{i=1}^n \Y_i^{-1}\bigg\|_2\\
			& \leq n^{-1}\sum_{i=1}^n\|\eY_i^{-1}-\Y^{-1}_i\|_2 =  n^{-1}\sum_{i=1}^n W_2(\eY_i,\Y_i)=\Op(\alpha_n)=\op(1),
		\end{align*}
		where the third equality is due to Assumption \ref{assu:eY} and Markov's inequality. Then $W_2(\hat\mu,\mu)\leq W_2(\hat\mu,\tilde\mu)+W_2(\tilde\mu,\mu)=\op(1)$.
	\end{proof}
	
	\begin{lemma}\label{lem:Lmuhat-mu-eY} 
		If $\hat\lambda$ is continuous, then  under Assumption  \ref{assu:eY},
		\begin{equation*}
			\dfrac{1}{n}\sum_{i=1}^n \|\pt_{\hat\lambda}^\lambda \rlog_{\hat \lambda} \eY_i-\rlog_\lambda \Y_i\|_\lambda^2=\Op\big(\alpha_n^2 +\nu_n^2\big).
		\end{equation*}
	\end{lemma}
	\begin{proof}[Proof of  \Cref{lem:Lmuhat-mu-eY}]
		Given the following observation
		\begin{align}
			\dfrac{1}{n}\sum_{i=1}^n \|\pt_{\hat\lambda}^\lambda \rlog_{\hat \lambda} \eY_i-\rlog_\lambda \Y_i\|_\lambda^2 & \leq \dfrac{2}{n}\sum_{i=1}^n \|\pt_{\hat\lambda}^\lambda \rlog_{\hat \lambda} \Y_i-\rlog_\lambda \Y_i\|_\lambda^2 + \dfrac{2}{n}\sum_{i=1}^n \|\pt_{\hat\lambda}^\lambda \rlog_{\hat \lambda} \eY_i-\pt_{\hat\lambda}^\lambda\rlog_{\hat\lambda} \Y_i\|_\lambda^2,\nonumber
		\end{align}
		 the conclusion  is a direct consequence of \Cref{lem:eY-Y} and the fact that $\pt_{\hat\lambda}^\lambda\wwlog_{\hat\lambda} Y_i-\wwlog_{\lambda} Y_i=0$ {when $\hat\lambda$ is continuous (so that $\hat\lambda\circ\hat\lambda^{-1}=\idf$)}. 
		
	\end{proof}
	
	\begin{lemma}\label{lem:eY-Y} 
		Under Assumption  \ref{assu:eY}, for a fixed $\epsilon>0$, we have
		\begin{equation*}
			\sup_{\lambda}\dfrac{1}{n}\sum_{i=1}^n \|\rlog_\mp \eY_i-\rlog_\mp \Y_i\|_\mp^2=\Op\big(\alpha_n^2 +\nu_n^2\big).
		\end{equation*}
	\end{lemma}
	\begin{proof}[Proof of Lemma \ref{lem:eY-Y}]
		This directly follows from Lemma \ref{lem:dist-norm-W} and Assumption \ref{assu:eY}.
	\end{proof}

	\begin{lemma}
		\label{lem:mu-rate-noisy-W}
		Suppose that $\hat\mu$ is the empirical Fr\'echet mean of $\eY_1,\ldots,\eY_n\in\ws$, and  $\tilde\mu$ is the empirical Fr\'echet mean of $Y_1,\ldots,Y_n\in\ws$. Then,  we have $W^2_2(\tilde{\mu},\mu)=\Op(n^{-1})$.  
		Suppose further that Assumption \ref{assu:eY} holds, then we have $W_2^2(\hat\mu,\tilde{\mu})=\Op(\alpha_n^2)$ and $W_2^2(\hat\mu,\mu)=\Op(n^{-1}+\alpha_n^2)$.
	\end{lemma}
	
	\begin{proof}[Proof of  \Cref{lem:mu-rate-noisy-W}]
		We apply the general theory from \cite{schotz2019convergence}. According to the discussion in Section 3 of \cite{schotz2019convergence}, the weak quadruple condition holds for a Wasserstein space $\ws$. The moment condition is met given the boundedness of $\ws$, while the growth condition is verified in Lemma \ref{lem:growth-condition-W}. 
		In light of \Cref{lem:consistency-mu-hat} and according to \citet{schotz2019convergence}, it is sufficient to verify the entropy condition  in a neighborhood of $\mu$. Such entropy condition holds as $\sup_{\omega\in\ws} \log N(\delta\epsilon,B_\delta(\omega),d)\leq K\epsilon^{-1}$ shown in the proof of Proposition 1 of \cite{Petersen2019}.
		
		Now we modify the  proof for Theorem 1  of \cite{schotz2019convergence} to the case that  only noisy surrogates $\eY_1,\ldots,\eY_n$ are observable. 
		Let $\tilde F_n(y)=n^{-1}\sum_{i=1}^n W_2^2(y,\Y_i)$ and $F_n(y)=n^{-1}\sum_{i=1}^n W_2^2(y, \eY_i)$. Define
		\begin{align*}
			\tilde\Xi_n(\delta) & \define\sup_{y\in \manifold:W_2(y,\mu)<\delta} \tilde F(y)-\tilde F(\mu)- \tilde F_n(y)+ \tilde F_n(\mu),\\
			\Xi_n(\delta) & \define\sup_{y\in \manifold:W_2(y,\mu)<\delta} \tilde F(y)-\tilde F(\mu)- F_n(y)+ F_n(\mu).
		\end{align*}
		According to Lemma 2 of \cite{schotz2019convergence}, we just need to show that $\expect\{\Xi_n^2(\delta)\}\leq c_1(n^{-1}+\alpha_n^2)\delta^2$ for some constant $c_1>0$ and all $\delta >0$. First, we observe that $ \Xi_n(\delta)\leq \tilde \Xi_n(\delta)+ G_n(\delta)$, where 
		\begin{align*}
			G_n(\delta) & = \sup_{y\in \manifold:W_2(y,\mu)<\delta}\dfrac{1}{n}\sum_{i=1}^n \{W_2^2(y,\eY_i)-W_2^2(y,\Y_i)-W_2^2(\mu,\eY_i)+W_2^2(\mu,\Y_i)\}\\
			& \leq  \sup_{y\in \manifold:W_2(y,\mu)<\delta}\dfrac{1}{n}\sum_{i=1}^n W_2(y,\mu)d(\Y_i,\eY_i) \leq  \dfrac{\delta}{n}\sum_{i=1}^n W_2(\Y_i,\eY_i),
		\end{align*}
		where the first inequality is due to the Weak Quadruple condition; see Section 3.2.3 of \cite{schotz2019convergence} for details.
		Thus, according to Assumption \ref{assu:eY}, 
		\begin{align*}
			\expect\{G_n^2(\delta)\} & \leq \delta^2\expect \left\{\dfrac{1}{n}\sum_{i=1}^n W_2(\Y_i,\eY_i)\right\}^2 \leq  \dfrac{\delta^2}{n}\sum_{i=1}^n \expect W_2^2(\Y_i,\eY_i)  \leq \delta^2\alpha_n^2.
		\end{align*}
		According to Lemma 3 of \cite{schotz2019convergence}, we have $\expect\{\tilde \Xi_n^2(\delta)\}\leq c_2\delta^2/n$ for some constant $c_2>0$. Consequently, 
		\begin{align*}
			\expect\{\Xi_n^2(\delta)\} & \leq 2 \expect\{\tilde \Xi_n^2(\delta)\} + 2\expect\{G_n^2(\delta)\} \leq 2(c_2+1)(n^{-1}+\alpha_n^2)\delta^2.
		\end{align*}
		The proof for $W_2^2(\hat\mu,\mu)=\Op(n^{-1}+\alpha_n^2)$ is completed by setting $c_1=2(c_2+1)$. The result for  $W_2^2(\hat\mu,\tilde\mu)$ follows from the same line of argument, conditional on $Y_1,\ldots,Y_n$.
	\end{proof}
	
	\section{Function-valued Empirical Processes}

	In the proof of \Cref{thm:AN-W-fixed-ref} we came cross the problem of finding the convergence rate of a random quantity in the form $n^{-1/2}\sum_{i=1}^n h(\xi_i) \{\hat g(\xi_i)-g_0(\xi_i)\}$ for fixed functions $h,g_0$ and independent observations $\xi_1,\ldots,\xi_n$ such that $\expect h(\xi_i)g(\xi_i)=0$ for any $g$; here  $\hat g$ is an estimate for $g_0$ based on $\xi_1,\ldots,\xi_n$ (and potentially on some other fixed quantities). For example, $\xi_i=A_i$, $h(\xi_i)=\pi(X_i)-A_i$, $g_0=m_a$, and $\hat g=\tilde m_a$ in the proof of \eqref{eq:IV}, where we hold $X_i$ fixed. Such problem is also  encountered in other scenarios of statistical research, for instance, in the proof of Theorem 9 of \cite{Mammen1997}. Unique in our context is that the function $\hat g$ may take values in a space of functions, e.g., in the case of $\ws$, so that the classic results in \cite{VandeGeer1990} do not apply directly and need to be extended. Instead of function spaces, below we consider more generally a separable Hilbert space that includes some function spaces as special examples.
	
	To set the stage, let $\xdomain$ be a compact space and $\nu$ a positive measure on it. Suppose that $\hilbert$ is a separable Hilbert space, and $\mathcal{F}$ is a class of functions that maps $\xdomain$ into $\hilbert$. An immediate example of $\hilbert$ is the space  $L^2([0,1])$ of real-valued squared integrable function defined on $[0,1]$. Consider an $\hilbert$-valued random process $S_n$ indexed by $\mathcal F$ of the form 
	$$S_n(g)=\frac{1}{\sqrt n}\sum_{i=1}^nV_i(g),$$
	where $V_1,\ldots,V_n$ are independent $\hilbert$-valued centered random process defined on $\mathcal F$. For instance, in the above context, $V_i(g)= h(\xi_i)\{g(\xi_i)-g_0(\xi_i)\}$ for each $g\in\mathcal F$.  
	
	Let $\eta$ be a pseudo-distance on $\mathcal{F}$ of the form $\eta^2 (f,g)=n^{-1}\sum_{i=1}^n\eta_i^2(f,g)$, where $\eta_1,\ldots,\eta_n$ are also pseudo-distance on $\mathcal{F}$, e.g., $\eta_i(g_1,g_2)=\|g_1(\xi_i)-g_2(\xi_i)\|_{\hilbert}$ with $\|\cdot\|_{\hilbert}$ denoting the norm on the Hilbert space $\hilbert$. Suppose $S_n(g_0)=0$ for $g_0\in\mathcal{F}$ and $\|V_i(g_1)-V_i(g_2)\|_{\hilbert}\leq M_i \eta_i(g_1,g_2)$, where $M_1,\ldots,M_n$ are uniformly subgaussian random variables, i.e, \begin{equation}\label{eq:subgaussian-hilbert}
		\sup_{i}\expect \exp\{\beta^2 M_i^2\}\leq \Gamma<0
	\end{equation}
	for some absolute constants $\beta,\Gamma >0$. Finally, let $\mathscr K(\delta,r)$ be the local entropy of a ball $ B(r;g_0)\subset\mathcal F$ with respect to the pseudo-distance $\eta$. Without loss of generality we assume $\mathscr K$ is continuous in $\delta$; otherwise, we just define it to be a continuous function that upper bounds the local entropy. 
	\begin{theorem}\label{thm:hilbertian-empirical}
		Suppose \eqref{eq:subgaussian-hilbert} holds and $\mathscr K(\delta,r)\leq K\delta^{-2\zeta}$ for some $K>0$ and  $\zeta\in(0,1)$. Then, there exist constants $b_0$ and $c$, depending only on $\beta$ and $\Gamma$,  such that for any $b\geq b_0$, we have
		\begin{equation}\label{eq:hilbert-concentration}
			\pr\left(\sup_{g\in  B(r;g_0)}\frac{\|S_n(g)\|_{\hilbert}}{\{\eta(g,g_0)\}^{1-\zeta}}\geq b\sqrt{K}\right)\leq \exp\left(-\textstyle\frac{cb^2K}{r^{2\zeta}}\right)
		\end{equation}
		for all $n\geq 1$.
	\end{theorem}
	\begin{proof}
		The subgaussian assumption implies that Eq (3.10) of \cite{Kuelbs1978} is valid for $S_n(g_1-g_2)$ for all $g_1,g_2\in\mathcal F$. Then, \eqref{eq:hilbert-concentration} follows from the argument that leads to Lemma 3.5 of \cite{VandeGeer1990}.
	\end{proof}

\end{document}